\documentclass[11pt,a4paper,reqno]{proc-l}

\usepackage[utf8]{inputenc}
\usepackage[T1]{fontenc}

\usepackage[toc, page]{appendix} 

\usepackage[numbers,sort&compress]{natbib}

\usepackage{lmodern}
\usepackage{textcomp}
\usepackage{stmaryrd}
\usepackage{pifont}

\usepackage[english]{babel}

\usepackage[margin=30mm]{geometry}
\usepackage{setspace}
\usepackage{parskip}

\usepackage{enumerate}
\usepackage[shortlabels]{enumitem}

\usepackage{fancyhdr}	              
\usepackage[bottom]{footmisc}    
\usepackage{perpage}    
\usepackage{afterpage}

\usepackage{lscape}
\usepackage{xypic}
\usepackage{color}
\usepackage[
color,matrix,arrow]{xy}
\usepackage{setspace}
\usepackage{tikz-cd}

\usepackage{graphicx}
\usepackage{pict2e}
\usepackage{longtable}
\usepackage{wrapfig}
\usepackage{rotating}
\usepackage[normalem]{ulem}
\usepackage{graphics}								
\usepackage{float} 									
\usepackage[export]{adjustbox}                      
\usepackage{caption}                                
\usepackage{wrapfig}                                
\usepackage{subfigure}
\usepackage{pdfpages}    
\usepackage{makecell}
\usepackage{boldline}

\usepackage{tikz}       
\usepackage{tikz-cd}
\usepackage[new]{old-arrows}

\usepackage{amsmath}
\usepackage{amsfonts}
\usepackage{amssymb}
\usepackage{amsthm}
\usepackage{mathtools}
\usepackage{mathrsfs}
\usepackage{commath}           

\usepackage{gensymb}           
\usepackage{dirtytalk}         
\usepackage{siunitx}            
\usepackage{tensor}             

\usepackage{physics}          
\usepackage{simpler-wick}
\usepackage{slashed}
\usepackage{pythonhighlight}
\usepackage{dynkin-diagrams}
\pgfkeys{/Dynkin diagram,
edge length=.5cm,
fold radius=1cm,
indefinite edge/.style={
draw=black,
fill=yellow,
thin,
densely dashed}}

\usepackage{url}
\usepackage{hyperref}
\hypersetup{colorlinks, citecolor=red, filecolor=green,linkcolor=black,urlcolor=pink}

\usepackage{capt-of}
\usepackage{grffile}
\usepackage[normalem]{ulem}
\usepackage{textcomp}
\usepackage{capt-of}
\usepackage{bbm}
\usepackage{grffile}

\NewDocumentCommand \numpar { o o }{
  \par \subsubsection{\IfNoValueF{#2}{#2}}\IfNoValueF{#1}{\label{#1}}
}

\MakePerPage{footnote}

\DeclareCaptionLabelFormat{r-parens}{#2)} 
{}

\newtheorem{theorem}{Theorem}[section]
\newtheorem{lemma}[theorem]{Lemma}
\newtheorem{cor}[theorem]{Corollary}
\newtheorem{prop}[theorem]{Proposition}
\theoremstyle{definition}

\theoremstyle{remark}

\newcommand{\Wa}{\mathcal{W}}

\renewcommand{\cD}{\mathcal{D}}
\newcommand{\cE}{\mathcal{E}}

\newcommand*{\bid}[1]{\ensuremath{(\!( #1 )\!)}}

\newcommand{\fa}{\ensuremath{\mathfrak{a}}}
\newcommand{\fb}{\ensuremath{\mathfrak{b}}}

\newcommand{\fg}{\ensuremath{\mathfrak{g}}}

\newcommand{\fk}{\ensuremath{\mathfrak{k}}}

\newcommand{\fm}{M}

\newcommand{\fn}{\ensuremath{\mathfrak{n}}}

\newcommand{\fp}{\ensuremath{\mathfrak{p}}}

\newcommand{\fu}{\ensuremath{\mathfrak{u}}}

\DeclareMathOperator{\ad}{ad}

\DeclareMathOperator{\Ext}{Ext}
\DeclareMathOperator{\Hom}{Hom}

\DeclareMathOperator{\Spec}{Spec}
\DeclareMathOperator{\Aut}{Aut}
\DeclareMathOperator{\Vect}{Vect}
\DeclareMathOperator{\Sym}{Sym}

\DeclareMathOperator{\Proj}{Proj}

\DeclareMathOperator{\Der}{Der}
\DeclareMathOperator{\Rep}{Rep}
\DeclareMathOperator{\Ob}{Ob}

\DeclareMathOperator{\Symc}{\overline{Sym}}

\newcommand{\Z}{\ensuremath{\mathbb{Z}}}

\newcommand{\C}{\ensuremath{\mathbb{C}}}

\NewDocumentCommand \deq { o m }{
\begin{equation}
#2
\IfNoValueF{#1}{\label{#1}}
\end{equation}
}

\newcommand{\beq}{\begin{equation}}
\newcommand{\beqn}{\begin{equation*}}
\newcommand{\eeq}{\end{equation}}
\newcommand{\eeqn}{\end{equation*}}
\newcommand{\beqa}{\begin{eqnarray}}
\newcommand{\beqan}{\begin{eqnarray*}}
\newcommand{\eeqa}{\end{eqnarray}}
\newcommand{\eeqan}{\end{eqnarray*}}
\newcommand{\bdm}{\begin{displaymath}}
\newcommand{\edm}{\end{displaymath}}

\newcommand{\ot}{\otimes}

\newcommand{\N}{\ensuremath{\mathcal{N}}}

\def\fb{\ensuremath{\mathfrak{b}}}

\NewDocumentCommand \ft { o }{%
\ensuremath{\mathfrak{t}%
\IfNoValueF{#1}{^{({#1})}}%
}}

\def\O{\ensuremath{\mathcal{O}}}

\def\so{\mathfrak{so}}
\def\su{\mathfrak{su}}

\def\gl{\mathfrak{gl}}

\def\osp{\mathfrak{osp}}

\newcommand{\bu}{{\bullet}}
\newcommand{\dCE}{d_{CE}}

\def\cat#1{\ensuremath{\text{\textsf{#1}}}}
\DeclareMathOperator{\GrVect}{\cat{GrVect}}
\DeclareMathOperator{\Vectcat}{\cat{Vect}}
\DeclareMathOperator{\Ch}{\cat{Ch}}

\def\ie{{\it i.e.}}

\def\eg{{\it e.g.}}

\def\apriori{{\it a priori}}

\def\sm{s_0}
\def\sF{s_1}

\def\nn{\nonumber}
\def\*{\partial}

\def\EWeight#1#2#3#4{\bigl({}^{\mathstrut}_{#1\mathstrut}{}_{#2\mathstrut}^{#4\mathstrut}{}_{#3\mathstrut}^{\mathstrut}\bigr)}

\renewcommand*\circled[1]{\tikz[baseline=(char.base)]{
            \node[shape=circle,draw,inner sep=1pt] (char) {\sf#1};}}

\makeatletter
\renewcommand*\env@matrix[1][\arraystretch]{%
  \edef\arraystretch{#1}%
  \hskip -\arraycolsep
  \let\@ifnextchar\new@ifnextchar
  \array{*\c@MaxMatrixCols c}}
\makeatother

\title{Canonical supermultiplets and their Koszul duals}

\begin{document}

\frenchspacing

\null\vskip-24pt

\vskip48pt

\maketitle

\vskip24pt

\begin{center}
  
  Martin Cederwall\textsuperscript{1,2}, Simon Jonsson\textsuperscript{3}, Jakob Palmkvist\textsuperscript{4} and Ingmar Saberi\textsuperscript{5} 
  \par \bigskip

\vskip12pt

  \textsuperscript{1}\footnotesize{{ Dept. of Physics, Chalmers Univ. of Technology, SE-412 96, Gothenburg, Sweden}} \par
  \textsuperscript{2} {NORDITA, Hannes Alfv\'ens v\"ag 12, SE-106 91 Stockholm, Sweden} \par 
  \textsuperscript{3} {Dept. of Physics, Astronomy, and Mathematics, 
Univ. of Hertfordshire,\\ Hatfield, Hertfordshire. AL10 9AB United Kingdom} \par
\textsuperscript{4} {School of Science and Technology, \"Orebro Univ., SE-701 82 \"Orebro, Sweden} \par 
\textsuperscript{5} {Ludwig-Maximilians-Universit\"at M\"unchen, Theresienstra\ss e 37, 80333 M\"unchen, Germany} \par \bigskip
\end{center}

\thispagestyle{empty}

\vskip32pt

\begin{center}
    \parbox{.85\hsize}{\small\underbar{Abstract:} The pure spinor superfield formalism reveals that, in any dimension and with any amount of supersymmetry, one particular supermultiplet is distinguished from all others. This ``canonical supermultiplet'' is equipped with an additional structure that is not apparent in any component-field formalism: a (homotopy) commutative algebra structure on the space of fields. The structure is physically relevant in several ways; it is responsible for the interactions in ten-dimensional super Yang--Mills theory, as well as crucial to any first-quantised interpretation. We study the $L_\infty$ algebra structure that is Koszul dual to this commutative algebra, both in general and in numerous examples, and prove that it is equivalent to the subalgebra of the Koszul dual to functions on the space of generalised pure spinors in internal degree greater than or equal to three. In many examples, the latter is the positive part of a Borcherds--Kac--Moody superalgebra. Using this result, we can interpret the canonical multiplet as the homotopy fiber of the map from generalised pure spinor space to its derived replacement. This generalises and extends work of Movshev--Schwarz and G\'alvez--Gorbounov--Shaikh--Tonks in the same spirit. We also comment on some issues with physical interpretations of the canonical multiplet, which are illustrated by an example related to the complex Cayley plane, and on possible extensions of our construction, which appear relevant in an example with symmetry type $G_2 \times A_1$.}
\end{center}

\vfill
\rule{\textwidth}{0.4pt}
{\tiny email: martin.cederwall@chalmers.se, d.jonsson@herts.ac.uk, jakob.palmkvist@oru.se, i.saberi@physik.uni-muenchen.de}

\newpage
\tableofcontents

\thispagestyle{empty}
\newpage

\pagenumbering{arabic}
\setcounter{page}{3}

\section{Introduction}
\vskip8pt
The study of supersymmetric field theories has been a catalyst for many new and interesting developments in mathematical physics over the last fifty years or so. Many of these developments can now be understood as pointing at broader themes that are not confined to the context of supersymmetric field theory. For example, the desire to understand supersymmetric extensions of the Poincar\'e algebra geometrically, as structure-preserving transformations of an appropriate ``super-spacetime'', was a major impetus behind much work in supergeometry. But the role of supergeometry and other supermathematics in physics is apparent in any theory with fermionic degrees of freedom, whether supersymmetric or not. 
Another issue which is in principle independent of supersymmetry, but which is unavoidable in the study of maximally supersymmetric theories and especially supergravity theories, is the presence of on-shell symmetries: a (super)symmetry acts on the critical locus of the action functional, but there is no obvious way of extending the transformations defining this symmetry to the space of all (off-shell) field configurations, and thus no obvious way of gauging the symmetry with standard techniques. Precisely this issue led to the development of the Batalin--Vilkovisky formalism, which deals with on-shell symmetries by replacing the space of fields by (a specific model of) the derived critical locus and giving a homotopy action of the on-shell symmetry on this space.
The BV formalism, in turn, led to the introduction of derived geometry and other homotopical methods into mathematical physics; these techniques have become a fruitful and active area of modern research quite independently of supersymmetry.

Pure spinor superfields were originally introduced in order to cope with both of the aforementioned problems simultaneously, in the context of maximally supersymmetric theories~\cite{Cederwallspinorial,Berkovits2000Super-PoincareSuperstring,Berkovits2001CovariantSpinors}. A major goal was to make supersymmetry ``manifest''; the supersymmetry transformations of the fields of the theory should be geometric in nature, and thus the theory should be formulated in terms of superfields. Correspondingly, the action of supersymmetry transformations should be strict. It was found that the pure spinor superfield formulation of ten-dimensional super Yang--Mills theory naturally gives rise to a superspace model of the derived critical locus, thus automatically reproducing the Batalin--Vilkovisky approach to this theory. In modern language, the homotopy action of supersymmetry on the standard BV theory arises via homotopy transfer.

While most of the literature on the formalism has to do with specific examples of maximally supersymmetric theories, it has been realised in recent years that the scope of the formalism is much broader. Pure spinors provide a systematic technique to construct off-shell superspace models of supermultiplets; in fact, \emph{any} supermultiplet can be obtained from a pure spinor superfield, and (in an appropriate sense) from a uniquely determined one~\cite{equivalence}. Here, too, ideas from derived geometry play an essential role: the construction provides an equivalence of dg-categories between an appropriately defined category of supermultiplets and the category of equivariant sheaves on (a derived replacement of) the space of pure spinors. The equivalence is essentially a version~\cite{KapranovKoszul} of the general phenomenon of Koszul duality. (In general, some care is required in discussing such equivalences. On general grounds, one has a fully faithful symmetric monoidal embedding from quasi-coherent sheaves on the derived replacement to representations of the supertranslation algebra; see~\cite[Theorem~2.4.1]{DAGX}. 
In our context, the remedy is that a \emph{supermultiplet} is not the same as a generic module over the supertranslation algebra; such issues are not our focus in this paper, and we will ignore them in what follows.)

Koszul duality phenomena have various incarnations, and it has become increasingly clear that they are of fundamental relevance to mathematical physics. For example, the observables of a theory in the BV formalism are essentially the Koszul dual (Lie algebra cochains) of the local $L_\infty$ algebra structure that the BV action defines on the fields~\cite[just for example]{CG1}. 
Building on this, Costello and Li~\cite{CostelloLi} have pointed out that Koszul duality should be relevant to the understanding of holographic dualities, generalising the standard piece of the AdS/CFT dictionary which says that bulk fields should correspond to boundary observables. For a recent survey of some other appearances of Koszul duality in modern mathematical physics, see~\cite{BrianNatalie}.
The mathematical literature on Koszul duality is enormous, and we cannot give complete references here; for some selected foundational work on Koszul duality, the reader might look at~\cite{Priddy,Floystad,BGS,Positselski1,Positselski2,operadginzburg}, just for example.

The instance of Koszul duality that will be relevant to us in this work is between $L_\infty$ algebras and commutative differential graded algebras. However, it is important to emphasise that we will \emph{not} make reference to the standard story in the BV formalism described above. Our application of Koszul duality will be, for those well-versed in the BV formalism or the picture of field theories as formal moduli problems, the ``wrong way around'': we will study a hidden commutative algebra structure on the \emph{fields} of particular free supermultiplets, and relate this to an $L_\infty$ structure on the \emph{observables}.

The commutative structure that is relevant for us is closely related to the commutative structure on the ring $S$ of functions on the space $Y$ of generalised pure spinors. (It is worth noting here that the term ``pure spinor'' is, in this context, in some sense a historical accident: our notion encompasses the space of Maurer--Cartan elements in any graded Lie algebra with support in degrees one and two, and therefore any graded ring $S$ defined as a quotient of formal power series by homogeneous quadratic equations.) The pure spinor formalism constructs, in particular, a supermultiplet from any appropriately equivariant sheaf on~$Y$. As such, there is always a uniquely determined ``canonical'' multiplet associated to the structure sheaf (the ring $S$ itself). 
In examples of physical relevance, the canonical multiplet is typically an important or well-known one: among the list of canonical multiplets, one sees, for example, the eleven-dimensional and type IIB supergravity multiplets; the Yang--Mills multiplets for minimal supersymmetry in dimensions three, four, six, and ten; and the $\N=(2,0)$ abelian tensor multiplet. In each case,
the commutative structure on $S$ gives rise to a commutative structure on the canonical multiplet.

This commutative structure has been understood and appreciated in examples. The pure spinor formulation of interactions in ten-dimensional super Yang--Mills theory, for example, uses a Chern--Simons-type action; just as in the BV approach to standard Chern--Simons theory, the multiplet comes equipped with a commutative structure, and tensoring this with the gauge Lie algebra gives rise to the $L_\infty$ structure describing the interacting BV theory. It is also relevant to first-quantised pure spinor models, in particular the pure spinor superparticle~\cite{Berkovits2001CovariantSpinors}. Here, the \emph{fields} of the theory in the target space arise as \emph{observables} of the worldline theory. We expect the commutative structure of the canonical multiplet to be relevant to first-quantised descriptions in general, although there are many subtleties and we do not pursue this here.

In order to understand the Koszul dual of this cdga, we need to pass to yet a larger model of the multiplet in question. Although the standard pure spinor description is freely resolved over superspace, it is not freely generated as a commutative algebra, because $S$ itself is not. We thus need to replace $S$ by its \emph{Tate resolution}~\cite{Tate}, which resolves the ideal defining $S$ by a cdga that is freely generated over formal power series. In the language of rational homotopy theory, we construct a Sullivan model for $S$; Sullivan algebras are perhaps better known in the supergravity literature as ``free differential algebras'', and have appeared in various places in prior work, notably in the work of the Italian school~\cite{CdAF}. The Tate resolution of $S$ can then be identified with the cochains of an $L_\infty$ algebra $\ft^{(\infty)}$, which arises via a sequence of iterated central extensions of the supertranslation algebra $\ft$. $\ft^{(\infty)}$ thus maps to~$\ft$; the multiplet itself can be viewed as the homotopy fiber of the corresponding map of formal moduli problems, and we prove that it is modeled by the subalgebra of $\ft^{(\infty)}$ supported in degrees greater than or equal to three.

In many of the examples we consider, the variety $Y$ is a minimal orbit of the action of the automorphisms of $\ft$ on $\ft^{1}$, and the algebra $\ft^{(\infty)}$ itself can be understood as the positive part of a Borcherds--Kac--Moody (BKM) superalgebra.
The connection between such algebras and functions on minimal orbits via Koszul duality
was studied previously in~\cite{Cederwall2015SuperalgebrasFunctions}.
As such, our results can be also understood as a generalisation of that work. The  relevance of BKM algebras to supergravity theories, $U$-dualities, and exceptional generalised geometry has been understood for some time (see, just for example, \cite{dualdualdual,Borch1,Borch2}, as well as~\cite{Cederwall:2018aab}); it would be interesting to pursue connections to these other appearances of similar algebraic structures, though we do not pursue this in any deep way here. We hope to return to the question in future work.

We are not the first to study the Koszul duals of the algebras that appear in the pure spinor formalism, though our setup is more general than has appeared in the literature so far.
The Tate resolution of the pure spinor constraint was considered in~\cite{Chesterman:2002ey}---again, motivated by issues related to first-quantised models---and continued in~\cite{Berkovits2005TheSpinors,Aisaka2008PureSpectrum}, in particular from the point of view of partition functions/characters.
In a series of fundamental papers~\cite{Movshev2004OnTheories, movshev2006algebraic, Movshev:2009ba}, Movshev and Schwarz studied constructions and proved results analogous to ours in the example of ten-dimensional super Yang--Mills theory. (In ten-dimensional minimal supersymmetry, the Yang--Mills multiplet is the canonical multiplet.) In~\cite{MovshevBar}, Movshev went on to uncover many of the key structural features of the formalism, including extensions to Yang--Mills theories with less supersymmetry, for which the canonical multiplet is the \emph{off-shell} (BRST, rather than BV) Yang--Mills multiplet.
Building on this body of work, an inspiring paper of Gorbounov--Schechtman~\cite{GorbounovSchechtman} commented on the connections to Tate resolutions of commutative algebras and ideas from rational homotopy theory, and further work of 
G\'alvez--Gorbounov--Shaikh--Tonks~\cite{Galvez2016} generalised the ideas of Movshev--Schwarz to general Koszul algebras. 
In addition, an important paper of Gorodentsev, Khoroshkin, and Rudakov~\cite{GKR}
gave a clear exposition of constructions and results entirely parallel to ours in the case of symmetric spaces.
In this paper, with applications to supersymmetric physics in mind, we generalise further to any canonical multiplet; we do not require Koszulity, though we comment on the role it plays later on.
We also note that the explicit study of the Tate resolution for the constrained spinors used in $D=11$ supergravity was initiated in~\cite{Jonsson}, where some conjectures along the lines of our results were formulated.

For the purposes of this work, we consider only canonical multiplets; we do not address the situation where multiplets are obtained from non-scalar pure spinor superfields (sheaves other than the structure sheaf on $Y$, or non-free $S$-modules) in any detail. 
Many of the statements should carry over to such multiplets, both in terms of modules over the rings we consider and of their Koszul dual algebras. We save these considerations for the future.

The organisation of the paper is as follows:
\S\ref{sec:setup} introduces some general concepts, including the categories we are working in, issues related to gradings, and partition functions. We briefly review relevant aspects of the pure spinor superfield formalism and of Borcherds--Kac--Moody superalgebras.
In \S\ref{sec:KoszulDual} we introduce the Koszul dual to the pure spinors. A main tool is the Tate resolution of the constraint on the ``pure spinor'' $\lambda$. The resolution is a freely generated commutative dg algebra, and hence defines an $L_\infty$ algebra structure on $\ft^{(\infty)}$, the linear dual of the space of generators (shifted by one). This $L_\infty$ algebra is the Koszul dual to pure spinors; note that, although the algebra $S$ is quadratic, this notion of Koszul duality coincides with the traditional quadratic one only when the coordinate ring of pure spinor space is Koszul. 
By incorporating the Tate resolution in 
the pure spinor superfield formalism, we obtain a freely generated cdga resolving, not just $S$, but the canonical multiplet itself. We can then use homotopy transfer to show equivalence of this cdga both to the standard version of pure spinor superfield cohomology and to the Lie algebra cohomology of the subalgebra of $\ft^{(\infty)}$ in degrees three and above.
We also discuss in some detail how $\ft^{(\infty)}$ can be constructed from the supertranslation algebra and the observables of the canonical multiplet by the introduction of certain cocycles, and give an interpretation in terms of forms on superspace, a more standard notion in the physics literature.
In \S\ref{sec:Properties}
 we list and discuss interesting structures, both on the mathematical, ring-theoretical, and on the physical sides.
\S\ref{sec:examples}
presents a large number of physical (and some less physical) examples, illustrating our approach.
Finally, in \S\ref{sec:Discussion}, our results are discussed, and some open questions and future directions are pointed out.

\vskip6pt
\section{Setup\label{sec:setup}}
\vskip8pt
\subsection{Conventions}
\hfill
\nopagebreak

We begin by briefly setting up our conventions and quickly reviewing some well-known material we will use in what follows. In particular, we fix conventions for gradings and for the monoidal structure on the category in which we work, as well as for partition functions (Hilbert and Poincar\'e series).

\numpar[sec:otherconventions][Notes on notation] 
We use the notation $\wedge$ for the graded antisymmetric, and $\vee$ for the graded symmetric part of any tensor product. 
We often use signs to denote parity; if $V$ is a vector space, an element in $-V$ is defined as an element in $V$ with odd parity (a fermion).
Note that $(-1)^n\vee^n V=\wedge^n(-V)$. 
Given a vector space, 
the tensor algebra on $V$ is denoted $T[V]$, and $\Sym^\bu(V)$ denotes the (graded) symmetric algebra.
The dual of a vector space $V$ is denoted $\overline{V}$; see below in~\S\ref{sec:bideg} for comments on duals.
An irreducible highest weight module of a semi-simple Lie algebra is denoted $R(\mu)$, where $\mu$ is the highest weight. If the module is labelled by a lowest weight, we use the notation $R(-\mu)$.
Representations and modules of finite-dimensional semi-simple Lie algebras are sometimes denoted by the Dynkin label of $\mu$, and sometimes by their dimension, in bold. 

We will usually denote Lie algebras by Fraktur letters, but do not adhere to this convention universally. Some exceptions are as follows:
the free Lie algebra on a vector space $V$ is denoted $F[V]$;
the Borcherds--Kac--Moody superalgebra labeled by a Kac--Moody superalgebra $\fg$ and a dominant integral weight $\mu$ of~$\fg$ is denoted $\mathscr{B}(\fg,\mu)$.
 Semidirect sums of Lie algebras are indicated with the symbol $\inplus$, with the algebra $\fa$ acting on its module $\fb$ in $\fa\inplus\fb$.
Categories, where they appear, are in sans-serif type, so that (for example) $\Ch$ denotes the category of cochain complexes. We will make repeated use of homotopical arguments; when we refer to two objects as equivalent, we mean equivalent up to homotopy in the category in which these objects live.

\numpar[sec:grvect] 
Throughout the paper, we work in an appropriate monoidal category of graded vector spaces over a field $k$ of characteristic zero. For the purposes at hand, one loses nothing by imagining that $k=\C$ throughout. An object of this category is a bigraded vector space
\deq{
V = \bigoplus_{p,q} V^{p,q}\;,
} with each graded component $V^{p,q}$ finite-dimensional, possibly equipped with a differential of bidegree $(+1,0)$. The monoidal structure uses the Koszul rule of signs, with the parity of a homogeneous element in bidegree $(p,q)$ being $p+q \bmod 2$.
This category was referred to as ``lifted dgs vector spaces'' in ref. \cite{perspectivesonpurespin}. One can alternatively think of it as being the category of cochain complexes of graded vector spaces, where the category $\GrVect$ of graded vector spaces is taken to be monoidal according to the Koszul sign rule. We will sometimes have cause to refer to the subcategories $\GrVect_\pm$ of $\GrVect$, consisting of objects with non-negative ($+$) or non-positive ($-$) grading; similarly, $\GrVect_>$ will refer to the subcategory of objects with grading bounded from below. $\GrVect$ is a full subcategory of $\Ch(\GrVect)$ 
by regarding an object as a cochain complex concentrated in degree zero. We will also need the subcategory consisting of cochain complexes in $\GrVect$ that are bounded from below (or above) in the internal grading (the grading on $\GrVect$) and bounded in cohomological degree in each internal degree. This category will be called $\Ch(\GrVect_>)_b$---but note that objects are \emph{not} necessarily bounded cochain complexes!

\numpar[sec:bideg] The grading on an object in $\GrVect$ will be called the ``internal grading'', and the grading on a cochain complex will be called the ``cohomological grading''. 
In physical terms, cohomological grading should be thought of as ghost number, and internal degree modulo two as intrinsic fermion parity (both with appropriate shifts). We will sometimes have cause to refer to the negatives of these gradings: the negative of the cohomological degree will be called ``homological degree'', and the negative of the internal degree will be called the ``weight degree''. Throughout, duals reverse grading: that is,
\deq{
(\overline{V})^{p,q} = \overline{V^{-p,-q}}\;.
}
As such, the dual of a cochain complex is also a cochain complex; our conventions are cohomological throughout, so that every differential has degree $+1$. Duals are defined to be the sum of the degreewise linear duals; note that this is \emph{not} the linear dual in  the infinite-dimensional setting, but that the dual of a dual is the original object.

The only important shift functor will be the shift functor on $\Ch(\GrVect)$, which preserves weight grading and shifts cohomological grading. For this shift functor, we use square brackets; our convention is that
\deq{
V[n]^{p,q} =  V^{p+n,q}\;.
}
Thus, for any object $W$ of $\GrVect$, $W[n]$ denotes the cochain complex consisting of $W$ placed in degree $-n$.

The totalisation of the cohomological degree and the internal degree will be called the ``consistent'' or ``totalised grading'', as the parity of any homogeneous element is determined by this degree. $V^{p,q}$ is in consistent degree $p+q$. In later sections, we sometimes have cause to  refer to the bidegree of an element using the ``Tate bidegree'', consisting of consistent and weight degree, rather than cohomological and internal degree: when we do this, we will use double parentheses, so that an element of bidegree $(p,q)$ has Tate bidegree $\bid{p+q,-q}$.

\numpar[Algebraic structures] No matter which algebraic structures are being considered, we always work internally in $\GrVect$. Thus a ``cochain complex'' or ``dg vector space'' is an object of $\Ch(\GrVect)$, a ``commutative algebra'' is a $\Z$-graded algebra which is graded commutative, and a ``commutative dg algebra'' is a bigraded algebra, commutative with respect to the totalised grading, and equipped with a differential of bidegree $(1,0)$.

An $L_\infty$ algebra (again internal to $\GrVect$) is an object $\fa$ of $\Ch(\GrVect)$, equipped with structure maps
\deq{
\mu_k: \wedge^k \fa \to \fa }
of bidegree $(2-k,0)$ satisfying all higher Jacobi identities. A ``Lie algebra'' is an $L_\infty$ algebra concentrated in cohomological degree zero: that is, a consistently $\Z$-graded Lie superalgebra.\footnote{In the sequel, we will have cause to consider strict minimal $L_\infty$ algebras---for which only the 2-ary bracket is nontrivial---which are nevertheless supported in nontrivial cohomological degrees. These are, of course, also (super) Lie algebras in the standard sense; we may abusively refer to them as such, but no confusion should arise.} 

Given an $L_\infty$ algebra $\fa$, its Lie algebra cochains $C^\bullet(\fa)$ are a cdga, freely generated by $\overline{\fa}[-1]$. We will write $\Sym^\bu(V)$ for the freely generated commutative algebra on an element $V$ of~$\Ch(\GrVect)$, 
 and $\Symc^\bu(V)$ for the completion with respect to the canonical maximal ideal $\Sym^{>0}(V)$.
The cochains are then the cdga
\deq{
C^\bu(\fa) = (\Symc^\bu(\overline{\fa}[-1]), d_\text{CE})\;,
}
where the Chevalley--Eilenberg differential $d_\text{CE}$ is the sum of the maps $\overline{\mu_k}$ over all $k>0$. If there are infinitely many nonvanishing $L_\infty$ operations, completion is essential; if not, we can choose to work without completing.
The Chevalley--Eilenberg differential has bidegree $(1,0)$, and it squares to zero precisely when the $\mu_k$ define an $L_\infty$ structure.
(If $\fa$ is the Lie algebra of a super Lie group $A$, then the generators $\overline{\fa}[-1]$ can be thought of as left-invariant one-forms; completing means that we treat fermionic one-forms as formal variables.)

\numpar[sec:equivariant][Working equivariantly]
In general, we will be interested not only in the category $\GrVect$, but also in an equivariant version $\GrVect^G$, where objects are graded vector spaces equipped with an action of a reductive algebraic group $G$. (The reader will lose nothing by imagining that $k=\C$ and that $G$ is a complex Lie group.) $\Ch(\GrVect^G)$ and other related objects are defined in the obvious way.

\numpar[sec:Z][Partition functions]
A useful tool for examining the structure of 
the objects we will study (in particular, cdgas and $L_\infty$ algebras) is provided by partition functions (characters). 
Let us quickly fix our conventions.

For $G$ a reductive algebraic group as above, we let $\Vectcat^G$ denote the category of $G$-representations, and $\Rep(G)$ denote the representation ring of $G$. An element of~$\Rep(G)$ is a formal integral linear combination of $G$-representations; sums are identified with direct sums, and the ring structure arises from the tensor product in~$\Vectcat^G$. (In other words, $\Rep(G)$ is the Grothendieck ring $K_0(\Vectcat^G)$.)

A \emph{partition function} can be meaningfully assigned to any object of $\Ch(\GrVect^G_>)_b$.  Recall that, for such objects,  the internal grading is bounded from below, and the homological grading is bounded in each internal degree. 
For such an object, the \emph{Euler characteristic} defines an assignment
\deq{
\chi: \Ob \left(\Ch(\GrVect^G_>)_b \right) \to K_0(\GrVect^G_>), \quad
\bigoplus_p V^{p,\bu} \mapsto \sum_p (-1)^p V^{p,\bu}\;.
}
We can then further observe that $K_0(\GrVect^G_>) = K_0(\Vectcat^G)(\!( t )\!) = \Rep(G)\bid{t}$, so that the Euler characteristic can be viewed as the formal Laurent series with coefficients in the representation ring constructed via this procedure. The Euler characteristic of quasi-isomorphic cochain complexes agree, and direct sums and tensor products are carried to addition and multiplication in $\Rep(G)\bid{t}$. 
In what follows, we will often use $\oplus,\ominus,\otimes$ for operations in this ring.

The \emph{partition function} is essentially the Euler characteristic, as thus defined; however, in our conventions for partition functions, we make a couple of adjustments. First off, we prefer for the sign to reflect the consistent grading; therefore, we define 
\deq{
Z_V(t) = \chi_V(-t), \qquad V \in \Ob\left( \Ch(\GrVect_>^G)_b \right).}
Secondly, we will have cause to study objects where the internal grading is bounded from \emph{above}. We maintain the convention that the partition function should be a formal Laurent series in the variable $t$, rather than $t^{-1}$; therefore, we define the partition function of an object of $\Ch(\GrVect_<^G)_b$
by $Z_V(t) = \chi_V(-t^{-1})$.

The partition function is thus an object that encodes the content in terms of $G$-modules at each weight degree; the weight degree is encoded as the power of a formal parameter $t$. It should be noted that the partition function no longer keeps track of the cohomological grading; it only carries information about the weight grading and the Koszul sign of the $G$-modules.

\numpar[sec:Spartition][The Hilbert and Poincar\'e series of a commutative ring $S$.]
According to the above, the partition function takes values in the tensor product of the representation ring of $G$ 
    and formal power series in $t$, and we write
\begin{align}
Z_S(t)=\bigoplus_{n=0}^\infty S_nt^n\;.
\end{align}
Replacing $S_n$ with $\dim(S_n)$ gives the Hilbert series.

Note that our definition of a partition function for $S$ is its Hilbert series; this is to be contrasted with the \emph{Poincar\'e series}, which is often encountered in the theory of commutative rings (see \eg~\cite{Avramov:1982,Froberg1999}), and which is based on a minimal (equivariant) resolution $L^\bu$ of $k$ in free $S$-modules:\footnote{This resolution is in a certain sense dual to the Tate resolution of $S$, see {\S}\ref{sec:dualTate}.}
\begin{equation}
L^\bu = \begin{tikzcd}[row sep = 16 pt, column sep = 16 pt]
\cdots\ar[r,"\phi_4"]
&S\otimes \overline{V}_{-3}\ar[r,"\phi_3"]
&S\otimes \overline{V}_{-2}\ar[r,"\phi_2"]
&S\otimes \overline{V}_{-1}\ar[r,"\phi_1"]&S
\end{tikzcd} \simeq k,
\end{equation}
where $\overline{V}_i$ are objects of $\GrVect^G$, $\dim V_i=b_i$, $V_0=k$. 
The partition function (the equivariant Poincar\'e series of $S$) is often defined in the literature as
\begin{align}
P_S(t) =\bigoplus_{i=0}^\infty \overline{V}_{-i} t^i\;.
\end{align}

The corresponding numerical Poincar\'e series is then 
$\sum_{i=0}^\infty b_it^i$. 

The complex $L^\bu$ computes the Koszul dual $\Ext_S^\bu(k,k)$ of $S$ in associative algebras. (In general, $\Ext_S^\bu(k,k)$ has an $A_\infty$ structure; when $S$ is commutative, this can be identified with the enveloping algebra, in the sense of Baranovsky~\cite{BaranovskyEnveloping}, of the $L_\infty$ Koszul dual we consider in this paper.) Since the complex resolves $k$ in free $S$-modules, we can compute
\deq{
\Ext^j_S(k,k) = \Hom_S(L^j,k) = V_{j}
.
}
This implies a relationship between $P_S(t)$ and $Z_S(t)$.
In order to see this relationship, we need to rearrange the terms in this partition function to make it adhere to our conventions above. 
In particular, we need to take the Euler characteristic; however, this does not correspond to setting $t=-1$ in the above definition, since it does not properly account for gradings.
$S$ is graded, so that $L^\bu$ is bigraded, and the resolution differential $\phi$ has (consistent) bidegree $\bid{1,0}$.
To be concrete, we can split each $\overline{V}_i$ as $\overline{ V}_i=\bigoplus_{w=i}^{N_i}\overline{V}_{i,w}$, where $\overline{V}_{i,w}$ carries weight degree $w$ and $N_i$ is a finite number.
$L^\bu$ is now an object of $\Ch(\GrVect_-^G)_b$, so that, following the above convention, we should define
\deq{
\tilde P_S(t):= Z_{\Ext_S^\bu(k,k)}(t) 
=
\bigoplus_{i=0}^\infty\bigoplus_{w=i}^{N_i} \overline{ V}_{i,w}(-1)^it^w\;.
\label{eq:Ptilde}}

Because $L^\bu$ is a resolution of $k$, we can then compute the total partition function of the complex as
$Z_S(t)\otimes\tilde P_S(t) = Z_{L^\bu}(t)={1}$.
Thus $Z_S(t)=(\tilde P_S(t))^{-1}$.
$S$ being a Koszul algebra is equivalent to all $\hbox{deg}(\phi_i)=1$ 
\cite{Froberg1999}. Then the sum over $w$ in \eqref{eq:Ptilde} contains a single term $w=i$, and
$Z_S(t)=(P_S(-t))^{-1}$.

\numpar[sec:Utpartition][Partition functions of enveloping algebras]
Let $\fa$ be a Lie algebra in positive internal degrees, containing the module $R_n$ (of ${\rm Der}(\fa)$, see {\S}\ref{supertranslations}) at internal degree $n$, where fermionic generators are counted with a minus sign. The partition function 
(always twisted by fermion number) of the universal enveloping algebra $U(\fa)$ is 
\begin{align}
Z_{U(\fa)}(t)=\bigotimes_{n=1}^\infty
(1-t^n)^{-R_n}\;. \label{eq:ZUa}
\end{align}
Here we use the shorthand $(1-t^n)^{-R_n}$ for the partition function of bosons at degree $n$ in the module $R_n$, 
\begin{align}
    (1-x)^{-V}=\bigoplus_{n=0}^\infty\vee^n V\,x^n\;,
\end{align}
which is practical since it obeys the usual multiplication rules for powers. In particular, its inverse is the partition function for a fermion in $V$, 

\begin{align}
    (1-x)^{V}=\bigoplus_{n=0}^{\dim V}(-1)^n\wedge^n V\,x^n\;.
\end{align}
The corresponding Hilbert series is obtained by replacing modules with their dimensions. 
If $\fa$ is $L_\infty$, but not strictly Lie, $U(\fa)$ as such is not well defined.
The right hand side of~\eqref{eq:ZUa}
is still meaningful, 
and has the interpretation as the partition function for the $A_\infty$ structure on $\Sym^\bullet(\fa)$ as defined by Baranovsky
\cite{BaranovskyEnveloping}, which is the proper generalisation of $U(\fa)$ to $L_\infty$ algebras. This will be taken as the meaning of $U(\fa)$ in the following.
It of course still agrees with 
the conjugate inverse of the partition function of the complex 
$C^\bullet(\fa)$, see {\S}\ref{sec:partitionKoszul}. 

\subsection{(Generalised) supertranslations} 
\hfill
\nopagebreak

In this section, we set up the general scenario in which  we would like to work. While our constructions are inspired by applications in supersymmetric field theory, and our examples will all be taken from this setting, we emphasise that the superalgebras we consider need not be  limited to just supertranslation algebras of physical interest. 

\numpar \label{supertranslations} Let $\ft$ be a Lie algebra, supported in internal degrees one and two. It follows immediately that $\ft$ is a central extension of the odd abelian Lie algebra~$\ft_1$ by~$\ft_2$:
\deq[SES1]{
  0 \to \ft_2 \to \ft \to \ft_1 \to 0\;.
}
We will think of $\ft$ as a physical supertranslation algebra: it will act by translation symmetries on flat superspace, and these symmetries preserve ghost number since $\ft$ is concentrated in cohomological degree zero.

Let $\Der(\ft)$ be the Lie algebra of bidegree-preserving endomorphisms of~$\ft$ that are derivations of the bracket. (This is the Lie algebra of the group $\Aut(\ft)$ of bidegree-preserving automorphisms of~$\ft$.) Since $\ft^{0,0}=0$, there are no bidegree-preserving inner derivations, and we can extend the Lie algebra by placing $\Der(\ft)$ in degree zero if we like. 
 
(In physical terms, the result would be the  super-Poincar\'e algebra, including Lorentz and $R$-symmetry as well as a generator for scale transformations, the grading generator.) In general, though, we will focus just on the algebra $\ft$, while insisting that all constructions are equivariant for the Lie algebra $\Der(\ft)$, or perhaps more generally for some chosen subalgebra thereof (though we will not make such choices in what follows). 

\numpar[sec:ps][Superspace.]
Given a Lie algebra $\ft$ of this form, we can construct a corresponding \emph{superspace} in the following fashion: We consider the supergroup $T = \exp(\ft)$. Since $\ft$ is two-step nilpotent, $T \cong \ft$ as (graded) vector spaces. As usual, $T$ acts on itself on the left and on the right. Both the left- and right-invariant vector fields are isomorphic to $\ft$ and define two commuting copies of~$\ft$ inside of~$\Vect(T)$. These are the operators typically called $Q_\alpha$ and $D_\alpha$ in the physics literature.

\numpar[sec:gradings][Pure spinor space and its coordinate ring] We consider now the Lie algebra cochains of~$\ft$.
Since $C^\bu(-)$ is a contravariant functor, the above short exact sequence~\eqref{SES1} determines a diagram 
\deq{
  k \to C^\bu(\ft_1) \to C^\bu(\ft) \to C^\bu(\ft_2) \to k
}
of cdga's,
in which the first and last terms have zero differential. The first  map in this sequence witnesses $C^\bu(\ft)$ as an algebra over the polynomial ring 
\deq{
  R = C^\bu(\ft_1) = \Sym^\bu(\overline{\ft_1}[-1])\;,
} 
graded so that its generators sit in bidegree $(1,-1)$. We will often denote these generators with the symbol $\lambda$, so that $R = k[\lambda]$. The second map shows that, as a graded commutative algebra over $R$, $C^\bu(\ft)$ is freely generated by $\overline{\ft_2}[-1]$, which we will represent by odd generators $v$ in bidegree $(1,-2)$. 

One sees immediately that $C^\bu(\ft)$, with respect to the consistent (totalised) grading, is the (cohomologically, thus negatively, graded) Koszul complex associated to the sequence of $\dim(\ft_2)$ quadratic elements of~$R$ that are determined by the structure constants. The zeroth cohomology, with respect to the consistent grading, is thus the ring 
\deq{
  S = R/\langle\lambda^2\rangle
}
of functions on the space of Maurer--Cartan elements in~$\ft$, where $\lambda^2$ denotes some expression quadratic in $\lambda$ corresponding to the generators of $\ft_2$. (Since the differential on $\ft$ is trivial, Maurer--Cartan elements are just elements $Q \in \ft_1$ satisfying $[Q,Q] = 0$.)
We  define this space, the ``Maurer--Cartan set'' or ``nilpotence variety'' or ``generalised pure spinor space'', to be\footnote{For non-experts: the spectrum $\Spec S$ of a ring $S$ is its set of prime ideals. One can think of a maximal ideal as the collection of those functions on $Y$ that vanish at a specific point $\lambda_0$ in pure spinor space. This ideal is generated by the monomials $\lambda-\lambda_0$. More generally, prime ideals can be thought of as those functions vanishing on particular algebraic subsets, defined by the equations that generate the ideal. The spectrum can be equipped with the Zariski topology, in which closed sets are algebraic subsets, as well as a sheaf of local functions for which $S$ is the global sections. We will freely make use of this correspondence between spaces and rings in what follows.}
\deq{
  Y = \Spec S .
}
Since $S$ is graded (defined by homogeneous quadratic equations in the polynomial ring $R$), $Y$ admits a $k^\times$ action; it is the affine cone over the projective scheme 
\deq{
  P(Y) = \Proj S\;.
}
The weight grading corresponding to the standard $k^\times$ action, with respect to which generators of~$R$ are placed in degree $1$, is the negative of the internal grading.
\subsection{The (generalised) pure spinor superfield formalism}
\hfill
\nopagebreak

\numpar[][Multiplets]
The pure spinor superfield formalism associates, to any $\Der(\ft)$-equivariant $C^\bu(\ft)$-module, a ``multiplet'' for $\ft$, which is a supersymmetry-equivariant sheaf on the body of the supermanifold $T$ (the normal bosonic spacetime). In fact, the formalism defines an equivalence of categories, so that every supermultiplet admits a unique description in terms of an equivariant $C^\bu(\ft)$-module (or, equivalently, in terms of an equivariant sheaf on the classifying space $B\ft = \Spec C^\bu(\ft)$). We briefly review the relevant aspects here; for details and further background, we refer to the literature~\cite{Cederwallspinorial,perspectivesonpurespin,Cederwall2013PureOverview,Cederwall:2022fwu}.
We note that, in the physics literature, the term ``supermultiplet'' is usually understood, depending on context, to mean an indecomposible representation module of the supersymmetry algebra, either on \emph{fields} (in the case of a local supermultiplet) or on \emph{particles} (which are unitary irreducible representations of the Poincar\'e group, finitely many of which combine into a unitary irreducible representation of the super-Poincar\'e group). In the presence of a free action, the Hilbert space of the free field theory of a field multiplet is the Fock space on a corresponding particle multiplet. Our definition of the term captures the notion of a \emph{field} multiplet, but only includes those gauge transformations or equations of motion that are necessary for supersymmetry to close in an appropriate off-shell formalism (either BRST or BV).

\numpar[sec:Koszul] 
The formalism works using a kernel which is a slight variant of the usual Koszul duality kernel.
By a ``kernel'' in this context, we mean a free resolution of the ground field $k$ in bimodules over the algebra and its Koszul dual; for Koszul duality in associative algebras, this is the object $L^\bu$ discussed in~\S\ref{sec:Spartition}.

For our purposes in this paper, only Koszul duality between the commutative and Lie operads is relevant. In this case, the kernel can be made more explicit using the Chevalley--Eilenberg cochains.
The usual kernel for commutative/Lie Koszul duality is $K_0(\ft) = U(\ft)  \otimes_k  C^\bu(\ft)$, 
equipped with the scalar, square-zero, acyclic differential that arises from the identity element in $\ft \otimes \overline{\ft}$ acting using both module structures simultaneously, and considered as a differential graded $(C^\bu(\ft),U(\ft))$-bimodule. 
The equivalence of categories then associates to a $U(\ft)$-module $N$ the $C^\bu(\ft)$-module $N\otimes_{U(\ft)}K_0(\ft) $, 
and similarly to a $C^\bu(\ft)$-module $M$ the $U(\ft)$-module $K_0(\ft)\otimes_{C^\bu(\ft)} M$.
The algebras $U(\ft)$ and $C^\bu(\ft)$ are then said to be \emph{Koszul dual}; in this special case, where no higher brackets are present, (the enveloping algebra of) a dg Lie algebra is dual to a dg commutative algebra. In general, the Koszul dual of a freely generated cdga is an $A_\infty$ algebra that can be understood as the enveloping algebra (in the sense of Baranovsky \cite{BaranovskyEnveloping}) of an $L_\infty$ algebra, of which the original cdga is the Chevalley--Eilenberg cochains. We are thus free to speak about a duality between (homotopy) commutative and Lie (or $L_\infty$) algebras, and we will do this in what follows.
We comment in more detail on this in~\S\ref{sec:commlie} and~\S\ref{sec:QK} below.

\numpar
To obtain multiplets via the (generalised) pure spinor formalism, we replace $U(\ft)$ in the kernel by the sheaf $C^\infty(T)$ of smooth functions on superspace, which (like $U(\ft)$ itself) is a $U(\ft)$-bimodule via the left and right actions discussed in \S\ref{sec:ps}.
Concretely,
the kernel takes the form
\deq{
    K(\ft) = C^\infty(T) \otimes_k C^\bu(\ft) = C^\infty(T)[\lambda,v]\;,
}
with a differential given by the formula 
\deq{
    \cD = v \pdv{ }{x} + \lambda \left( \pdv{}{\theta} - \theta \pdv{ }{x} \right) + \lambda^2 \pdv{ }{v}\;.\label{cDEq}
}
In fact, this kernel is just the de Rham complex of the superspace $T$, but expressed in an unusual (right-invariant) frame. Note that 
\deq{
    \Omega^\bu(T) = C^\infty(T)[dx,d\theta]\;, \quad d_\text{dR} = dx \pdv{ }{x} + d\theta \pdv{ }{\theta}\;.
}
The two are related via the map of commutative algebras that sends $d\theta \mapsto \lambda$ and $ dx \mapsto v - \lambda\theta$. (This is in fact a cochain-level isomorphism of cdga's.) There is a $(U(\ft),C^\bu(\ft))$-bimodule structure: $U(\ft)$ acts by the left-invariant vector fields
\deq{
    Q = \pdv{ }{\theta} + \theta \pdv{ }{x}\;,
}
and the $C^\bu(\ft)$-module structure arises via the identification 
\deq{
    C^\bu(\ft)  \cong \Omega^\bu(T)^T \subset \Omega^\bu(T)\;.
}
Given an $\Der(\ft)$-equivariant $C^\bu(\ft)$-module $\Gamma$, the multiplet is constructed by tensoring over $C^\bu(\ft)$ with $\Gamma$;  we will denote the functor that associates a multiplet to $C^\bu(\ft)$-modules by $A^\bu(-)$, so that  
\deq{
    A^\bu(\Gamma) = K_0(\ft) \otimes_{C^\bu(\ft)} \Gamma\;.
    \label{eq: purespinorfunctor}
}
The inverse functor that witnesses this as an equivalence of dg-categories is given by taking the derived $\ft$-invariants of a multiplet---in other words, by considering $C^\bu(\ft,\cE)$ for any multiplet $\cE$. 

\numpar[sec:O][Canonical multiplets]
Recall that the ring $S = R/\langle\lambda^2\rangle$ of functions on the space of square-zero elements in~$\ft$ arises as $H^0(\ft)$. Since $C^\bu(\ft)$ is non-positively graded, there is a quotient map 
\deq{
    C^\bu(\ft) \rightarrow H^0(\ft) = S
}
of commutative differential graded algebras, witnessing $S$ as a $C^\bu(\ft)$-module. (Via this map, any $S$-module is a $C^\bu(\ft)$-module; the resulting multiplets are the ones that can be described by the standard, rather than ``derived'', pure spinor superfield formalism.) We can thus consider the multiplet associated to $S$ for any supertranslation algebra $\ft$. This multiplet is uniquely determined just by the datum of the supertranslation algebra, and we will correspondingly call it the \emph{canonical multiplet}.
In examples of physical relevance, the canonical multiplet tends to play a central role in supersymmetric physics: the list includes the eleven-dimensional and type IIB supergravity multiplets; Yang--Mills multiplets for minimal supersymmetry in dimensions three, four, six, and ten; and the $\N=(2,0)$ tensor multiplet in six dimensions.

Explicitly, the multiplet is given by
\deq{
    A^\bu(S) = \left( C^\infty(T) \otimes_k S, \cD = \lambda^\alpha D_\alpha = \lambda \left( \pdv{ }{\theta} + \theta \pdv{ }{x} \right) \right)\;.
    \label{pureD}
   }
   One sees the usual differential of the pure spinor superfield formalism, for a scalar pure spinor superfield, appearing.
(Where possible without ambiguity, we will leave indices and contractions implicit.)
The cohomology of the canonical multiplet is (a sheaf of) $A_\infty$ algebras on the spacetime; in what follows, we will denote (the global sections of) this cohomology by $\fm$, and the minimal model of the local observables of the canonical multiplet by $\overline{M}$.

Of course $A^\bu(S)$ is a multiplet for~$\ft$, with a strict local $\ft$-module structure at the cochain level. 
However, there are additional important pieces of structure: $A^\bu(S)$ has a cdga structure (in particular, a graded commutative multiplication), inherited from the product structure on~$S$.
We furthermore know that $A^\bu(S)$ is an $S$-algebra: there is an obvious map of cdga's 
\begin{equation}
    S \to A^\bu(S)\;,
\end{equation}
which embeds $S$ along the constant functions on~$T$. Lastly, any multiplet that is associated to an $S$-module (so a ``non-derived'' pure spinor superfield) naturally acquires an $A^\bu(S)$-module structure. 
 The commutative structure is an important and often overlooked piece of data, which (just for example) gives rise to the interactions of ten-dimensional super Yang--Mills theory, simply via tensoring with the gauge algebra. In this paper, we will dedicate ourselves to studying this structure (and its Koszul dual) in detail.

\subsection{Borcherds--Kac--Moody superalgebras}
\label{sec:BorcherdsSuperalgebras}
\hfill
\nopagebreak

As we will see, the Lie superalgebra $\ft=\ft_1\oplus\ft_2$  can in some interesting cases be considered as a subspace (subquotient) of a Borcherds--Kac--Moody (BKM) superalgebra.
Here we recall the general definition. We will assume $k=\mathbb{C}$ in this subsection.

\numpar Let $B$ be a real-valued symmetrisable $(r\times r)$-matrix with non-positive off-diagonal entries
and let $J$ be some subset of the index set $I$ labelling rows and columns.
We say that $B$ is a Cartan matrix of a BKM superalgebra if
each $i\in I$ satisfies one of the following three conditions:
\begin{itemize}
\item[$(i)$] $B_{ii}=2$, $i\notin J$ and $B_{ij} \in \mathbb{Z}$ for all $j\in I$,
\item[$(ii)$] $B_{ii}=1$, $i\in J$ and $B_{ij} \in \mathbb{Z}$ for all $j\in I$,
\item[$(iii)$] $B_{ii}\leq0$.
\end{itemize}
Let $\mathcal{S}$ be a set of $3r$ generators $e_i,f_i,h_i$ for any $i\in I$,
where $e_i,f_i$ are odd for $i\in J$, and all other generators are even.
The BKM superalgebra $\mathscr{B}(B)$ associated to $B$ (together with the set $J$) is the Lie superalgebra generated by $\mathcal S$
modulo the Chevalley relations
\begin{align}
[h_i,e_j]&=B_{ij}e_j\;, & [e_i,f_j]&=\delta_{ij}h_j\;,\nn\\
[h_i,f_j]&=-B_{ij}f_j\;, & [h_i,h_j]&=0\,
\end{align}
and the Serre relations
\begin{align}
(\ad{e_i})^{1-2{B_{ij}}/{B_{ii}}}(e_j)=(\ad{f_i})^{1-2{B_{ij}}/{B_{ii}}}
(f_j)&=0  \qquad \text{if $B_{ii}>0$ and $i\neq j$}\,,\nn\\
[e_i,e_j]=[f_i,f_j]&=0 \qquad \text{if $B_{ij}=0$}\;.
\end{align}
Let $B'$ be the matrix obtained by multiplying row $i$ by some $\varepsilon_i>0$, so that $B'{}_{ij}=\varepsilon_i B_{ij}$,
and let $\mathscr{B}(B')$ be the Lie superalgebra constructed in the same way, but with $B$ replaced with $B'$ in the construction.
It then follows from the relations that the algebras $\mathscr{B}(B)$ and $\mathscr{B}(B')$ are isomorphic,
with an isomorphism given by
$h_i \leftrightarrow \varepsilon_i h_i$
and $e_i  \leftrightarrow \varepsilon_i e_i$,
the other generators unchanged.
Therefore, 
we may equivalently work only with the symmetrised Cartan matrix (obtained as $B'$ by choosing $\varepsilon_i$
such that $B'{}_{ij}=B'{}_{ji}$) as no information about the algebra get lost in the symmetrisation.

Here we will only consider the special cases of BKM algebras where each $i\in I$ in fact satisfies
one of the following two conditions:
\begin{itemize}
\item[$(i)$] $B_{ii}=2$, $i\notin J$ and $B_{ij} \in \mathbb{Z}$ for all $j\in I$,
\item[$(iii)'$] $B_{ii}=0$, $i\in J$ and $B_{ij} \in \mathbb{Z}$ for all $j\in I$.
\end{itemize}
The whole matrix is thus integer-valued, with non-positive
integers as off-diagonal entries. 
When $B$ in addition is symmetric, we may visualise the data contained in the Cartan matrix by a Dynkin diagram 
consisting
of $r$ nodes and, for each pair of nodes $(i,j)$, a number of lines between them, which is equal to $-B_{ij}$.
In order to distinguish between the two cases for the diagonal entries,
the node $i$
is painted ``grey'' in case $(iii)'$, which means that it is drawn as $\otimes$, but left ``white'' in case $(i)$. 

If $J =\varnothing$, then the resulting BKM superalgebra
is a Kac--Moody algebra. 
If $J$ consists of only one element, then we choose a numbering of rows and columns, or of nodes in the Dynkin diagram,
such that $J=\{0\}$. We thus have $B_{00}=0$, and the entries $B_{ij}$ with $i,j>0$ form the Cartan matrix of a Kac--Moody algebra
$\fg$ (possibly finite-dimensional), whose Dynkin diagram is obtained by removing the grey node. 
The off-diagonal entries $B_{i0}$ 
are non-positive integers, which constitute the Dynkin labels of the negative of a dominant integral weight $\mu$ of $\fg$.
Assuming that the BKM superalgebra is of this type, it can thus be characterised by the pair $(\fg,\mu)$ equally well as by the Cartan matrix.
Given such a pair of a Kac--Moody algebra $\fg$ and a dominant integral weight $\mu$ of $\fg$, we denote the corresponding BKM superalgebra
by $\mathscr{B}(\fg,\mu)$, or sometimes only by $\mathscr{B}(\fg)$, when the choice of $\mu$ is understood or unimportant.
(As discussed above, we can set $B_{i0}=B_{0i}$ without loss of generality.)

\numpar Any BKM superalgebra $\mathscr{B}(\fg,\mu)$ has a consistent $\mathbb{Z}$-grading where the subalgebra $\fg$ has degree $0$ and the generators $e_0$ and 
$f_0$ have degree $1$ and $-1$, respectively. It has a subalgebra of the form $\ft_0 \oplus \ft_1 \oplus \ft_2 \oplus \cdots$, where $\ft_0=\fg$
and $\ft_1$ is a lowest weight module $R(-\mu)$
of $\fg$ generated by $e_0$. From $\ft_1$ all subspaces $\ft_k$ with $k\geq1$ can be constructed
from the free Lie algebra $F(\ft_1)$ generated by $\ft_1$ by factoring out the ideal generated by the Serre relation $[e_0,e_0]=0$.
Since the weight of $[e_0,e_0]$ is $-2\mu$, the subspace $\ft_2$ is a module $\vee^2 R(-\mu) \ominus R(-2\mu)$ of $\fg$. 
Considering the full BKM superalgebra $\mathscr{B}(\fg,\mu)$, there are also corresponding modules at negative degrees
with $R_{-p}=\overline{R_p}$.

Suppose that both $B$ and the Cartan matrix of $\fg$ are non-degenerate. Then there is a  
unique element $\tilde \Lambda_0$ in the weight space of $\mathscr{B}(\fg,\mu)$ which is orthogonal to the simple roots of $\fg$
and has coefficient $1$ corresponding to the grey node in the basis of simple roots.
For any $p\geq 1$, the subspace $\ft_p$ 
of $\mathscr{B}(\fg,\mu)$ is spanned by 
root vectors for which the corresponding roots belong to  
a subset $\Phi_p$ of
the root system $\Phi$ of $\mathscr{B}(\fg,\mu)$. Let $R_p$ be the corresponding $\fg$-module, 
let $|R_p|$ be the number of weights in it (counting a weight with multiplicity $m$ as $m$ distinct weights)
and set $t=e^{-\tilde\Lambda_0}$. Then,
\begin{align}
\prod_{\beta\in \Phi_p}(1-e^{-\beta})^{{\rm mult}\,\beta}&=\sum_{k=0}^{|R_p|}(-1)^k\chi({\wedge}^k \overline R_p)t^{kp}\;,\nn\\
\prod_{\beta\in \Phi_p}(1+e^{-\beta})^{-{\rm mult}\,\beta}&=\sum_{k=0}^{|R_p|}(-1)^k\chi({\vee}^k \overline R_p)t^{kp}\;,
\end{align}
where $\chi(R)$ denotes the character of a $\fg$-module $R$.
Following the notation introduced in {\S}\ref{sec:Utpartition}, we can write this as
\begin{align}
\prod_{\beta\in \Phi_p}(1\pm e^{-\beta})^{\mp{\rm mult}\,\beta}&=\chi\big((1\pm t^p)^{\mp \overline{R}_p}\big)\;.
\end{align}
As shown in \cite{Cederwall2015SuperalgebrasFunctions}, it then follows from the Weyl--Kac denominator formula for BKM algebras \cite{Wakimoto,Ray,Ray95,Miyamoto} that
\begin{align}
\bigotimes_{p=1}^\infty (1-t^p)^{(-1)^p \overline{R}_p}=\bigoplus_{q=0}^\infty R(q\mu)t^q\;.
\label{eq:borcherdspartition}
\end{align}
This provides a method of recursively computing the modules $R_p$, and also suggests that
they can be obtained by the Tate resolution described in {\S}\ref{sec:tateres}. 
This will be explained in {\S}\ref{sec:BorcherdsMinimalDual}.
For details, we refer to \cite{Cederwall2015SuperalgebrasFunctions}. 

\numpar[sec:bkmsym] Given a BKM superalgebra $\mathscr B(\fg,\mu)$ it is often interesting to study the BKM superalgebra 
obtained by replacing the gray node in the Dynkin diagram with a chain of $d$ white nodes, followed by a grey node.
We will denote this BKM superalgebra by $\mathscr B(\fg,\mu)^{(d)}$, where $(d)$ may also be replaced by $d$ plus signs
(and $\mu$ still denotes a weight of $\fg$, corresponding to the first white node added in the extension. 
Let $\fg^{(d)}$ be the Kac--Moody algebra obtained by adding only the $d$ white nodes ($d\geq1$).
Whenever there is an $m$ such that $\fg^{(m)}$ is an affine Kac--Moody algebra, a symmetry $\overline R_p = R_{(m-d)-p}$
appears in $\mathscr B(\fg,\mu)^{(d)}$ for $p=1,\ldots,m-d-1$ \cite{Kleinschmidt:2013em,Palmkvist:2012nc}. 
To some extent, $\mathscr B(\fg,\mu)^{(d)}$
may also be defined for negative $d$ in a natural way
\cite{Kleinschmidt:2013em}. 
The above statement may
then be phrased as $\overline R_p = R_{1-d-p}$ for $p=1,\ldots,-d$ in $\mathscr B(\fg,\theta)^{(d)}$ for $d<0$, finite-dimensional $\fg$ and
$\theta$ being the highest root of $\fg$ (so that $R_1$ in $\mathscr B(\fg,\theta)$ is the adjoint of $\fg$ and $\fg^{(1)}=\fg^+$
is the untwisted affine extension of $\fg$).

\vskip6pt
\section{Koszul duals to pure spinor algebras\label{sec:KoszulDual}}
\vskip8pt
\subsection{\texorpdfstring{The $L_\infty$ Koszul dual to $S$}{The Linfinity Koszul dual to S}}
\hfill
\nopagebreak

In this section, we always work with respect to the consistent (totalised) grading on~$C^\bu(\ft)$, thus regarding $R$ as concentrated in total degree zero. Unless otherwise specified, cohomological degree will mean cohomological degree in the Tate resolution, \ie, consistent degree. We will, however, also make use of the weight (internal) grading; recall from~\S\ref{sec:bideg} above that these together define the ``Tate'' bigrading.

\numpar[sec:commlie][Commutative/Lie Koszul duality and multiplicative resolutions]
In~\S\ref{sec:Koszul} above, we have recalled that a graded Lie algebra is Koszul dual to the commutative dg algebra consisting of its Chevalley--Eilenberg cochains. This duality extends to a duality between (the free resolutions of) the operads controlling algebra structures of this type, so that homotopy commutative algebras are dual to homotopy Lie ($L_\infty$) algebras. We have also seen how this works in practice: given an $L_\infty$ algebra, one simply examines its cochains; to go in reverse, one considers freely generated commutative dg algebras, and observes that the differential is dual to an $L_\infty$ structure on the dual of the space of generators, shifted by one. All higher Jacobi identities are contained in the requirement that the differential square to zero.

A generic commutative algebra is, of course, not freely generated, and one may ask how to understand its Koszul dual. The key observation is that any such algebra can be resolved by a freely generated commutative dg algebra, for which the Koszul dual is defined by the above procedure. For algebras $S$ that are quotients of a commutative ring $R$, one can construct a free resolution of~$S$ by a freely generated differential $R$-algebra by iteratively adjoining new generators to kill cohomology in each degree~\cite{Tate}; this is known as the \emph{Tate resolution}.

If $R$ is a polynomial ring, the Tate resolution will also be a free minimal cdga over~$k$, and will therefore be the Lie algebra cochains of an $L_\infty$ algebra with zero differential, defined on the dual of the set of generators. Since the differential is zero, this $L_\infty$ algebra is the minimal model of the corresponding $L_\infty$ structure. (Note, however, that the Tate resolution need \emph{not} be a minimal resolution in $R$-modules; the resolution differential can contain terms of order zero in $\lambda$.) 
We are interested in understanding this $L_\infty$ algebra. In general, the Tate resolution will not be finitely generated as an~$R$-algebra, so that its Koszul dual will not be finite-dimensional; it will, however, be of finite dimension in each graded piece. Because we use cohomological grading conventions (so that the differential has consistent degree $+1$), the Tate resolution, like any projective resolution, is concentrated in non-positive degrees.

\numpar We have seen that $H^0(\ft) = S = \O(Y)$, and that $C^\bu(\ft)$ is a freely generated commutative differential graded $R$-algebra in non-positive degrees. If (and only if) $Y$ is a complete intersection, $C^\bu(\ft)$ is a multiplicative free resolution of $S$ over~$R$. 
In general, though, $C^\bu(\ft)$ will not define a resolution, due to the presence of higher cohomology. 
In a certain sense, the intuition that theories with more supersymmetry are \say{easier to deal with} is related to the fact that $Y$ is necessarily a complete intersection if a sufficient number of supercharges are present. Theories such as ten-dimensional Yang--Mills theory and eleven-dimensional supergravity occupy a fertile middle ground, where $Y$ fails to be a complete intersection---and is therefore algebrogeometrically more interesting---but still has certain good algebraic properties. For example, in both examples mentioned, $S$ is a Gorenstein ring.
We comment more on algebraic properties of the relevant rings in~\S\ref{sec:Properties} below.

\numpar[sec:tateres][The Tate resolution in detail]
The Tate resolution is constructed iteratively, by killing cohomology degree by degree. To resolve the quotient $S = R/\langle \lambda^2  \rangle$ of a polynomial ring $R$ over $R$, we begin with $R$, and then construct a freely generated cdga by adding (anticommuting) generators of degree $-1$ that are mapped by the differential to the generators of the ideal. (Since the differential is a derivation, it is uniquely specified by giving its action on generators.) Note that, in the examples at hand here, the resulting cdga is just $C^\bu(\ft)$, as recalled above. The degree in the Tate resolution is to be identified with the \emph{consistent} degree from~\S\ref{sec:gradings}; we will discuss the bigrading further in what follows. In physical terms, what we are doing is to create the BRST differential for the constraint on
\nolinebreak $\lambda$.

It may then be that $H^{-1}(\ft)$ is non-zero.
When this occurs, we go on to choose a set of generators for the module $H^{-1}(\ft)$, and adjoin new generators to the cdga in degree $-2$ that are mapped by the differential to the generators of $H^{-1}(\ft)$. We thus obtain a new cdga which has no cohomology in degree $-1$. The Tate resolution is constructed by successively adding new generators in this way to kill all cohomology in negative degrees, producing a multiplicative free resolution of $S$ over~$R$ in the limit.

We will number the stages of the Tate resolution so that, at stage $k$, generators in (consistent) degree 
$-k$ are added, resulting in a freely generated cdga whose cohomology is correct in degrees $0 \geq d > -k$. Thus, the first stage is $C^\bu(\ft)$, the second stage adds generators to kill $H^{-1}(\ft)$, and so on. ($R$ itself is stage zero of the resolution.)
 
Both $R$ and $S$ are graded by weight, hence the Tate resolution is bigraded by the consistent degree and the weight degree, with parity determined by the consistent degree. The differential has Tate bidegree $\bid{1,0}$. Thus the resolution will be a weight-wise resolution of each weighted piece of $S$. The generators adjoined at stage one will have weight two (as they kill quadratic constraints). At higher stages the nontrivial cohomology classes will in general be distributed among different weight degrees. $H^{-1}$ will decompose with respect to the weight as $H^{-1}=\bigoplus_{q=3}^\infty H^{\bid{-1,q}}$, where each non-zero $H^{\bid{-1,q}}=H^{q-1,-q}$ is spanned by monomials with one generator in degree $\bid{1,2}$ and $q-2$ generators of $R$.  At each stage $k$ we adjoin generators of degree $-k$ with weight equal to the weight of the corresponding nontrivial cohomology class $H^{\bid{-k+1,q}}$. There is a lower bound to the weight degree at each stage; strictly greater than the minimum of the weights of the generators introduced at stage $k-1$. This is because at each stages the non-zero cohomology will always contain at least one generator from the last stage~\cite{Henneaux1994QuantizationSystems}. The minimal thing one could do would be to have one such generator and one generator of $R$. We denote the minimum of the weight of generators at stage $k$ by $\omega_k$ (set $\omega_{-1}=0$). 

\numpar[sec:xtn][Iterated central extensions.]
Each stage of the Tate resolution defines a new freely generated cdga, which is a relative Sullivan algebra over the cdga of the previous stage. We can apply Koszul duality to each stage independently to get an $L_\infty$ structure at each stage; each stage is then a central extension of the previous one.

Let us denote the space of generators introduced at stage $k$ by $\Wa_{k}$, which is a direct sum over the weights:
\begin{equation}
    \Wa_k=\underset{q>\omega_{k-1}}{\bigoplus}\Wa^{\bid{-k,q}}\;.
\end{equation}
At stage $k$, the complete set of generators is denoted $\Wa^{(k)}=\bigoplus_{n=0}^k\Wa_n$, with $\Wa^{(0)}=\overline{\ft_1}[-1]$. The complex is then the symmetric algebra $\Sym^\bu(\Wa^{(k)})$, which is the cochains of an $L_\infty$ algebra $\ft[k]=\overline{\Wa^{(k)}}[-1]$. Thus, $\ft^{(0)}$ is the abelian odd Lie algebra $\ft_1$, $\ft^{(1)} = \ft$, and $\ft^{(k)}$ is an $L_\infty$ algebra supported in consistent (totalised) degrees $1\leq d \leq k+1$. We also define $\ft^{k+1}\coloneqq\overline{\Wa_{k}}[-1]$.

At stage $k$, we thus have the sequence
\deq{
    k \rightarrow C^\bu(\ft^{(k-1)}) \rightarrow C^\bu(\ft^{(k)}) \rightarrow C^\bu(\ft^{k+1}) \rightarrow k
    }
  of cdgas,
which arises from the dual short exact sequence of $L_\infty$ algebras
        \deq{
        0 \rightarrow \ft^{k+1} \rightarrow \ft^{(k)} \rightarrow \ft^{(k-1)} \rightarrow 0\;.
    }
    This sequence witnesses $\ft^{(k)}$ as an $L_\infty$ central extension of $\ft^{(k-1)}$ by $\ft^{k+1}$. At the end of the day, we have an inverse system
    \begin{equation}
        \begin{tikzcd}
            & & \widetilde{\ft} \ar[dl] \ar[d] \ar[dr] & & \\
            \cdots \ar[r] & \ft^{(k+1)} \ar[r] & \ft^{(k)} \ar[r] &  \ft^{(k-1)} \ar[r] &  \cdots
        \end{tikzcd}
  \end{equation}
  of which $\widetilde{\ft}$ is the inverse limit $\varprojlim \ft^{(k)}$.

We recall that $\ft$ has an internal and a cohomological degree. One might  ask what the internal and cohomological degrees of these new generators are.
We recall the notation from above that $\ft_i$ for $i=1,2$, is in cohomological and internal degree $(0,i)$, respectively.
The internal degree is minus the weight degree, and the cohomological degree is the sum of the consistent degree and the weight degree. The $k$-brackets, which have Tate bidegree $\bid{2-k,0}$ will hence have bidegree $(2-k,0)$.
From this construction, it is clear that $\widetilde{\ft}$ maps canonically to any stage $\ft^{(k)}$, and in particular to~$\ft$. This quotient map determines a short exact sequence
\deq[eq:def-n]{
0 \to \fn \to \widetilde{\ft} \to \ft \to 0,
}
where $\fn$ is the subalgebra of $\widetilde{\ft}$ in internal degrees greater than or equal to three. An analogous sequence
\deq[eq:def-n(k)]{
0 \to \fn^{(k)} \to \widetilde{\ft} \to \ft^{(k)} \to 0
}
can be defined for any stage of the Tate resolution by taking $\fn^{(k)}$ to be the kernel of the relevant map.
In the dual picture, the cochains of any $\ft^{(k)}$ (and in particular, $C^\bu(\ft)$ itself) map to the Tate resolution. 
This observation will be important in what follows.

\begin{center}
\begin{table}[ht]
\begin{align*}
\xymatrix@!0@C=2.247cm@R=1.2cm{
 \ar@{-}[]+<0.8cm,-1.5em>;[ddddddd]+<0.8cm,-1em> 
 &&&&&&&&
\\
\cdots
&&&&&\cdots
&*+[F-:blue][blue]{ \text{stage $0$}}\ar@{-}@[blue][dl] 
\\5 &&&&S_4\ar@{<-}@[red][d] &\cdots&*+[F-:blue][blue]{ \text{stage $1$}}\ar@{-}@[blue][dl]
\\4 &&&
S_3\ar@{<-}@[red][d] &\vee^4R_1\ar@{-}@[blue][ur]\ar@{<-}@[red][d] &\cdots\ar@{-}@[blue][dl]&*+[F-:blue][blue]{\text{stage $2$}} \\
3 &          &           
S_2\ar@{<-}@[red][d]  
&\vee^3R_1\ar@{-}@[blue][ur]\ar@{<-}@[red][d] &\vee^2 R_1\otimes R_2\ar@{<-}@[red][d] &\cdots\ar@{-}@[blue][ur]&*+[F-:blue][blue]{\text{stage $3$}}\\
2&                 
 S_1 \ar@{<-}@[red][d]  
 & \vee^2R_1\ar@{-}@[blue][ur]\ar@{<-}@[red][d]  &R_1 \otimes R_2
\ar@{-}@[blue][ur]\ar@{<-}@[red][d] &{\begin{matrix}R_1\otimes R_3\,\oplus\\\wedge^2R_2\oplus R_{4(2)}\end{matrix}}\ar@{<-}@[red][d]  \ar@{-}@[blue][ur]&\cdots\ar@{-}@[blue][ur] &*+[F-:blue][blue]{\text{stage $4$}}\\
1& 
R_1 \ar@{-}@[blue][ur] & R_2\ar@{-}@[blue][ur] &  \ar@{-}@[blue][ur]R_3& R_{4(1)}\ar@{-}@[blue][ur]&\cdots\ar@{-}@[blue][ur]&\cdots  \\
\ar@{-}[]+<-0.25cm,1.5em>;[rrrrr]+<0.8cm,1.5em>
&1&2&3&4&\cdots&
&\\
}
\end{align*}
\caption{\it The Tate resolution depicted in the bigrading. Weight degree is on the horizontal axis and cohomological degree on the vertical. The arrows represent the differential. The modules of $R$ are decomposed into modules of $\Der(\ft)$ where $R_{i(j)}$ is the module appearing at
weight degree $i$ and cohomological degree $j$. However, since there is never a module $R_{i(j)}$ with $i=1,2,3$ and $j\neq1$, we omit the index
$(1)$ in these cases.\label{table:tatebigrading}}
\end{table}
\end{center}

\newpage
\subsection{Koszulity, minimal orbits, and special classes of rings}

\numpar[sec:QK][Quadratic algebras and quadratic Koszul duality.] Since $S$ is a quadratic algebra, it fits into the framework of Koszul duality theory as originally developed for quadratic algebras~\cite{Priddy}. Let $\Lambda = \overline{\ft_1}[-1]$ be the space of generators $\lambda$ of $R$, and $T(\Lambda)$ denote the tensor algebra on this vector space. We define $I=\dCE \;\overline{\ft_2}[-1]\subset\vee^2\Lambda\subset\Lambda^{\ot 2}$ 
 We observe that 
\deq{
  S = T(\Lambda)/\langle \wedge^2\Lambda\oplus I\rangle
}
is a quadratic presentation for this algebra.
Its \emph{quadratic Koszul dual} is defined in standard fashion to be 
\deq{
  S^\perp = T(\overline{\Lambda}[-1]) / \langle I^\perp[-2]\rangle\;,
}
where $I^\perp \subset (\overline{\Lambda})^{\otimes 2}$, is the space of linear functions on $\Lambda\ot \Lambda$ which vanish on $I\oplus \wedge^2 \Lambda$. As it must vanish on $\wedge^2\Lambda$, it is a subspace of $\vee^2 \overline{\Lambda}$ It consists of all irreducible representations of $\Der(\ft)$ in the symmetric square of $\Lambda$, \emph{other} than the copy of $\ft_2$ determined by the structure constants.

It should be noted that the general definition of quadratic algebras does not force the constraints to contain the full skew symmetric product $\wedge^2\Lambda$. However, as we are interested in the Koszul duality between commutative and Lie algebras, this is the the class we will discuss.

Whenever the constraints cover the full skew symmetric product the Koszul dual algebra is the universal envelope of the free Lie algebra on $\overline{\Lambda}[-1]=\ft_1$ modulo $I^\perp$. Indeed, in our case we have that 
\begin{equation}
    S^\perp= T[\ft_1] / \langle I^\perp[-2]\rangle\simeq U\Big(\frac{F[\ft_1]}{\langle I^\perp[-2]\rangle}\Big)\;.
\end{equation}
In the context of quadratic algebras and Koszul duality, there is a certain class called Koszul algebras. These exhibit different remarkable homological properties. One property, of interest to us, concerns resolutions of the quadratic algebra $S$. If and only if a symmetric quadratic algebra $S$ is Koszul, the Chevalley-Eilenberg cochains of the Lie algebra associated to the quadratic Koszul dual $S^\perp$ provides a multiplicative free resolution of $S$. Thus our notion of Koszul duality coincides with the original notion whenever the algebras in question are Koszul. Indeed, if the quadratic algebra $S$ is Koszul, then $S^\perp$ will be isomorphic to~$U(\widetilde{\ft})$.

\numpar[sec:BorcherdsMinimalDual][The Koszul dual of the coordinate ring of a minimal orbit.]
A special class with quadratic Koszul duality is provided by functions on minimal orbits. Let $\lambda\in S_1$ where $S_1=R(\mu)$ is some highest weight module of $\Der(\ft)$, and $S=\Sym^\bullet(S_1)/\langle R_2\rangle$, where  
$R_2=\vee^2S_1\ominus R(2\mu)$.
Then the $n$-th power of $\lambda\in S_1\subset S$ will be in the representation $S_n=R(n\mu)$, and $S$ is identified as functions on the minimal orbit of $S_1$ under $\Der(\ft)$.
$S$ is Koszul \cite{bezrukavnikov1995koszul}, and the dual Lie superalgebra is the positive part ${\mathscr B}_+$ (with regard to the grading with respect to the grey node) of a BKM superalgebra
constructed from the above data as in
\S\ref{sec:BorcherdsSuperalgebras}
\cite{Cederwall2015SuperalgebrasFunctions}.
This duality is suggested by the partition function equality
\eqref{eq:borcherdspartition}. According to \cite{Froberg1999}, 
the Koszul property of $S$ follows already from the partition function identity, given that $S^\perp=T[-R(-\mu)]/\langle R(-2\mu)\rangle$ 
is the universal enveloping algebra of 
${\mathscr B}_+$.

\numpar[sec:partitionKoszul][Koszul duality at the level of partition functions.]
Identifying $-\overline{R_n}$ as the module of the dual $1$-forms, and 
the coalgebra differential as the BRST operator for the quadratic constraint,
we get (see {\S}\ref{sec:Utpartition}) $\bigoplus_{n=1}^\infty(1-t^n)^{\overline{ R_n}}=Z_S(t)$, \ie,
\begin{align}
Z_S(t)\otimes \overline{Z_{U(\widetilde{\ft})}}(t)={\bf1}\;.
\label{eq:PartitionKoszul}
\end{align} 
(We could have chosen to express partition functions on both sides as series where the power of the formal parameter is the same degree, \eg\ internal degree, and then let the conjugation also include reversion of the degree. This would have lead to less attractive expressions for \eg\ $Z_S$, containing (non-standard) negative powers.)
Note that parity is included in the definition of the sign for $R_n$, unlike in \eqref{eq:borcherdspartition}.
The relation can be interpreted as a denominator formula for $\widetilde{\ft}$. 

\numpar[sec:dualTate][The dual of the Tate resolution]
The Tate resolution, as presently described, is dual to the resolution of $k$ in free modules of $S$, briefly discussed in {\S}\ref{sec:Spartition}.
While the Tate resolution works in the complex $C^\bullet(\tilde \ft )$ and expresses the equality $Z_{C^\bullet(\widetilde{\ft})}(t)=Z_S(t)$, the resolution in {\S}\ref{sec:Spartition} works on a complex $C'$
with $Z_{C'}(t)=(Z_S(t))^{-1}$. Thus, 
$Z_{C^\bullet(\widetilde{\ft})}(t)=(Z_{C'}(t))^{-1}$.
It is well known that one can equip $\Ext_S(k,k)$ with an $A_\infty$ algebra structure, which is the universal enveloping algebra  
of an $L_\infty$ algebra $\mathfrak l$ \cite{Andre1971}.
We observe (see {\S}\ref{sec:Utpartition}) that $\mathfrak l$ can be identified with the $L_\infty$ algebra
$\widetilde{\ft}$.  
The degree adjustments made in {\S}\ref{sec:Spartition} make $\mathfrak l$ and $\widetilde{\ft}$ agree in the bigrading. 

\numpar[sec:MinGorenstein][Minimal orbits and Gorenstein rings]
We recall, for future reference, that 
a Noetherian local ring $R$ with residue field $k$ and Krull dimension $n$ is \emph{Gorenstein} if
\[
\Ext^\bu_R(k,R) = k[-n].
\]
A Noetherian commutative ring is \emph{Gorenstein} when its localization at every maximal ideal is a Gorenstein local ring. 

A Noetherian local ring $R$ with residue field $k$ and Krull dimension $n$ is \emph{Cohen--Macaulay} if and only if 
\[
\Ext^i_R(k,R) = 0 \,,\quad \forall i < n.
\]
(In other words, the depth and the dimension are equal.) A Noetherian commutative ring is \emph{Cohen--Macaulay} when its localization at every maximal ideal is a Cohen--Macaulay local ring. It is clear that every Gorenstein ring is Cohen--Macaulay. Both of these properties are relevant in the context of the pure spinor superfield formalism: via results of~\cite{AG}, the Koszul homology of a local ring is equipped with a perfect pairing if and only if the ring is Gorenstein. The Cohen--Macaulay property is related to the existence of an antifield multiplet that is describable at the level of the non-derived pure spinor formalism. For details, we refer the reader to~\cite{perspectivesonpurespin}.

One family of interesting Gorenstein rings, including some important examples in the context of the pure spinor formalism, are (the coordinate rings of closures of) certain minimal orbits. 
%
By work of Panyushev, the coordinate ring of the closure of the minimal orbit in the highest-weight module~$R(\mu)$ is Gorenstein when $\mu$ is a fundamental weight~\cite[\S3.4, Corollary 2]{Panyushev}. Panyushev further analyzes conditions characterizing ``Gorenstein weights'', and we refer to that paper for more details. One can give an alternative proof of the result for minimal orbits in fundamental highest-weight modules using Serre duality and the Borel--Weil--Bott theorem for $\Proj S$.\footnote{We gratefully thank the anonymous referee for pointing this out.} See also \cite{GKR,HS} for important related work.


We thus have a subclass of rings $S$ which are Koszul, the Koszul duals being BKM algebras, with the additional property that $S$ is Gorenstein.
Gorenstein rings are particularly interesting in the context of pure spinor spinor superfield theory, since they provide a pairing which can be thought of as a Calabi--Yau-form on $Y$ \cite{perspectivesonpurespin,Cederwall:2011yp}. If in addition the codimension is odd, this pairing gives rise to a field-antifield pairing in the Batalin--Vilkovisky (BV) framework~\cite{perspectivesonpurespin}. The examples in {\S}\ref{sec:SYM}, {\S}\ref{sec:SL5} and {\S}\ref{sec:D11} have this property, the latter without being a minimal orbit.

\subsection{The Koszul dual of the canonical multiplet}
\hfill
\nopagebreak

\numpar[sec:Ainf][A resolution of $A^\bu(S)$ by a freely generated cdga.]
Like any pure spinor multiplet, $A^\bu(S)$ is freely resolved as a sheaf over superspace. As we pointed out above in~\S\ref{sec:O}, it also carries a natural commutative algebra structure. However, $A^\bu(S)$ is \emph{not} freely generated as a (graded) commutative algebra, just because $S$ itself is not. 

It is interesting and profitable to further resolve $A^\bu(S)$ by a freely generated cdga. In fact, we have already seen how to do this, at least implicitly. Since we know, by construction, that $C^\bu(\widetilde{\ft})$ resolves $S$, we can apply the pure spinor functor to $C^\bu(\widetilde{\ft})$, thereby considering the cdga
   \deq{ 
        \widetilde{A}^\bu = A^\bu(C^\bu(\widetilde{\ft})) =  \left( C^\infty(T) \otimes_k C^\bu(\widetilde{\ft}), \widetilde{d} = \cD + \dCE \right) .\label{dinftyeq}
    }
    By~\S\ref{sec:xtn}, there is a $C^\bu(\ft)$-module structure on~$C^\bu(\widetilde{\ft})$. We recall that the differential takes the form
    \deq{
        \widetilde{d} = v \pdv{ }{x} + \lambda \left( \pdv{ }{\theta} + \theta \pdv{ }{x}\right) + \lambda^2 \pdv{ }{v} + \cdots,
    }
    where the additional terms come from the brackets at higher levels in $\widetilde{\ft}$.
This construction realises a proposal in \cite{Chesterman:2002ey}. From the first-quantised perspective, our construction produces the complete BRST system for covariant quantisation of the Berkovits superparticle. We expect analogous considerations to be relevant to other pure spinor sigma models, and plan to make this concrete in future work.

\numpar[sec:multipletisCn][The multiplet is $H^\bullet(\fn)$.]

Recall that we have the cofiber diagram 
    \begin{equation}
        \begin{tikzcd}
       C^\bu(\ft) \ar[r] \ar[d] & C^\bu(\widetilde{\ft}) \cong S  \ar[d] \\
       k \ar[r] & C^\bu(\fn)
        \end{tikzcd}
        \label{eq:cofib1}
    \end{equation}
 of cdgas, 
    where $\fn$ is the subalgebra of $\widetilde{\ft}$ defined by~\eqref{eq:def-n}. 
    Recalling from above that the multiplet associated to an algebra $S$ carries an $S$-algebra structure, obtained by embedding $S$ along the constant functions on the superspace $T$, 
    we can obtain a similar diagram by applying the functor $A$ to the top row:
        \begin{equation}
        \begin{tikzcd}
       C^\bu(\ft) \ar[r] \ar[d] & C^\bu(\widetilde{\ft}) \cong S  \ar[d] \\ 
      \Omega^\bu(T) \cong k \ar[r] &  \widetilde{A}^\bu \cong A^\bu(S).
        \end{tikzcd}
        \label{eq:cofib2}
    \end{equation}
    Comparing the two diagrams suggests that one might find an equivalence between $A^\bu(S)$ and $C^\bu(\fn)$, witnessing both as pushouts of the same system.
    In the next sections, we will demonstrate explicitly  that this is in fact the case: 
    the canonical multiplet can be identified with the cochains of~$\fn$, and as such with the homotopy fiber of the map of formal moduli problems $\widetilde{\ft}\to \ft$. 
\begin{theorem}
There is a homotopy equivalence of cdgas between $A^\bu(S)$ and $C^\bu(\fn)$. In particular, the cohomology $\fm$ of the canonical multiplet is given by the Lie algebra cohomology of $\fn$.
\label{thm:main}
\end{theorem}
    \begin{proof}
    The proof proceeds by constructing a roof diagram
    \begin{equation}
        \begin{tikzcd}[row sep = 1 ex]
            & \widetilde{A}^\bu \ar[rd] \ar[ld] &  \\
                A^\bu(S) & & C^\bu(\fn)
        \end{tikzcd}
        \label{eq:roof}
    \end{equation}
of quasi-isomorphisms of cdgas.
In Propositions \ref{prop:ASisAInftyS} and \ref{prop:AInftyisn}
 below, we show that each map in \eqref{eq:roof} is an equivalence, so that $\widetilde{A}^\bullet$ is homotopy equivalent both to the pure spinor multiplet $A^\bullet(S)$, and to the cochains of $\fn\subset\widetilde{\ft}$. The theorem follows immediately.
 \end{proof}

\numpar[sec:equiv][The equivalence of $\widetilde{A}^\bu$ and the multiplet.]
We will here construct the left-hand map in \eqref{eq:roof}, proving the equivalence between $\widetilde{A}^\bu$ and the standard pure spinor multiplet $A^\bu(S)$. 

\begin{prop}\label{prop:ASisAInftyS}
The pure spinor multiplet $A^\bu(S)$ is homotopy equivalent to $\widetilde{A}^\bu$. That is, there exists a deformation retract
\begin{equation*}
    \begin{tikzcd}
(A^\bu(S),\cD) \arrow[r, shift left=1ex, "{i'}"] 
 &  \arrow[l, shift left=1ex, "{p'}"]  (\widetilde{A}^\bu,\widetilde{d}) \arrow[loop right]{r}{h'} 
\end{tikzcd}\,.
\end{equation*}
\end{prop} \begin{proof}
The proof will make use of the homological perturbation lemma.
The differential on $\widetilde{A}^\bu$ splits into
\begin{equation}
    \widetilde{d}=\cD+\dCE\;,
\end{equation}
where $\dCE$ is the differential on $C^\bu(\widetilde{\ft})$, and $\cD=v\pdv{}{x}+\lambda\pdv{}{\theta}+\lambda\theta\pdv{}{x}$. $\cD$ can be seen as a deformation of $\dCE$. As a commutative algebra 
\begin{equation}
    \widetilde{A}^\bu=C^\infty(T)\ot C^\bu(\widetilde{\ft})\;.
\end{equation}
Taking cohomology with respect to $\dCE$ we get 
\begin{equation}
H^\bu\big(C^\infty(T)\ot C^\bu(\widetilde{\ft}),\dCE\big)=C^\infty(T)\otimes S\;,
\end{equation}
which is the underlying algebra of $A^\bu(S)$. It is always possible to construct a deformation retract from any complex to its cohomology \cite{lodval}. 
\begin{equation*}
    \begin{tikzcd}
(A^\bu(S),0) \arrow[r, shift left=1ex, "{i}"] 
 &  \arrow[l, shift left=1ex, "{p}"]  (\widetilde{A}^\bu,\dCE) \arrow[loop right]{r}{h} 
\end{tikzcd}\,.
\end{equation*}
For any consistent degree $k$ we can choose a decomposition $\widetilde{A}^k=B^k\oplus H^k \oplus \dCE^{-1} B^{k+1}$, where $B^k=\dCE \widetilde{A}^{k-1}\subset A^k$, denotes the space of boundaries of $\dCE$. The homotopy $h$ has consistent degree $-1$ and is zero on $ H^k \oplus\dCE^{-1} B^{k+1}$ and identifies $B^{k}$ with $\dCE^{-1}B^{k-1}$ \cite{lodval}. This means that the image of $h$ will always contain elements from stages higher than or equal to 1. Hence, all higher $A_\infty$ products, except for the binary one, on $A^\bu(S)$ are trivial as the product of anything with something which is nontrivial under $h$ will be zero under \nolinebreak $p$.

The deformation $\cD$ of $\dCE$ induces, by homotopy transfer of $D_\infty$ structures \cite{Lapin}, a new homotopy retract.
\begin{equation*}
    \begin{tikzcd}
(A^\bu(S),\cD') \arrow[r, shift left=1ex, "{i'}"] 
 &  \arrow[l, shift left=1ex, "{p'}"]  (\widetilde{A}^\bu,\widetilde{d}) \arrow[loop right]{r}{h'} 
\end{tikzcd}\,.
\end{equation*}
where the new maps are constructed as perturbation series. In particular
 \begin{equation}\begin{split}
\label{eq:perturbed retract2}
    p'=& \;p\circ \sum_{n=0}^\infty(\cD h)^n,\\
    \cD'=&\sum_{n=0}^\infty p (\cD h)^n\cD i\;.
\end{split}
\end{equation}
The induced differential is the ordinary pure spinor differential. The first term is 
\begin{equation}
    p\cD i=p\left(v\pdv{}{x}+\lambda\pdv{}{\theta}+\lambda\theta\pdv{}{x}\right) i=p \left(\lambda\pdv{}{\theta}+\lambda\theta\pdv{}{x}\right) i\;,
\end{equation}
as $v$ is by construction zero under $p$. Thus we retrieve the original pure spinor differential as the first approximation. 

The higher approximations are zero; because of the construction of $h$, we have that the image $\cD i(q(x,\theta,\lambda))$, of a polynomial $q(x,\theta,\lambda)\in A^\bu(S)$, will only be nonzero under $h$ for elements that are exact with respect to $\dCE$. (We assume that $q$ is homogeneous with respect to the splitting chosen above.) If it is exact it will be replaced by some element $w_k$ from a stage $k>0$. The image under $\cD$ of this will either be zero or still contain elements from higher stages than 0. All these are zero under $p$. Thus, all higher corrections to $\cD'$ are zero, and we get that the induced differential is just the standard pure spinor differential. By the same argument, we can also conclude that there will be no higher approximations to $p$, and hence, no higher products induced on $A^\bu(S)$, except for the binary one.

The convergence of the other series (those defining $h'$ and $i'$) presents no problems. We are working in $C^\infty(T) \otimes C^\bu(\widetilde{\ft})$, which can be thought of as the formal series algebra on $\widetilde{\ft}^\vee[-1]$ valued in smooth functions on superspace. $C^\bu(\widetilde{\ft})$ is graded by weight, and $\dCE$ is of weight zero; we can choose the unperturbed homotopy data, and in particular $h$, to also be homogeneous of weight zero. We can extend the weight grading to define a filtration on $C^\infty(T) \otimes C^\bu(\widetilde{\ft})$ by placing functions on superspace in weight zero. The perturbation is compatible with this filtration, and indeed $\cD$ has filtered weight $+1$. Thus the term $h(\cD h)^k$ in the series defining $h'$ has filtered weight $+k$, and only finitely many terms contribute to the action of $h'$ on a homogeneous element at any fixed weight.
\end{proof}

\numpar[sec:n-equiv][$\widetilde{A}^\bu$ is equivalent to $C^\bu(\fn)$.]
We complete the proof of Theorem~\ref{thm:main} by constructing the right-hand map in~\eqref{eq:roof}. To check that this is an equivalence, we will perform a homotopy transfer from $\widetilde{A}^\bu$ to 
$C^\bullet(\fn)$. 
\begin{prop}\label{prop:AInftyisn}
$\widetilde{A}^\bullet$ and $C^\bullet(\fn)$ are homotopy equivalent as cdgas.
\end{prop}
\begin{proof}
Let $\widetilde{d}=d_0+d_1$, where $d_0=\lambda\frac{\*}{\*\theta}+v\frac{\*}{\*x}$, with $d_0^2=0$, is the part of the differential encoding the $1$-bracket described by the arrows in Table \ref{fig:dinftyalgebra}, and $d_1$ the rest. $d_1$ is not nilpotent, but for the moment seen as a deformation of $d_0$.
The transfer will then be from $(C^\infty(T)\otimes C^\bullet(\ft)\otimes C^\bullet(\fn),d_0)$ to $(C^\bullet(\fn),0)$ (the cohomology of $d_0$). 

We need an operator $h$ which is the ``inverse'' or adjoint of $d_0$, $h\sim \theta\frac{\*}{\*\lambda}+x\frac{\*}{\*v}$, maybe with some factor depending on the degree of homogeneity, which is irrelevant for the argument. 
The projection and inclusion of this first little transfer are simply the naive ones,
\begin{align}
i:&\quad\phi\mapsto\phi\;,\nn\\
p:&\quad\phi\mapsto\phi\;,\quad(\lambda,v)\mapsto0
\;,\quad(\theta,x)\mapsto0\;,
\end{align}
where $\phi\in\fn$ is denoted by the same letter on both sides. This defines a strong deformation retract
\begin{equation}
    \begin{tikzcd}
(C^\bu(\fn),0) \arrow[r, shift left=1ex, "{i}"] 
 &  \arrow[l, shift left=1ex, "{p}"]  (\widetilde{A}^\bu,d_0) \arrow[loop right]{l}{h}
\end{tikzcd}\,.
\end{equation}
The transferred $A_\infty$ structure on $C^\bu(\fn)$ will only be the binary one induced by the product on $\widetilde{A}^\bu$. All higher products will vanish as the product of anything with something nontrivial under $h$ will be zero under $p$. Furthermore, following the result of  \cite{Lapin},
the retract allows us to, for any deformation of $d_0$, transfer the total differential to $C^\bu(\fn)$, so that we get another retract

\begin{equation}
    \begin{tikzcd}
(C^\bu(\fn),d') \arrow[r, shift left=1ex, "{i'}"] 
 &  \arrow[l, shift left=1ex, "{p'}"]  (\widetilde{A}^\bu,d_0+d_1) \arrow[loop right]{l}{h'}
\end{tikzcd}.
\end{equation}

The actual projection, inclusion, $h$ and differential 
are then constructed 
as a perturbation series.  The $A_\infty$ structure on $C^\bu(\fn)$ will again only be the binary one, by the same argument as in proposition \ref{prop:ASisAInftyS}.

The new differential $d'$ is given as a sum of terms 
\begin{align}
d'=\sum_{n=0}^\infty p (d_1 h)^n d_1 i\;.
\end{align}
The term at $n=0$ gives back the differential on $C^\bullet(\fn)$. For $n>0$, $h$ picks up terms with $\lambda$ or $v$ and converts them to $\theta$ or $x$. These variables are untouched by remaining $d_1$’s (except for the term that state that they transform under supertranslations), and any such terms finally vanish under $p$.
In other words, there is no zig-zag landing in $C^\bullet(\fn)$, since $d_1$ goes out of $C^\bullet(\fn)$ but never into it from outside (this applies to the whole complex); $\fn$ is a subalgebra.
\end{proof}

\numpar[sec:A-inf-shriek][The multiplet as a homotopy fiber]
While we have formulated the results of the above section in terms of cdgas, we can of course also work with the Koszul dual $L_\infty$ algebras. The diagram dual to~\eqref{eq:cofib1} is
\begin{equation}
    \begin{tikzcd}
    \ft & \widetilde{\ft} \ar[l] \\
    0 \ar[u] & \fn \ar[l] \ar[u]
    \end{tikzcd}
    \label{eq:fib1}
\end{equation}
The diagram Koszul dual to~\eqref{eq:cofib2} is
\begin{equation}
    \begin{tikzcd}
        \ft & \widetilde{\ft} \ar[l] \\
    0 \ar[u] & (\widetilde{A}^\bu)^! \ar[l] \ar[u]
    \end{tikzcd}
    \label{eq:fib2}
\end{equation}
As already demonstrated above, $\fn$ is (a model of) the Koszul dual to $\widetilde{A}^\bu$. But it is interesting to observe that $\widetilde{A}^\bu$ is a freely generated cdga, so that there is an apparent model for its Koszul dual, defined on the linear dual of its generating set. This model $(\widetilde{A}^\bu)^!$ is, in fact, isomorphic (as a cochain complex) to the (suspended) mapping cone of the canonical morphism $\widetilde{\ft} \to \ft$; the standard superspace coordinates can be interpreted as the generators of the Chevalley--Eilenberg complex coming from the shifted copy of~$\ft$. We sketch the structure of the resulting algebra in Table~\ref{fig:dinftyalgebra}. The suspended mapping cone is a well-known model for the homotopy fiber; the relevant $L_\infty$ structure was studied in~\cite{FM}. 

\begin{table}
    \centering
\begin{picture}(250,150)(0,-30)
\put(-12,0){hom. degree $=$}
\put(80,0){\llap{$-$}$1$}
\put(80,30){$0$}
\put(80,60){$1$}
\put(50,-30){int. degree $=$}
\put(120,-30){$1$}
\put(160,-30){$2$}
\put(200,-30){$3$}
\put(240,-30){$\cdots$}
\put(80,90){$\vdots$}
\put(100,-15){\line(0,1){120}}
\put(100,-15){\line(1,0){160}}
\put(120,30){$\ft_1$}
\put(160,30){$\ft_2$}
\put(120,0){$\ft'_1$}
\put(160,0){$\ft'_2$}
\put(180,15){\line(0,1){90}}
\put(180,15){\line(1,0){80}}
\put(215,60){$\fn$}
\multiput(123,25)(40,0){2}{\vector(0,-1){15}}
\end{picture}
    \caption{\it Generators in double grading of the $L_\infty$ algebra 
    $(\widetilde{A}^\bullet)^!$ with coalgebra differential $\widetilde{d}$. The arrows indicate the only $1$-brackets.}
    \label{fig:dinftyalgebra}
\end{table}
This gives us a pleasing interpretation of the geometric origin of the canonical multiplet. If we think of it as the affine dg scheme $\Spec C^\bu(\fn)$, it is the homotopy fiber of the map from $\Spec S$ to $\Spec C^\bu(\ft)$---that is, the map from the ordinary pure spinor space to its derived replacement.

We note that there is a map
\deq{
A^\bu(S) \to C^\infty\left(T\right)
}
from the canonical multiplet to functions on superspace. The spectrum of the canonical multiplet is thus, in a precise sense, also an extension of the normal notion of superspace: just as the body of any supermanifold maps to the supermanifold, the normal superspace $T$ maps to $\Spec A^\bu(S)$. Any non-derived pure spinor superfield can be thought of as a sheaf over~$\Spec A^\bu(S)$, and it is the geometry of this space that underlies many constructions in pure spinor superfield theories.

(The reader may be concerned about the use of smooth functions on superspace in our context, as constructions in homotopy transfer often involve formal series. However, as remarked in the proofs of Propositions~\ref{prop:ASisAInftyS} and~\ref{prop:AInftyisn}, the relevant series actually terminate in our examples, and no series involving the coordinate $x$ appear. Indeed, the essential ingredient of our construction is nothing other than the Poincar\'e lemma in disguise: the sheaf of smooth de Rham forms is equivalent to the sheaf of locally constant functions.)

\numpar[sec:LHS][Filtrations and the Hochschild--Serre spectral sequence] 
Recall the short exact sequence~\eqref{eq:def-n(k)} of $L_\infty$ algebras from above:
\deq[eq:SES]{
    0 \to \fn^{(k)} \to \widetilde{\ft} \to \ft^{(k)} \to 0,
}
where $\fn^{(k)}$ is the ideal subalgebra of $\widetilde{\ft}$ consisting of all generators added after stage $k$. (In particular, $\fn = \fn^{(1)}$.) We can consider the two-step filtration defined by assigning weight zero to $\fn^{(k)}$ and weight one to the remaining  generators. (We remark that this is \emph{not} a grading on $\widetilde{\ft}$---but it is a subfiltration of the filtration defined by the consistent grading.)
The fact that the quotient algebra for $k=1$ is precisely the supertranslation algebra $\ft$ with which we started will be significant in the sequel, in particular due to the fact that we will be interested in considering the $\ft$-multiplet associated to $C^\bu(\widetilde{\ft})$. However, we will see examples where higher values of~$k$ seem to appear naturally (see~\S\ref{sec:G2A1}).

To any such short exact sequence of Lie algebras, one can associate a Hochschild--Serre spectral sequence which abuts to the Lie algebra cohomology of $\widetilde{\ft}$. This spectral sequence is constructed by filtering $C^\bu(\widetilde{\ft})$ as induced from the two-step filtration $\fn^{(k)} \subset \widetilde{\ft}$. In the case $k=1$, the $E_2$ page of this spectral sequence is
\deq{
    E_2 \, : \, H^\bu(\ft, H^\bu(\fn)) \qquad \Rightarrow \qquad E_\infty = H^\bu(\widetilde{\ft} ) \simeq S\;. 
    }
 But we know that derived $\ft$-invariants are the inverse functor to the pure spinor functor $A^\bu(-)$! Therefore, the $E_2$ page of the spectral sequence is identified with the Koszul dual of the $U(\ft)$-module $H^\bu(\fn)$. 
    But we know---via the identification of $H^\bu(\fn)$ with $A^\bu(S)$ above---that this Koszul dual is $S$. Our result above can  thus be understood as the assertion that the Hochschild--Serre spectral sequence collapses at $E_2$ in this instance.

\subsection{\texorpdfstring{From $\ft\inplus F[\overline{\fm}]$ to $\widetilde{\ft}$} {tinfty from t F bar m}\label{sec:freeconstruction}}
\hfill
\nopagebreak

The purpose of this section is to obtain as much concrete information about the resolution $\widetilde{\ft}$ as possible, including explicit forms of its brackets.

In general one cannot expect to arrive at a point where one has complete understanding of what $\widetilde{\ft}$ looks like, and what its brackets are. We will see however, by establishing some idea of how the multiplet ``fits'' into $\fn$ (to be more precisely defined), that it is in some cases possible to get a complete understanding of what the algebra is. One should again mention the class of Koszul algebras, if the ring $S$ is found to be $S$ Koszul, then the ``traditional''  Koszul dual algebra will be the universal envelope of the Lie algebra $\widetilde{\ft}$.

Another class of models, partly overlapping with the one above, is characterised by the property that the subalgebra $\fn$ will be freely generated by the shifted dual of the multiplet. 
All models describing (canonical multiplets of) standard supersymmetry fall in this class.
This is the topic of the present section.
We will soon give conditions for when $\fn=F[\overline{\fm}[-1]]$, which we will write $F[\overline{\fm}]$ for compactness, the shift in cohomological degree being understood.

\numpar[sec:nfree][$\fn$ is freely generated in a large class of models.]

We now wish to establish a result which tells us that $\fn$ in many cases is in fact a strict Lie algebra (\ie, with only a 2-bracket), and moreover, that it is freely generated by the dual of the multiplet. 

We note that $\fn$ is a minimal (\ie, without 1-bracket) $L_\infty$ algebra in (internal) degrees greater than or equal to three. 
We can filter $C^\bu(\fn)$ by assigning filtration degree one to the generators $\overline{\fn}$. This assignment determines a grading on~$C^\bu(\fn)$ which is not preserved by the differential. However, the associated decreasing filtration is. Indeed, since $\fn$ is minimal, the differential \emph{strictly} increases the filtration degree. 

We will make use of the following lemma.
\begin{lemma}\label{lemma:generatingset}
Let $\fa$ be a minimal $L_\infty$ algebra in positive internal degrees. Then there is a well-defined subspace $H^{(1)}(\fa)$ of $H^\bu(\fa)$, defined by having representatives linear in $\overline{\fa}[-1]$. $H^{(1)}(\fa)$ is dual to a minimal generating set of $\fa$.
\end{lemma}
\begin{proof}
Call the space of linear cochains $C^{(1)}(\fa) \simeq \overline{\fa}[-1]$.
No element in $C^{(1)}(\fa)$ is exact, since the 1-bracket vanishes.
Thus there is a subspace $H^{(1)}(\fa) \subset H^\bu(\fa)$ consisting of linear cohomology (represented by classes in~$C^{(1)}$).
Furthermore, $H^{(1)}(\fa) \subset \overline{\fa}[-1]$ is a subspace.
Closedness means that the dual generator does not appear as the bracket of any generators; 
in fact, the dual map $\fa \to \overline{H^{(1)}(\fa)}[-1]$ identifies the shifted dual of $H^{(1)}$ with $\fa/[\fa,\fa]$, and thus (noncanonically) with a minimal generating set of~$\fa$.
\end{proof} 

In fact, 
positive internal grading together with the requirement of equivariance under $\Der(\ft)$ ensures that the minimal generating set can be chosen uniquely, at least in our examples.
The cochains dual to elements of lowest internal degree necessarily belong to $H^{(1)}(\fa)$. Reasoning by increasing internal degree sequentially identifies $H^{(1)}(\fa)[1]$ as dual to the set of generators that do not arise as brackets of generators at lower degrees; Schur's lemma identifies these uniquely with particular $\Der(\ft)$-modules.\footnote{One might worry that $\Der(\ft)$ could be too small to ensure splittings of this kind. In particular, it could be abelian. In every example we study in the sequel, $\Der(\ft)$ is a nonabelian reductive Lie algebra, and no issues arise.}
Thus, $H^{(1)}(\fa)$ is dual to (a shift of) this minimal generating set of $\fa$. In particular, we have a wrong-way map that identifies $\overline{H^{(1)}(\fa)}[-1]$ with a subspace of~$\fa$.
   
Given a minimal generating set $V \subset \fa$, we can introduce a new filtration on the $L_\infty$ algebra $\fa$ itself. This filtration is by ``generation number'' $\gamma$, 
defined by assigning $\gamma=1$ to the generating set $V=\overline{H^{(1)}(\fa)}[-1]$. 
($\overline{V}[-1]$ thus has $\gamma=-1$.)
We extend the definition by requiring that any 2-bracket carries generation number 0;
this would be a grading in the strict case, but it can \apriori\ be violated by higher brackets.

\begin{prop}\label{prop:nfree}
        Let $\fa$ be a minimal $L_\infty$ algebra in positive internal degrees. Then $\fa$ is a freely generated (strict) Lie algebra if and only if $H^\bullet(\fa)=k\oplus H^{(1)}(\fa)$. 
\end{prop}
\begin{proof}
        One direction is immediate. Indeed, the  
        cohomology of any free Lie algebra $F[V]$ is 
    \begin{equation}
        H^{(n)}(F[V])=\begin{cases} 
        k,\; &n=0\\
        \overline{V}[-1],\;&n=1\\
        0,\; &n\geq 2.
        \end{cases}
    \end{equation}
    (This follows by observing that Lie algebra cohomology is the cohomology of the algebra $U(\fa)$, which---as an associative algebra---does not depend on the grading; the enveloping algebra of a graded Lie algebra is just the tensor algebra $T[V]$, equipped with a grading that can be forgotten.)
    The other direction follows from the observation that $C^{\bu}(\fa)$ provides a free multiplicative resolution of the ring $k\oplus \overline{V}[-1]= \Sym^\bu(\overline{V}[-1])/\Sym^2(\overline{V}[-1])$, where 
    $\overline{V}[-1]=H^{(1)}(\fa)$.
   The ring $k\oplus \overline{V}[-1]$ is Koszul, and its Koszul dual Lie algebra, the cochains of which provide a free multiplicative resolution (the Tate resolution) is just $F[V]$ (see {\S}\ref{sec:trivex}, item 
   (\ref{item:free})). 
   Minimal free multiplicative resolutions are unique,
   hence $\fa\cong F[V]$. 
\end{proof}

An immediate corollary is:
\begin{cor}
If $\overline{ M}[-1]$ is a generating set of $\fn$ (\ie, $\overline{M}[-1]=H^{(1)}(\fn)$), then
$\fn$ is the freely generated (strict) Lie algebra $F[\overline{M}[-1]]$, and $\widetilde{\ft}$ is freely generated from internal degree 3.
\end{cor}

Note that the homological degrees of the elements in $\fn$ do not need to be zero, even if $\fn$ is a strict Lie algebra, which is seen in many of the examples in {\S}\ref{sec:examples}.
If $\overline{\fm}$ resides in a single homological degree, generation is proportional to homological degree. The concept of generation is strictly needed only when this is not true, which is typically the case in supergravity models due to the presence of super-Killing vector cohomology in the ghost sector. See \eg\ the example of $D=11$ supergravity in 
\S\ref{sec:D11}.

For the sake of testing the limits, we will also give an example, \S\ref{sec:E6}, where the Tate resolution does not give a $\widetilde{\ft}$ which is freely generated from degree $3$; the cohomology will then also have ``unphysical'' support in negative ghost number.
This example violates the assumptions in Proposition \ref{prop:nfree} in the following way. $S$ is Koszul and hence $\widetilde{\ft}$ is a strict Lie algebra. The multiplet has support in $H^1(\fn)$ and $H^2(\fn)$. Since $\fn$ is a strict Lie superalgebra, only the dual of $H^1(\fn)$ can be a generating set. Relations for brackets on these generators are responsible for the appearance of $H^2(\fn)$.

\numpar[sec:partitionmultiplet][Extracting the multiplet from $Z_S(t)$.]
Useful information can be extracted from the partition function by discarding the first or the first two factors in the product form
\eqref{eq:ZUa}. This is because they cancel against the partition functions for $\theta$ and $x$, see {\S}\ref{sec:n-equiv}. 

Extracting the first level in $Z_S$ as
$Z_S(t)=(1-t)^{\overline{ R_1}}P(t)$ (note that $R_1$ is negative) provides a minimal free resolution, an additive resolution in terms of functions of an unconstrained spinor, \ie, in $R$.
Such resolutions are always finite by Hilbert's syzygy theorem 
\cite{Hilbert}. The numerator is thus
polynomial, and is read as the zero-mode cohomology, \ie, the cohomology of $\lambda\frac{\*}{\*\theta}$. The interpretation is a list of component fields (including all ghosts, antifields etc.), which is the origin of the tables with arrows in the examples.
The arrows indicate the possible action of 
bosonic derivatives
after homotopy transfer to component fields 
\cite{perspectivesonpurespin}.

One can go one step further and consider the partition function for the canonical multiplet $A^\bu(S)$;
\begin{equation}
    Z_{A^\bu(S)}(t)=Z_{C^\infty(T)}(t)\ot Z_{S}(t).
\end{equation}
(Of course, $C^\infty(T)$ does not have a well-defined partition function; to reason in this way, we will have to work formally, replacing smooth functions on the body of~$T$ by formal power series $k[\![x]\!]$.)
We know that the partition function of $S$ is the same as the partition function for the Tate resolution of $S$ (this is a general fact about partition functions/characteristics; they are the same for complexes and their cohomology), so we can replace $Z_S(t)$ with $Z_{C^\bu(\widetilde{\ft})}(t)$
\begin{equation}
\begin{split}
    Z_{A^\bu(S)}(t)=&Z_{C^\infty(T)}(t)\ot Z_{C^\bu(\widetilde{\ft})}(t)=Z_{\widetilde{A}^\bu}(t)\\
    =&(1-t)^{-\overline{R}_1}(1-t^2)^{-\overline{R}_2}(1-t)^{\overline{R}_1}(1-t^2)^{\overline{R}_2}\ot Z_{C^\bu(\fn)}(t),    
\end{split}
\end{equation}
where $R_1$ and $R_2$ are the $\Der(\ft)$-representations of $\ft_1$ and $\ft_2$, respectively. On the right hand side we are left with the partition function for the cochains of $\fn$. We can make a further substitution, and replace $Z_{A^\bu(S)}(t)$ with the partition function for its cohomology. But the cohomology of $A^\bu(S)$ is the physical multiplet $\fm$, together with a copy of the constants in degree 0. We end up with the equation\footnote{The minus sign in front of $Z_M(t)$ is arbitrarily chosen so that it will be cancelled by the parity of a fermionic pure spinor superfield $\Psi$, which is the case at hand when the constant mode of $\Psi$ represents the highest ghost of a connection which is an odd form, as in $D=10$ super-Yang--Mills theory and $D=11$ supergravity. If the connection on which the superfield is based is an even form, the multiplet will come out with reversed parity. This just means that the pure spinor superfield in those cases should be chosen to be bosonic.}
\begin{equation}
    1-Z_{\fm}(t)=Z_{C^\bu(\fn)}(t).
\end{equation}
Hence, extracting both level $1$ and level $2$ from $Z_S(t)$, will reveal the partition function for the multiplet $Z_S(t)=(1-t)^{\overline{ R_1}}(1-t^2)^{\overline{ R_2}}(1-Z_\fm(t))$.
(The initial $1$ represents constant cohomology, which by our conventions is not part of the multiplet.)
Note that this implies
\begin{align}
Z_{U(\fn)}(t)=\frac1{1-Z_{\overline{\fm}}(t)}\;,
\end{align}
which agrees with the partition function $Z_{T[\overline{\fm}]}(t)$ of the tensor algebra of $\overline{\fm}$, which is the universal enveloping algebra of the freely generated algebra on $\overline{\fm}$, $U(F[V])=T[V]$.
Note, however, that this does not show that $\fn=F[\overline{\fm}]$, since the partition function forgets about cohomological degree.

\numpar[sec:deformationprocedure][The procedure.] We now restrict to the case when we know $\fn=F[\overline{\fm}]$. We will demonstrate how $\widetilde{\ft}$ can be obtained by starting with the semidirect sum $\ft \inplus F[\overline{\fm}]$, and then introducing brackets in order to kill the unwanted cohomology. We use the fact that we know $\widetilde{\ft}$ as a vector space,
and rely heavily on the existence of the resolution.
We also use important pieces of the brackets: the supertranslation algebra, the transformation of $\overline{\fm}$ as a module of supertranslations, which lifts by tensor product to $F[\overline{\fm}]$, and the Lie superalgebra brackets of $F[\overline{\fm}]$.
In this way we keep track of entire $\ft$-modules. 
Concretely, if a $\ft$-cocycle $\eta$ appears in the Tate resolution, 
the module structure tells us that $\eta$ will be the part of a $\ft$-module $m_\eta$, of which $\eta$ is the part of lowest weight degree,
and the $\eta$ can be extended to a cocycle taking values in this module.
In the examples of \S\ref{sec:examples}, concrete expressions are given for various (low) degrees. This does not present a problem of existence, since we know that the Tate resolution exists and contains modules of the supertranslation algebra\footnote{This is true regardless of $\fn$ being free on $\overline{\fm}$.}.
Instead of performing a stepwise Tate resolution (by internal degree
or the stages of {\S}\ref{sec:tateres}), introducing new variables killing cohomology at each step, we want to deform the differential $s$ in a way that ``ties together'' the structure on the given vector space in order to eliminate cohomology, at the same time using the module structure inside $F[\overline{\fm}]$. Denote $\fa=\ft\inplus\fn=\ft\inplus F[\overline{\fm}]$.
Linear (infinitesimal) deformations will belong to
$H^\bullet(\fa,\fa)$, and in fact only to $H^\bullet(\fa,\fn)$.

\numpar[sec:doublecomplexdiff][The differential] The corresponding differential is written as
$s=s_0+s_1$, where $s_0$ encodes the supertranslation algebra and the supertranslation transformations of the modules in $F[\overline{\fm}]$, and $s_1$ encodes the brackets of $\fn=F[\overline{\fm}]$.
Then, $s$ is the coalgebra differential of the Lie superalgebra $\fa$.
Note that 
\begin{align}
s_0^2=0\;,\quad s_1^2=0\;,\quad s_0s_1+s_1s_0=0\label{eq:doublecomplex}
\end{align}
(the structure constants of $F[\overline{\fm}]$ are supertranslation invariant).

\numpar[sec:connectingH][The connecting cocycle.]
We note that stage 2 in the Tate resolution 
(see Table \ref{table:tatebigrading}) introduces the element in $\fm$ of lowest weight degree to kill the corresponding cohomology. 
Denote this supertranslation cohomology $\eta\in H^p(\ft)$, it
provides the ``seed'' of the cocycle. 
It generates stage 2 of the Tate resolution.
It can be lifted to
supertranslation cohomology $\omega_0\in H^p(\ft,\overline{\fm})$ taking values in the module $\overline{\fm}$, of which $\eta$ is the part of lowest weight degree.
In the bigrading of the Tate resolution depicted in Table \ref{table:tatebigrading}, $\eta$ appears somewhere on the stage 2 line and at cohomological degree $p$, consequently at weight degree
$p+2$. The $\ft$-module $m_\eta\subseteq M$, in which $\eta$ sits at the lowest weight degree, stretches to the right in the Table, to higher weight degrees. The cocycle $\omega_0$ will of course still contain a finite number of terms.

This introduces a certain $p$-bracket between supertranslation generators, $[D_{A_1},...,D_{A_p}]\in\overline{\fm}$, dually encoded in the differential as a term $\omega_0=c^{A_1}...c^{A_p}\frac{\*}{\*\psi^{A_1\ldots A_p}}$, $\psi\in\fm$, where $c^A$, $A=(\alpha,a)$, are supertranslation coalgebra $1$-forms, $c^\alpha=\lambda^\alpha$, $c^a=v^a$.
(We can formally think of the $\psi$'s thus obtained as dual to closed superspace $p$-forms $F_{A_1\ldots A_p}$ with $F_{\alpha_1\ldots\alpha_p}=0$; it is no coincidence that 
the ``connecting cocycle'' for models containing a $p$-form field strength is a $p$-cocycle, see \S\ref{sec:superspace}.)

\numpar[sec:doublecomplexomega][The deformation from the double complex.] 
The supersymmetric cocycle
$\omega_0$ does not yet provide an element of $H^p(\fa,\fa)$, since only 
$\{s_0,\omega_0\}=0$. We can now use that fact that  \eqref{eq:doublecomplex} provides a double complex. 
Since $s_0$ acts within the supermultiplets, and $s_1$ has the schematic form $s_1\sim ww\frac{\*}{\*w}$, where $w$
denotes $1$-forms in $\overline{\fn}[-1]$,
$s_0$ does not change the number of $w$'s, while $s_1$ increases it by $1$. We can therefore find the linear deformation 
$\omega=\omega_0+\omega_1+\ldots+\omega_q$ in
$H^p(A,\fn)$, for some $q\leq p$, by solving descent equations
\begin{align}
\{s_0,\omega_0\}&=0\;,\nn\\
\{s_0,\omega_1\}+\{s_1,\omega_0\}&=0\;,\nn\\
\ldots&\\
\{s_0,\omega_q\}+\{s_1,\omega_{q-1}\}&=0\;,\nn\\
\{s_1,\omega_q\}&=0\;.\nn
\end{align}
The part $\omega_i$ is dual to a $p$-bracket where $i$ entries are in $\fn$ and $p-i$ are supertranslation generators.

We used supertranslation cohomology as the starting point $\omega_0$ for the solution to the descent equations. The end point, $\omega_q$, represents cohomology of $s_1$. However, a freely generated algebra has no cohomology higher than $H^1$, in any 
module \cite{weibel_1994}. (It is also straightforward to see that a non-vanishing $H^1$ is supported in generation $1$.) We can easily observe that $q>0$, since $\{s_1,\omega_0\}\neq0$. Therefore, $q=1$. (Strictly speaking, there is a representative with $q=1$.)

Explicit forms for $\omega=\omega_0+\omega_1$ are of course highly model-dependent. 
This completes the linear deformation, yielding $p$-brackets
$[D_{A_1},\ldots,D_{A_p}]$ and $[D_{A_1},\ldots,D_{A_{p-1}},\phi]$.
It is interesting to note that it suffices to know the latter bracket for $\phi\in\overline{\fm}$, \ie, generation $1$. The cocycle condition states that $\omega_1$ encodes an outer derivation on $\fn=F[\overline{\fm}]$, so knowledge of its action on generation $1$ specifies the full action on $\fn$.

\numpar[sec:nonlinearomega][Non-linear deformation.] So far, we have considered linear deformations. The next question is to ask if $\omega^2=0$, or if the $(2p-1)$-identity requires the introduction of a $(2p-2)$-bracket dual to $\omega'$ with $\omega^2+\{s,\omega'\}=0$, and so on. 
From the examples we consider below, it seems that $\omega^2=0$ for non-gravitational models, although we do not have a proof. In $D=11$ supergravity, there is indeed a ``slight'' violation of $\omega^2=0$---the $5$-identity requires a $4$-bracket (which may not be too surprising, since the theory contains a $4$-form field strength). 

\subsection{Superspace formulation\label{sec:superspace}}
\hfill
\nopagebreak

\numpar[sec:Psigauge][Closed forms from pure spinor superfield cohomology.]
It is well established (see for example  \cite{Cederwall2013PureOverview}) that scalar pure spinor superfields 
$\Psi$
describe gauge theories (for connections of some form degree) on superspace, and reproduce closed superspace forms (field strengths) as gauge invariant cohomology. The component of $\Psi$ at appropriate cohomological degree (ghost number) is identified as the lowest-dimensional superspace connection component, \ie, $\Psi=\lambda^\alpha A_\alpha$ for a gauge connection, $\Psi=\frac12\lambda^\alpha\lambda^\beta A_{\alpha\beta}$ for a $2$-form connection, etc., while the leading $\lambda$-independent component represents the leading scalar ghost.

We have seen this picture emerging also in \S\ref{sec:freeconstruction}, where the connecting cocycle $\omega_0$ has
a dual interpretation as a closed superspace $p$-form,
\begin{align}
[D_{A_1},\ldots,D_{A_p}]=F_{A_1\ldots A_p}\;.
\label{eq:DDDF}
\end{align}
This is however not a standard way to introduce field strengths, 
except when $p=2$.
Note that the interpretation is quite formal; the right hand side of \eqref{eq:DDDF} is not the field strength, but a certain set of basis elements (generators) for the {\it dual} of the multiplet $\fm$.

\numpar[sec:tinftyinteractions][$\widetilde{\ft}$ and interactions.]
It is a remarkable observation 
that when $p=2$, \eg\ in $D=10$ super-Yang--Mills theory \cite{Movshev2004OnTheories,Movshev:2009ba},
this deformation leads to a deformation of the supertranslation algebra which is isomorphic to the superalgebra of gauge covariant derivatives $D^{(A)}=D+A$ together with the free algebra on the multiplet of the superspace field strengths. The free generation mimics the decoration of an abelian field strength with $\gl_\infty$ Chan--Paton factors, so any Cayley--Hamilton relations for specific gauge groups are absent; this is a ``master algebra''.
The same interpretation holds for any example with a BKM superalgebra which is freely generated from internal degree $3$, among which are our examples in \S\ref{sec:PhysicalBorcherds}.

The most remarkable property of this phenomenon is the fact that, although the positive BKM superalgebra $\widetilde{\ft}$ is a ``linear'' structure in the BV sense---it is the algebra which is responsible for the appearance of the linear supermultiplet---it encodes precisely the information that deforms the linear multiplet to an interacting one.
It is therefore an interesting question to ask what the corresponding non-linear structure on the supertranslations and the supermultiplet
means when $p>2$.

\numpar[sec:higherforms][Higher forms.]
Let us now consider arbitrary values of $p$.
Even if the geometric meaning of \eqref{eq:DDDF} remains somewhat unclear\footnote{We thank Christian S\"amann for communication on this issue.} for $p>2$, the deformation has one more part, namely $\omega_1$, encoding
the bracket $[D_{A_1},\ldots,D_{A_{p-1}},X]$, $X\in F[\overline{\fm}]$.
When $p>2$, there is no way to incorporate the connection in a gauge covariant derivative. It is possible, however, to do an analogous thing for a multi-derivative, and write
\begin{align}
[D_{A_1},\ldots,D_{A_{p-1}},X]=[A_{{A_1}\ldots A_{p-1}},X]
\label{eq:DDDX}
\end{align}
for some suitable interpretation of the bracket on the right hand side.
If this bracket obeys the Jacobi identity together with the bracket on
$\ft\inplus F[\overline{\fm}]$, the $(p+1)$-identity
\begin{align}
&[[D_{A_1},\ldots,D_{A_p}],X]+p[D_{[A_1},\ldots,D_{A_{p-1}},[D_{A_p]},X]]
\nn\\
&+p[D_{[A_1},[D_{A_2},\ldots,D_{A_p]},X]]
+\frac{p(p-1)}2[[D_{[A_1},D_{A_2}],D_{A_3},\ldots,D_{A_p]},X]=0
\end{align}
(sign factors omitted)
is automatically satisfied, thanks to 
$[D_A,D_B]=-T_{AB}{}^CD_C$
and
\begin{align}
F_{A_1\ldots A_p}=pD_{[A_1}A_{A_2\ldots A_p]}
+\frac{p(p-1)}2T_{[A_1A_2}{}^BA_{|B|A_3\ldots A_p]}\;.
\end{align}
This is not too surprising, since what we are effectively demanding is that the $p$-bracket with one field is an outer derivation of $F[\overline{\fm}]$, see \S\ref{sec:freeconstruction}.

The question is then how to define the $2$-bracket of \eqref{eq:DDDX}.
Concerning the action of supertranslations, it is obvious that the multiplet $\overline{\fm}$, the gauge invariant multiplet, can be continued to the left (lower internal degree) to include the components of the superspace connection. When $p=2$, $A_\alpha$ and $A_a$ occur in the same position as $D_\alpha$ and $D_a$, and the deformation is obtained by considering the sums, the gauge covariant derivatives. When $p\geq3$, the connections will appear at homological degree $p-2$, same as $\overline{\fm}$.
The algebra $\fn=F[\overline{\fm}]$ can be extended to lower internal 
degrees. $2$-brackets with elements in the new positions will indeed then be derivations of $F[\overline{\fm}]$, which are outer since the new elements are outside  $F[\overline{\fm}]$. The concrete realisation will be highly model-dependent, but exists by definition.
A schematic picture is given in Table \ref{table:ConnectionAlgebra} for a model with $p=3$, like the models in the examples of
\S\ref{sec:D6N20} and \S\ref{sec:D4N2}.

The algebra $\fn=F[\overline{\fm}]$ is now extended to $\underline{\fn}$, which as a vector space is $\fb\oplus\fn$, where $\fb$ is spanned by the potentials. Elements in $\fb$ act as outer derivations on $\fn$.
In the examples, in particular {\S}\ref{sec:D6N20} and {\S}\ref{sec:D11}, we will come back to the question whether further outer derivations at higher homological degree are obtained by repeatedly applying elements in $\fb$, and formulate conjectures essentially stating that this happens in supergravity but not otherwise. 

\begin{table}
    \centering
\begin{picture}(200,150)(150,-30)
\put(35,12){hom.}
\put(35,0){degree $=$}
\put(80,0){$0$}
\put(80,30){$1$}
\put(80,60){$\vdots$}
\put(50,-30){int. degree $=$}
\put(120,-30){$1$}
\put(160,-30){$2$}
\put(200,-30){$3$}
\put(240,-30){$4$}
\put(280,-30){$5$}
\put(320,-30){$6$}
\put(360,-30){$\cdots$}
\put(80,60){$\vdots$}
\put(100,-15){\line(0,1){90}}
\put(100,-15){\line(1,0){280}}
\multiput(200,0)(40,0){5}{$\cdot$}
\multiput(120,30)(40,0){1}{$\cdot$}
\multiput(360,30)(40,0){1}{$\cdots$}
\multiput(120,60)(40,0){7}{$\cdot$}
\put(115,0){$D_\alpha$}
\put(155,0){$D_a$}
\put(155,30){\textcolor{red!80!black}{$A_{\alpha\beta}$}}
\put(192,30){\textcolor{red!80!black}{$A_{a\beta}$}}
\put(222,30){\textcolor{red!80!black}{$A_{ab}$}, $F_{a\beta\gamma}$}
\put(275,30){$F_{ab\gamma}$}
\put(315,30){$F_{abc}$}
\put(220,15){\line(0,1){60}}
\put(220,15){\line(1,0){160}}
\end{picture}
    \caption{\it Generators in the double grading of $\widetilde{\ft}$, extended with connections, for a model with a superspace $2$-connection.  Only the generating set of $F[\overline{\fm}]$ is displayed. }
    \label{table:ConnectionAlgebra}
\end{table}
\numpar[sec:interpretation][Interpretation.]
An interpretation of the non-linear structure on the multiplet in terms of self-interaction seems difficult for higher gauge symmetries ($p\geq3$).
Indeed, to be physically sensible it would demand a model where the pure spinor superfield carries dimension $0$, \ie, where internal degree and physical dimension agrees, up to some conventional proportionality constant
(unless, of course, the interactions involve a dimensionful constant).
This happens when $p=2$, where the leading ghost is dimensionless, but not necessarily in models with higher gauge fields. In $D=11$ supergravity it does not hold.
In other words, self-interactions should, at least na\"\i vely, involve some modification of the
$2$-bracket $[D_\alpha,\cdot\,]$; this can then be understood as a deformation of the pure spinor superspace differential
${\mathcal D}=\lambda^\alpha D_\alpha$.
This does not exclude that relations like \eqref{eq:DDDF} can be given a deeper geometric meaning.
We will see that a certain such modification actually occurs in $D=11$ supergravity, and it can be seen as a sourcing of the physical multiplet by the Killing ghosts.

\vskip6pt
\section{\texorpdfstring{Properties of the ring $S$ and the multiplet $\fm$}{Properties of the ring S and the multiplet M}
\label{sec:Properties}}
\vskip8pt
There are a number of interesting properties that the pure spinor ring may or may not have in different cases, and similarly for the multiplet.
We will list some of them. Even if they ultimately all are properties of $S$, statements about $\fm$ distinguish themselves in being statements about $\fn=\widetilde{\ft}_{\geq3}$.

\underbar{Properties of $S$}

\begin{enumerate}[label=\protect\circled{\Alph*}]
    \item\label{property:minbasic} $Y$ is the closure of a minuscule minimal $G$-orbit (see {\S}\ref{sec:MinGorenstein});
    \item\label{property:min} $Y$ is the closure of a (set of) minimal G-orbit(s);
    \item\label{property:CI} $S$ is a complete intersection;
    \item\label{property:Gor} $S$ is Gorenstein;
    \item\label{property:C-M} $S$ is Cohen--Macaulay;
    \item\label{property:Kos} $S$ is Koszul.
 \end{enumerate}
For a discussion of the relevance of some of the properties of the ring $S$ to physical constructions involving the corresponding multiplet, we refer to~\cite{perspectivesonpurespin}. Some of the constructions there, in particular the construction of cotangent BV theories in examples where the canonical multiplet is off-shell, were anticipated in~\cite{MovshevBar}.
 
Implications between these properties are as follows:
 
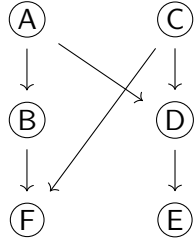
\begin{center}
    \begin{tikzcd}[row sep = 16 pt, column sep = 16 pt]
        \ref{property:minbasic}\ar[d]\ar[drr]
                &&\ref{property:CI}\ar[ddll]\ar[d]\\
        \ref{property:min}\ar[d]&&\ref{property:Gor}\ar[d]\\
        \ref{property:Kos}&&\ref{property:C-M}
    \end{tikzcd}
\end{center}

\underbar{Properties of the multiplet}

 \begin{enumerate}[resume,label=\protect\circled{\Alph*}]
     \item\label{property:coho1} $\fm$ has support only in (cohomological) degree $1$;
     \item\label{property:cohosingle} $\fm$ has support in a single degree;
     \item\label{property:infcohosingle} $\fm$ has $\infty$-dimensional support in at most one degree;
     \item\label{property:nfree} $\fn$ is freely generated;
     \item\label{property:freeinn} $\fn$ has a (non-abelian) freely generated subalgebra.
 \end{enumerate}
The implications between these statements are:
 
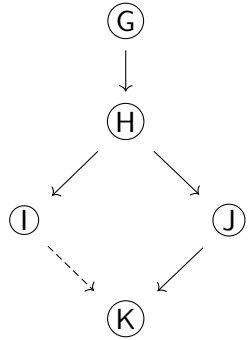
\begin{center}
    \begin{tikzcd}[row sep = 16 pt, column sep = 16 pt]
        &\ref{property:coho1}\ar[d]\\
        &\ref{property:cohosingle}\ar[dl]\ar[dr]\\
        \ref{property:infcohosingle}\ar[dr,dashed]
        &&\ref{property:nfree}\ar[dl]\\
        &\ref{property:freeinn}
    \end{tikzcd}
\end{center}
The dashed arrow is conjectural. All examples we have encountered obey it. The example with $G_2\times SL(2)$ in \S\ref{sec:G2A1} has a non-free $\fn$ with a free subalgebra when the conventions otherwise used in the paper are adopted, in particular with the definition $\fn=\widetilde{\ft}_{\geq3}$. In this case it turns out to be well motivated to instead choose $\fn=\widetilde{\ft}_{\geq4}$, the conjectural arrow would then be 
\ref{property:infcohosingle} $\dashrightarrow$ \ref{property:nfree}.

It is striking how few implications can be made from a property in one group to one in the other.
A trivial exception is the vanishing of the multiplet being equivalent to $S$ being a complete intersection.
In particular, the physically motivated ``good'' property of free generation is difficult to relate to standard ``good'' properties of rings. 
For example, Gorenstein or Koszul are about the ``best'' properties one can demand from a ring. They are typically properties of some well behaved physical examples, such as $D=10$ super-Yang--Mills, {\S}\ref{sec:SYM}. Indeed, Gorenstein is necessary for a BV action without the introduction of a separate dual multiplet. Supergravity models are never Koszul, but SYM models are.
On the other hand, rings with both these properties, as the 
$E_6$ model of {\S}\ref{sec:E6}, may display pathological behaviour from a physical point of view, such as infinite-dimensional cohomology in more than one ghost number.

Avramov~\cite{Avramov:1982} made a very general conjecture about free subalgebras in Noetherian local rings. In our context it implies arrows
\ref{property:C-M} $\dashrightarrow$ \ref{property:freeinn}
$\dashleftarrow$ \ref{property:Kos}.
To our knowledge, the conjecture has neither been proven or falsified.
All our examples 
satisfy condition \ref{property:freeinn} (except for
$D=4,\, \mathcal{N}=4$, where $\mathfrak{n}$ is trivial, see
\S\ref{sec:D4N4}).
We have among our examples, \S\ref{sec:E6}, a minimal orbit under $E_6$, thus defining a BKM superalgebra. Property \ref{property:minbasic} holds, implying 
properties \ref{property:C-M} and \ref{property:Kos}.
Condition \ref{property:freeinn} also holds,
since $\mathscr{B}(E_5,\Lambda_1)_{\geq3}$, which is freely generated, is contained in $\mathscr{B}(E_6,\Lambda_1)_{\geq3}$
\cite{Kleinschmidt:2013em}.
However, the whole of $\mathfrak{n}$ is not freely generated
in the $E_6$ case unlike the other examples
considered below.

The properties of the different examples of
\S\ref{sec:examples} are listed in Table \ref{tab:examplesproperties}, in order of appearance in \S\ref{sec:examples}.

\renewcommand{\arraystretch}{1.5} 
\begin{table}[h]
    \centering
    \begin{tabular}{|l V{4}c|c|c|c|c|c V{4} c|c|c|c|c|}
    \hline
&\ref{property:minbasic}
&\ref{property:min}
&\ref{property:CI}
&\ref{property:Gor}
&\ref{property:C-M}
&\ref{property:Kos}
&\ref{property:coho1}
&\ref{property:cohosingle}
&\ref{property:infcohosingle}
&\ref{property:nfree}
&\ref{property:freeinn}\\
\hlineB{4}
$D=4$, $\N=4$ &No&No&Yes&Yes&Yes&Yes&---&---&---&---&--- \\
\hline
$D=10$, $\N=1$ &Yes&Yes&No&Yes&Yes&Yes&Yes&Yes&Yes&Yes&Yes \\
\hline
$SL(5)$&Yes&Yes&No&Yes&Yes&Yes&Yes&Yes&Yes&Yes&Yes \\
\hline
$D=6$, $\N=1$&No&Yes&No&No&Yes&Yes&Yes&Yes&Yes&Yes&Yes \\
\hline
$D=4$, $\N=1$&No&Yes&No&No&No&Yes&Yes&Yes&Yes&Yes&Yes \\
\hline
$D=6$, $\N=(2,0)$&No&No&No&No&Yes&No&No&Yes&Yes&Yes&Yes \\
\hline
$D=4$, $\N=2$&No&No&No&No&Yes&No&Yes&Yes&Yes&Yes&Yes \\
\hline
$D=11$&No&No&No&Yes&Yes&No&No&No&Yes&Yes&Yes \\
\hline
$E_6$&Yes&Yes&No&Yes&Yes&Yes&No&No&No&No&Yes \\
\hline
$G_2\times SL(2)$&No&No&No&No&No&No&No&No&Yes&No&Yes \\ 
\hline
    \end{tabular}
    \caption{\it Properties of $S$ and $\fm$ in the examples. Note that many ``Yes'' in group \ref{property:minbasic} - \ref{property:Kos} does not guarantee the same in group \ref{property:coho1} - \ref{property:freeinn}, and vice versa.}
    \label{tab:examplesproperties}
\end{table}

\vskip6pt 
\section{Examples\label{sec:examples}}
\vskip8pt

In all examples with physical multiplets below,
\S\ref{sec:PhysicalBorcherds}-\S\ref{sec:G2A1}, we use the notation and method for deriving the bracket structure of $\widetilde{\ft}$ introduced in \S\ref{sec:freeconstruction}.
We also use a convention where $1$-forms dual to bosonic and fermionic generators commute, a slight violation of our general conventions which can be changed to anticommutator by trivial sign changes.

\subsection{Trivial examples}
\hfill
\nopagebreak

\numpar[sec:trivex][Generic classes of trivial examples]
\nopagebreak
\begin{enumerate}
  \item Let $\ft_1$ and~$\ft_2$ be any vector spaces, and take $\ft$ to be abelian. Then $S = R$ is just a polynomial ring, and 
  $S^\perp=U(\ft_1)$ is the exterior algebra on $\ft_1$ (with shifted, bosonic, parity). 
  \item \label{item:complint} Consider any complete intersection, defined by a regular sequence of quadratic elements $f_i \in R$. This determines an algebra of the form $\ft$ by taking the  Chevalley--Eilenberg differential to be $d(a_i) = f_i$. The partition function is $Z_S(t)=(1-t)^{\overline{ R_1}}(1-t^2)^{\overline{ R_2}}$, and $\widetilde{\ft}=\ft$. The multiplet is empty; pure spinor superfield cohomology is de Rham cohomology. 
  \item \label{item:free} Let $\ft_2=\wedge^2\ft_1$. These are the first two degrees of $F[\ft_1]$. Then $S$ is empty from degree $2$. The partition function is $Z_S(t)=1+S_1t=1-R_1t$, and $\widetilde{\ft}=F[\ft_1]$.
\item Let $\ft_1 =-k$, $ \ft_2  = k$, with the bracket given by tensor product. This is a special case of both (\ref{item:complint}) and (\ref{item:free}).
The partition function is
$Z_S(t)=1+t=(Z_{F[-k]})^{-1}$.
\end{enumerate}

\numpar[sec:D4N4][$D=4$, $\N=4$ super-Yang--Mills.]
This example is trivial in the sense that it falls under item (\ref{item:complint}) above, $S$ is a complete intersection, and $\widetilde{\ft}=\ft$. $D=4$ is the highest dimension where this happens for maximally supersymmetric Yang--Mills theory.
The zero-mode cohomology then trivially becomes the de Rham complex in $D=4$.
\begin{table}[H]
\[
    \begin{tikzcd}[row sep = 8 pt, column sep = 8 pt]
        \Omega^0 \ar[dr,shorten=-1mm]  \\
        \cdot & \Omega^1 \ar[dr,shorten=-1mm] \\
        \cdot&\cdot&\Omega^2\ar[dr,shorten=-1mm] \\
        \cdot & \cdot & \cdot &\Omega^3 \ar[dr,shorten=-1mm]\\
        \cdot & \cdot & \cdot &\cdot&\Omega^4 
\end{tikzcd}
\]
\caption{\it Zero-mode cohomology for $D=4$, $\N=4$ SYM.
\label{table:D4N4}}
\end{table}
In Table \ref{table:D4N4} and all subsequent tables of zero-mode cohomology, the $\lambda$ expansion is vertical and the $\theta$ expansion horizontal.
In the present case the multiplet is empty; there is no local cohomology. Nevertheless, this pure spinor superfield can be used ``minimally coupled'' with a field describing the gauge-covariant multiplet based on the scalars, so that the equations of motion identify the field strength of the gauge connection in Table~\ref{table:D4N4} with the one in the scalar multiplet
\cite{Cederwall:2017cez}.
There is a natural (contragredient) extension of $\widetilde{\ft}$ to non-positive levels, which is the superconformal algebra
$\su(4|4)$.

\subsection{Physical examples with BKM superalgebras\label{sec:PhysicalBorcherds}}
\hfill
\nopagebreak

The dual BKM superalgebras appearing in this subsection are all freely generated from internal degree $3$. They all correspond to multiplets concentrated in a single ghost number. In all cases, the Koszul duality is the standard quadratic one. The algebras $\widetilde{\ft}$ are isomorphic to algebras of gauge covariant derivatives and field strengths. As we will see,
the examples in {\S}\ref{sec:SYM}, {\S}\ref{sec:D6N1} and {\S}\ref{sec:D4N1} are related to the exceptional Lie algebras $E_8,E_7$ and $E_6$, respectively.

\numpar[sec:SYM][Ten-dimensional minimal supersymmetry.]
This is the standard example for pure spinor superfield theory, possessing all the good properties that allow for a supersymmetric BV theory of Chern--Simons type \cite{Cederwallspinorial,Berkovits2000Super-PoincareSuperstring,Berkovits2001CovariantSpinors,Movshev2004OnTheories,Movshev:2009ba,GorbounovSchechtman}. 
$S$ is a Gorenstein ring of odd codimension ($5$).
Of particular interest for us is the construction of the Koszul dual algebra $\widetilde{\ft}$, which we know is the BKM superalgebra
${\mathscr B}(D_5,\Lambda_5)$ (Figure \ref{fig:be5}) \cite{Cederwall2015SuperalgebrasFunctions}, and which Movshev and Schwarz \cite{Movshev2004OnTheories,Movshev:2009ba} showed is isomorphic to the superalgebra of superspace gauge covariant derivatives and field strengths of (interacting) super-Yang--Mills theory. We will rephrase that construction in the framework of the present paper.
The algebra is free from internal degree $3$.
\begin{figure}[H]
    \centering
\begin{picture}(170,85)(45,-15)
\multiput(50,10)(40,0){5}{\circle{10}}
\put(130,50){\circle{10}}
\multiput(55,10)(40,0){4}{\line(1,0){30}}
\put(208,-10){$1$}
\put(168,-10){$2$}
\put(128,-10){$3$}
\put(88,-10){$5$}
\put(130,15){\line(0,1){30}}
\put(145,45){$4$}
\put(50,10){\line(1,1){3.5}}
\put(50,10){\line(1,-1){3.5}}
\put(50,10){\line(-1,1){3.5}}
\put(50,10){\line(-1,-1){3.5}}
\end{picture}
    \caption{\it Dynkin diagram for $\mathscr{B}(D_5,\Lambda_5)$}
    \label{fig:be5}
\end{figure}
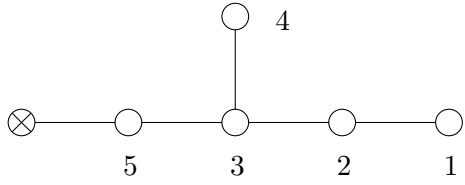
The partition function is
\begin{align}
Z_S(t)&=\bigoplus_{n=0}^\infty(0000n)t^n\nn\\
    &=(1-t)^{-(00001)}\otimes\Bigl(\Bigr.  \label{D5partfn}
        (00000)\ominus(10000)t^2\oplus(00010)t^3\\
    &\hspace{32mm}    \ominus(00001)t^5
        \oplus(10000)t^6\ominus(00000)t^8
    \Bigl.\Bigr)\nn\\
    &=(1-t)^{-(00001)}\otimes(1-t^2)^{(10000)}\!\otimes\!
    \Bigl[
        (00000)\ominus\bigoplus_{i=0}^\infty
        \left(
        \ominus(i0010)t^{3+2i}\oplus(i1000)t^{4+2i}
        \right)
    \Bigr]\;,\nn
\end{align}
where the Dynkin labels are given in the order indicated by the numbering of the nodes in the Dynkin diagram in Figure~\ref{fig:be5}.

The first form of the partition function in (\ref{D5partfn}) expresses that the pure spinor variety is a minimal orbit, the second gives the zero mode cohomology and the third the multiplet, the on-shell $D=10$ super-Yang--Mills fields in a derivative expansion. From the exponents of the first two factors in the
the third form we can read of the $D_5$ representations $R_1=(00010)$ and $R_2=(10000)$, corresponding to $\ft_1$ and $\ft_2$.
Below, these are denoted by their dimensions in bold, $R_1={\bf 16}$ and $R_2={\bf 10}$. Continuing the factorisation, we would get $R_3=\overline{\bf 16}$ and $R_4=\overline{\bf 45}$ (the adjoint). The fact
that $\overline{ R_1}=R_3$ and $\overline{ R_2}=R_2$ agrees
with the general symmetry described in {\S}\ref{sec:bkmsym}, considering ${\mathscr B}(D_5,\Lambda_5)$ as ${\mathscr B}(E_8{},\theta)^{(-3)}$.

\begin{table}[H]
\[
    \begin{tikzcd}[row sep = 8 pt, column sep = 8 pt]
        {\bf1} \ar[dr,shorten=-1mm]  \\
        \cdot & {\bf10} \ar[drrr,controls={+(1,-1) and +(-1,1)}] 
        & {\bf16}\ar[dr,shorten=-1mm] \\
        \cdot & \cdot & \cdot &\overline{\bf16} & {\bf10}\ar[dr,shorten=-1mm]\\
        \cdot & \cdot & \cdot &\cdot&\cdot &{\bf1}
\end{tikzcd}
\]
\caption{\it Zero-mode cohomology for $D=10$ SYM.}
\end{table}

The supertranslation algebra is encoded in the differential $s_0$,
\begin{align}
    s_0v^a= (\lambda\gamma^a\lambda)\;.
\end{align}
The action on the on-shell multiplet is encoded in
\begin{align}
    s_0\psi_\alpha&=0\;,\nn\\
    s_0 f_{ab}&=-(\lambda\gamma_{ab}\psi)\;,\label{eq:SYM10transf}\\
    s_0\psi_{a\alpha}&=v_a\psi_\alpha
        +(\gamma^b\lambda)_\alpha f_{ab}-(\gamma\hbox{-trace})\;,\nn\\
    &\cdots\nn
\end{align}
Note that here, as in all subsequent examples, we are representing the multiplet by basis elements in the coalgebra. These basis elements are in the modules conjugate to the ones of the component fields. Also, there is no $x$-dependence; derivatives of the fields in the multiplet are represented as separate elements. The coalgebra element corresponding to a derivative of the spinor $\*_a\Psi^\alpha$ is $\psi_{a\alpha}$. The third relation in  \eqref{eq:SYM10transf} states among other things that
on the algebra side
$[D_a,\Psi^\alpha]=\Psi_a{}^\alpha$, where $\gamma^a_{\alpha\beta}\Psi_a{} ^\beta=0$.

As detailed in \S\ref{sec:freeconstruction}, the action of $\ft$ on $\overline{\fm}$, specified by $s_0$ above, is naturally extended to 
$F[\overline{\fm}]$. Together with $s_1$, the differential on $C^\bullet(F[\overline{\fm}])$, $s_0+s_1$ is the differential on
$C^\bullet(\ft\inplus F[\overline{\fm}])$.

The connecting cocycle is $\eta_\alpha=v^a(\gamma_a\lambda)_\alpha$. It is closed thanks to the identity 
$(\lambda\gamma_a\lambda)(\gamma^a\lambda)_\alpha=0$, which is trivially true thanks to $\vee^3{\bf16}\not\supset\overline{\bf16}$. It is lifted to a supersymmetric cocycle $\omega_0\in H^2(\ft,\overline{\fm})$ as
\begin{align}
    \omega_0\psi_\alpha&=v^a(\gamma_a\lambda)_\alpha\;,\nn\\
    \omega_0 f_{ab}&=v_av_b\;.\label{eq:omega0-10}
\end{align}
The construction precisely mimics the deformation of $D_\alpha\in\ft_1$, $D_a\in\ft_2$ to gauge covariant derivatives by including connections $A_A$, $A=(\alpha,a)$ and demanding $F_{\alpha\beta}=0$. Then, \eqref{eq:omega0-10} is dual to
$[D_A,D_B]=F_{AB}$. This is the standard superspace procedure for introducing interaction in super-Yang--Mills theory. Here, the same algebra is responsible for defining the states in the linear theory.
There is of course also an $\omega_1$, defined as in \S\ref{sec:freeconstruction}, which makes $\omega=\omega_0+\omega_1$ a cocycle in $H^2(\ft\inplus F[\overline{\fm}],F[\overline{\fm}])$.
The first example of a term in $\omega_1$ acts on $\psi_{\alpha,\beta}$, which is the generation $2$ element at weight degree $6$ with 
$s_1\psi_{\alpha,\beta}=\psi_\alpha\psi_\beta$. 
One then gets\footnote{The seemingly strange number $\frac{16}{9}$ is a consequence of 1) decomposing the derivative of $F$ as
$[D_a,F_{bc}]=(DF)_{(11000)}-\frac29\eta_{a[b}[D^d,F_{c]d}]$, and
2) the Fierz decomposition 
$[\Psi^\alpha,\Psi^\beta]
=\frac{1}{16}\gamma_a^{\alpha\beta}\gamma^a_{\gamma\delta}[\Psi^\gamma,\Psi^\delta]+\ldots$}
\begin{align}
\omega_1\gamma_a^{\alpha\beta}\psi_{\alpha,\beta}
=\frac{16}9v^bf_{ab}-\frac{20}9\lambda^\alpha\psi_{a\alpha}\;.
\end{align}
This encodes the sourcing of the gauge field by the fermions, the dual form of the first term being 
\begin{align}
    [D^b,F_{ab}]=\frac12\gamma_{a\alpha\beta}[\Psi^\alpha,\Psi^\beta]
\end{align}
(with some normalisation constant).
Seen from the perspective of gauge covariant derivatives 
the occurrence of such terms is obvious; repeated action of covariant derivatives gives curvatures, this one arises from
the Jacobi identity of the formally undeformed 
$\gamma^a_{\alpha\beta}[D_a,\Psi^\beta]=0$ with a $D_\alpha$.
We can think of the algebra as a master algebra for super-Yang--Mills theory with any gauge group, so that all fields are decorated with (invisible) $\gl_\infty$ Chan--Paton factors. The absence of Cayley--Hamilton relations for finite-dimensional gauge groups corresponds to the free generation.

\numpar[sec:SL5][\texorpdfstring{$SL(5)$, or twisted eleven-dimensional, supersymmetry.}{SL(5), or twisted eleven-dimensional, supersymmetry.}]

The formulation of this multiplet was given in 
\cite{Galvez2016,Cederwall:2021ejp}. It appears as a twisting of $D=11$ supergravity \cite{Raghavendran:2021qbh}.

$S$ is a Gorenstein ring of odd codimension ($3$). It is the algebra of functions on a c\^one over the Grassmannian $Gr(2,5)$ of $2$-planes in $5$ dimensions, a minimal $SL(5)$ orbit. Therefore, the Koszul dual Lie superalgebra is ${\mathscr B}(A_4,\Lambda_2)$, see Figure \ref{fig:be4}.
\begin{figure}[h]
    \centering
\begin{picture}(130,70)(45,0)
\multiput(50,10)(40,0){4}{\circle{10}}
\put(90,50){\circle{10}}
\multiput(55,10)(40,0){3}{\line(1,0){30}}
\put(90,15){\line(0,1){30}}
\put(50,10){\line(1,1){3.5}}
\put(50,10){\line(1,-1){3.5}}
\put(50,10){\line(-1,1){3.5}}
\put(50,10){\line(-1,-1){3.5}}
\end{picture}
   \caption{\it Dynkin diagram for $\mathscr{B}(A_4,\Lambda_2)$.}
    \label{fig:be4}
\end{figure}
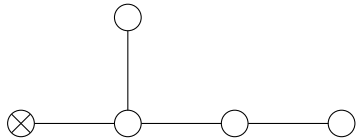
The partition function is
\begin{align}
    Z_S(t)&=\bigoplus_{n=0}^\infty(00n0)t^n\nn\\
    &=(1-t)^{-(0010)}\Bigl( 
        (00000)\ominus(1000)t^2\oplus(0001)t^3\ominus(00000)t^5\Bigr)\\
    &=(1-t)^{-(0010)}\otimes(1-t^2)^{(1000)}\otimes
    \Bigl[
        (0000)\ominus\bigoplus_{i=0}^\infty
        \left(
        \ominus(i001)t^{3+2i}\oplus(i100)t^{4+2i}
        \right)
    \Bigr]\;.\nn
\end{align}
where the Dynkin labels are given in the same order as
in the preceding example.

The first form expresses the pure spinor variety as a minimal orbit, the second gives the zero mode cohomology and the third the multiplet. It contains a fermionic divergence-free vector and a closed $2$-form.
This is the (parity-reversed) content of the superalgebra $E(5,10)$. In fact, ${\mathscr B}(A_4,\Lambda_2)$ is freely generated from internal degree $3$ by the coadjoint module of $E(5,10)$
\cite{Cederwall:2021ejp}.

\begin{table}[H]
\[
    \begin{tikzcd}[row sep = 8 pt, column sep = 8 pt]
        \bf1 \ar[dr,shorten=0mm]  \\
        \cdot & \bf5 & \overline{\bf5}\ar[dr,shorten=0mm] \\
        \cdot & \cdot & \cdot & \bf1 
\end{tikzcd}
\]
\caption{\it Zero-mode cohomology for the $SL(5)$ model.}
\end{table}

Supertranslations are expressed as
\begin{align}
    s_0v^m&=\frac18\epsilon^{mnpqr}\lambda_{np}\lambda_{qr}\;,\nn\\
    s_0\psi_m&=0\;,\nn\\
    s_0f^{mn}&=\frac12\epsilon^{mnpqr}\lambda_{pq}\psi_r\;,\\
    s_0\psi^m{}_n&=v^m\psi_n+\lambda_{np}f^{mp}-\hbox{(trace)}\;,\nn\\
    &\cdots\nn
\end{align}

The connecting cohomology from the supertranslation algebra to the
module is $\eta_m=\lambda_{mn}v^n$. It is extended to the supersymmetric cocycle $\omega_0$ as
\begin{align}
    \omega_0\psi_m&=\lambda_{mn}v^n\;,\nn\\
    \omega_0f^{mn}&=v^mv^n\;.
\end{align}
The lowest term in $\omega_1$ is
\begin{align}
    \omega_1\psi_{m,n}=2\lambda_{p(m}\psi^p{}_{n)}\;,
\end{align}
where $\psi_{m,n}$ is the generation $2$ coalgebra element with
$s_1\psi_{m,n}=\psi_m\psi_n$. Unlike the $D=10$ SYM example, {\S}\ref{sec:SYM}, there is no room for a $vf$ 
term; the field strength $F_{mn}$ only satisfies a Bianchi identity which remains homogeneous.

Similar constructions are certainly possible also for BKM superextensions of ``$E_n$'', $n<4$, and may be related to other exceptional Lie superalgebras, \eg\ $E(3,6)$.

\numpar[sec:D6N1][Six-dimensional minimal supersymmetry.]

The degree $0$ automorphism group is $SL(4)\times SL(2)$, and we use standard Dynkin labels. 
Let $\lambda\in(001)(1)=({\bf4},{\bf2})$. Since
\begin{align}
\vee^2(001)(1)=(002)(2)\oplus(010)(0)=({\bf10},{\bf3})\oplus({\bf6},{\bf1})\,,
\end{align}
the pure spinor constraint defines a minimal orbit, and the Koszul dual algebra $\widetilde{\ft}$ is the positive-degree part of a BKM superalgebra, with the Dynkin diagram of Figure \ref{fig:ba3a1}.

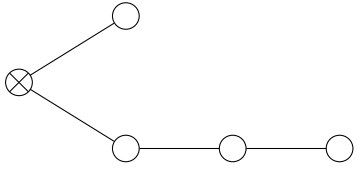
\begin{figure}[h]
    \centering
\begin{picture}(130,80)(45,0)
\multiput(90,10)(40,0){3}{\circle{10}}
\put(90,60){\circle{10}}
\put(50,35){\circle{10}}
\multiput(95,10)(40,0){2}{\line(1,0){30}}
\put(54.1,37.6){\line(40,25){31.5}}
\put(54.1,32.4){\line(40,-25){31.5}}
\put(50,35){\line(1,1){3.5}}
\put(50,35){\line(1,-1){3.5}}
\put(50,35){\line(-1,1){3.5}}
\put(50,35){\line(-1,-1){3.5}}
\end{picture}
    \caption{\it Dynkin diagram for $\mathscr{B}(A_3\oplus A_1,\Lambda_3+\Lambda_4)$}
    \label{fig:ba3a1}
\end{figure}

The diagram is drawn this way in order to distinguish the BKM superalgebra from the finite-dimensional superalgebra $\mathfrak{sl}(2|4)$ (the latter has a Cartan matrix with sign changes across the grey node).
In the same way as the BKM superalgebra in {\S}\ref{sec:SYM}
was described as $\mathscr B(E_8,\theta)^{(-3)}$
this BKM superalgebra can also be described as
$\mathscr B(E_7,\theta)^{(-3)}$.
$S$ is not Gorenstein, but Cohen--Macaulay.

It is well known that the superspace gauge theory with this supersymmetry yields the off-shell multiplet, with a triplet of auxiliary fields. This is of course reproduced by the pure spinor superfield formalism
\cite{Cederwall:2008zv,Cederwall:2017ruu}.
The partition function is
\begin{align}
Z_S(t)&=\bigoplus_{n=0}^\infty(00n)(n)t^n\\
&=(1-t)^{-(001)(1)}\otimes\left(
(000)(0)\ominus(010)(0)t^2\oplus(100)(1)t^3\ominus(000)(2)t^4
\right)\;,\nn
\end{align}
from which the off-shell multiplet of Table \ref{D6SYMzero} is read.
Obviously, extracting also degree $2$ from the partition function gives the multipet in the same way as in {\S}\ref{sec:SYM}, but without equations of motion in the absence of antifields. (The antifields sit in a separate pure spinor superfield in the module $(000)(2)$ \cite{Cederwall:2017ruu}.)

\begin{table}[H]
\[
    \begin{tikzcd}[row sep = 8 pt, column sep = 8 pt]
        \bf1 \ar[dr,shorten=-.5mm]  \\
        \cdot & \bf(6,1) & \bf(4,2) & \bf(1,3)\\
\end{tikzcd}
\]
\caption{\it Zero-mode cohomology for $D=6$, $N=1$ SYM.\label{D6SYMzero}}
\end{table}

The procedure for going from $\ft\inplus F[\overline{\fm}]$ by introduction of the connecting cocycle parallels the one in {\S}\ref{sec:SYM}, so we will not display it in detail. Again, the BKM superalgebra, which is freely generated from internal degree $3$, is isomorphic to the algebra of gauge covariant superspace derivatives and field strengths.

\numpar[sec:D4N1][Four-dimensional minimal supersymmetry.]
The R symmetry is $U(1)$. We denote modules by $\su(2)\oplus\su(2)$ Dynkin labels (or dimensions in boldface) with $\fu(1)$ charge as subscript.
Let $\lambda\in(1)(0)_{-1}\oplus(0)(1)_1=({\bf2},{\bf1})_{-1}\oplus({\bf1},{\bf2})_1$.
The partition function yields the zero-mode cohomology of Table \ref{D4N1zero}. The multiplet is off-shell, including the auxiliary scalar field.

The ring $S$ is not Cohen--Macaulay, due to the pure spinor variety consisting of two components. This makes the construction of the anti-field multiplet more subtle \cite{perspectivesonpurespin}.

\begin{figure}[h]
    \centering

\begin{picture}(50,70)(45,0)
\multiput(50,10)(40,0){2}{\circle{10}}
\multiput(50,50)(40,0){2}{\circle{10}}
\multiput(55,10)(0,40){2}{\line(1,0){30}}
\multiput(48.5,15)(3,0){2}{\line(0,1){30}}
\multiput(50,10)(0,40){2}{\line(1,1){3.5}}
\multiput(50,10)(0,40){2}{\line(1,-1){3.5}}
\multiput(50,10)(0,40){2}{\line(-1,1){3.5}}
\multiput(50,10)(0,40){2}{\line(-1,-1){3.5}}
\end{picture}
    \caption{\it Dynkin diagram for the BKM superalgebra associated to $D=4$, $N=1$ SYM.}
    \label{fig:bsomething}
\end{figure}
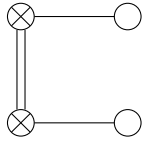

The procedure for constructing the BKM superalgebra as an extension of the supertranslation algebra follows the same lines as in {\S}\ref{sec:SYM} and {\S}\ref{sec:D6N1}, and has the same superspace interpretation. Accordingly,
the BKM superalgebra with Dynkin diagram in Figure
\ref{fig:bsomething}
could be denoted
$\mathscr B(E_6,\theta)^{(-3)}$, indicating that
it is ``three steps away'' from 
$\mathscr B(E_6,\theta)$, where in turn changing the grey node to a white one would give the affine extension of $E_6$. However, this case is different compared
to the previous ones in that we get {\it two} grey nodes in the last step of the ``oxidisation'' procedure
going back from $\mathscr B(E_6,\theta)$ to $\mathscr B(E_6,\theta)^{(-3)}$ \cite{Kleinschmidt:2013em}.
Also, from the point of view of this procedure, it is natural to connect the grey nodes by a double line without arrows,
meaning that both the corresponding off-diagonal entries in the Cartan matrix are equal to $-2$. However, any negative integer would give the same
algebra at positive degrees, as discussed in
\cite{Gomis:2018xmo}. Only when both positive and negative degrees are considered, different choices for the Cartan matrix give different algebras.
In \cite{Gomis:2018xmo} this algebra was studied in further detail.

\begin{table}[H]
\[
    \begin{tikzcd}[row sep = 8 pt, column sep = 8 pt]
        {\bf(1,1)}_0 \ar[dr,shorten=-.5mm]  \\
        \cdot & {\bf(2,2)}_0 & {\bf(2,1)}_1\oplus{\bf(1,2)}_{-1} & {\bf(1,1)}_0\\
\end{tikzcd}
\]
\caption{\it Zero-mode cohomology for $D=4$, $N=1$ SYM.
\label{D4N1zero}}
\end{table}

\subsection{Examples involving higher brackets\label{sec:HigherBrackets}}
\hfill
\nopagebreak

\numpar[sec:D6N20][\texorpdfstring{Six-dimensional $\N=(2,0)$ supersymmetry.}{Six-dimensional N=(2,0) supersymmetry.}]

The R-symmetry is $\mathfrak{so}(5)$. We use a notation with $D=6$ Pauli matrices $\gamma^a$ and $\mathfrak{so}(5)$ gamma matrices $\sigma^i$ acting on $16$-component spinors of $\mathfrak{so}(6)\oplus\mathfrak{so}(5)$. $\gamma^a$ and $\sigma^i$ commute.
The symmetric products of a spinor are
$(\lambda\gamma^a\lambda)$, $(\lambda\gamma^a\sigma^i\lambda)$ and 
$(\lambda\gamma^{abc}\sigma^{ij}\lambda)$.
Therefore, the pure spinor variety is not a minimal orbit.

The supertranslation algebra is dual to
\begin{align}
    s_0v^a&=(\lambda\gamma^a\lambda)\;.
\end{align}

\begin{table}[H]
\[
    \begin{tikzcd}[row sep = 4 pt, column sep = 4 pt]
        (\bf1,1) \ar[dr] &  \\
        \cdot & (\bf6,1) \ar[dr,end anchor={[xshift=-2ex]}]  \\
        \cdot & \cdot & (\bf15,1)\ar[dr,start anchor={[xshift=-6ex,yshift=.5ex]}, bend right=10]\oplus(1,5)\ar[drrr, bend left = 40,start anchor={[xshift=1ex]}] &  (\bf4,4) \ar[dr]  \\
        \cdot & \cdot & \cdot & (\overline{\bf10},\bf1) & (\overline{\bf4},\bf4) & (\bf1,5)
\end{tikzcd}
\]
\caption{\it Zero-mode cohomology for the $D=6$, $\N=(2,0)$ model.}
\end{table}

Supertranslations on the on-shell supermultiplet $\overline{\fm}$ of linear observables are encoded in
\begin{align}
    \sm\phi^i&=0\;,\nonumber\\
    \sm\psi^\alpha&=(\sigma_i\lambda)^\alpha\phi^i\;,\nonumber\\
    \sm g_{abc}&=(\lambda\gamma_{abc}\psi)\;,\\
    \sm\phi_a{}^i&=v_a\phi^i-(\lambda\gamma_a\sigma^i\psi)\;,\nonumber\\
    \sm\psi_a{}^\alpha&=v_a\psi^\alpha+(\sigma_i\lambda)^\alpha\phi_a{}^i
    +\frac14(\gamma^{bc}\lambda)^\alpha g_{abc}-(\gamma\hbox{-trace})\;,
    \nonumber\\
    \ldots\nonumber
\end{align}
Higher components represent $x$-derivatives of on-shell component fields, from which their transformations follow.

The supersymmetric cocycle $\omega_0$ in $H^3(\ft,\overline{\fm})$ is
\begin{align}
    \omega_0\phi^i&=-\frac12(\lambda\gamma^a\sigma^i\lambda)v_a\;,\nn\\
    \omega_0\psi^\alpha&=\frac14(\gamma^{ab}\lambda)^\alpha v_av_b\;,\\
    \omega_0g_{abc}&=v_av_bv_c+\hbox{dual}\;.\nn
\end{align}
$\omega_0$ encodes the $3$-bracket between three supertranslation generators, $[D_A,D_B,D_C]$. It is not a cocycle in $H^3(\mathfrak{t}\inplus F[\overline{ \fm}],F[\overline{\fm}])$, but the starting point of one. We know that the full $3$-cocycle is $\omega=\omega_0+\omega_1$.
We have not solved for the complete explicit form of $\omega_1$. Its first term is
\begin{align}
    \omega_1\phi^{i,j}=(\lambda\gamma^a\sigma^{(i}\lambda)\phi_a{}^{j)}
     -\frac16\delta^{ij}(\lambda\gamma^a\sigma_k\lambda)\phi_a{}^k
     -\frac16\delta^{ij}c^a(\lambda\gamma_a\psi)\;,
     \label{eq:outderphi}
\end{align}
where $\phi^{i,j}$ is the coalgebra element dual to $[\Phi_i,\Phi_j]$,
$\sF\phi^{i,j}=\phi^i\phi^j$.
Note that knowledge of $[D_A,D_B,\Phi]$ for $\Phi$ in $\overline{\fm}$ gives information of the $3$-bracket with one entry in the full $F[\overline{\fm}]$. They form a set of outer derivations of $F[\overline{\fm}]$.

So far, we have constructed a linear deformation of $s$ as a $3$-bracket differential $\omega$. Will $\omega$ represent a ``total deformation'', \ie, is also $\omega^2=0$, or will a $4$-bracket (and possibly higher) be generated?
Obviously, $\omega_0^2=0$. An identity $\{\omega_0,\omega_1\}=0$ would correspond to a $5$-identity $[D_{[A},D_B,[D_C,D_D,D_{E]}]]=0$. A quick inspection of the elements in the third generation of $F[\overline{\fm}]$ gives at hand that this identity holds---the content of $\mathfrak{so}(6)\oplus\mathfrak{so}(5)$ modules at the respective internal degrees does not allow any non-zero result of the $5$-identity. Thus, there is no $4$-bracket $[D_A,D_B,D_C,D_D]$. The only possibility would be a $4$-bracket $[D_A,D_B,D_C,\Phi]$ compensating for a possible
non-vanishing of $\omega_1^2$, and $4$-brackets with more than one entry in $F[\overline{\fm}]$ vanishing by the same argument as for the $3$-bracket.

Let us rephrase the above in the superspace formulation of \S\ref{sec:superspace}. 
We would then let $[D_A,D_B,D_C]=G_{ABC}$ and 
$[D_A,D_B,\Phi]=[B_{AB},\Phi]$, $\Phi\in F[\overline{\fm}]$, for superspace forms $G=dB$.
A 4-bracket would correspond to 
$[D_A,D_B,D_C,D_D]=H_{ABCD}$, $[D_A,D_B,D_C,\Phi]=[C_{ABC},\Phi]$.
The $5$-identity with all $D$'s now reads (modulo unimportant numerical coefficients)
$dH+[B,G]=0$, where antisymmetrisation (wedge product) is understood in the second term.
It is consistent with the $5$-identity with one element in $F[\overline{\fm}]$, which translates to
$[H+dC+[B,B],\Phi]=0$.
Note that the ``fields'' in the multiplet carry the opposite statistics compared to a physical multiplet---for example, $G_{abc}=g_{abc}$ is fermionic in the coalgebra---so the bracket in the last formula is non-trivial.
We have not checked that it actually vanishes. This requires at least a few more levels of the action of the outer derivations than given in
\eqref{eq:outderphi}. The supposed absence of any meaningful $4$-form observables in the physical multiplet leads us to believe that
$[B_{[AB},B_{CD]}]=0$ (as an outer derivation),
a statement that should be possible to prove with some more work.
In {\S}\ref{sec:D11}, we will see an example where both $3$- and $4$-forms indeed are present, and where the $4$-bracket arises precisely by the mechanism sketched here.  

We conjecture that the $3$-bracket $\omega$ is a total deformation, and that $d_{CE}=s+\omega$.

The multiplet is conformal. It is straightforward to check that in all steps in our construction the $\ft$-module property is lifted to a superconformal module property, where the superconformal algebra 
is the contragredient extension of $\ft$ to non-positive levels,
$\osp(8|4)$.

\numpar[sec:D4N2][\texorpdfstring{Four-dimensional $\N=2$ supersymmetry}{Four-dimensional N2 supersymmetry}]
The multiplet obtained from a scalar pure spinor superfield with $\mathcal{N}=2$ supersymmetry in $D=4$ is not a super-Yang--Mills multiplet, 
but a version of the hypermultiplet where one scalar is dualised to a $2$-form. This is an off-shell multiplet. Fields, and generators at all levels in the algebra, transform under $\mathfrak{so}(4)\oplus\mathfrak{su}(2)\oplus\mathfrak{u}(1)
\simeq\mathfrak{su}(2)\oplus\mathfrak{su}(2)\oplus\mathfrak{su}(2)\oplus\mathfrak{u}(1)$.
We use $SL(2,\mathbb{C})$ notation for the first two $\mathfrak{su}(2)$'s.
The physical fields are a real $2$-form $B$ with fields strength $G$ in
${\bf(2,2,0)}_0$, a triplet of real scalars $\Phi$ in ${\bf(1,1,3)}_0$ (represented as an antihermitean $(2\times2)$-matrix) and a complex spinor $\Psi_{\alpha i},\overline{\Psi}_{\dot\alpha}{}^i$ in ${\bf(2,1,2)}_1\oplus{\bf(1,2,2)}_{-1}$. In addition, there is an auxiliary complex scalar field $H,\overline{ H}$ in 
${\bf(1,1,1)}_2\oplus{\bf(1,1,1)}_{-2}$. Subscripts denote $\mathfrak{u}(1)$ charge.
Elements in the coalgebra are denoted by the corresponding lowercase letters.

The dual to the supertranslation algebra is
\begin{align}
    \sm v^{\dot\alpha\alpha}&=\overline{\lambda}^{\dot\alpha}{}_i\lambda^{\alpha i}\;.
\end{align}

\begin{table}[H]
\[
    \begin{tikzcd}[row sep = 4 pt, column sep = 10 pt]
        {\bf(1,1,1)}_0 \ar[dr] &  \\
        \cdot & {\bf(2,2,1)}_0 \ar[dr]  \\
        \cdot & \cdot & \begin{matrix}[1]{\bf(3,1,1)}_0\\\oplus{\bf(1,3,1)}_0\\\oplus{\bf(1,1,3)}_0\end{matrix} &  
        \begin{matrix}[1]{\bf(2,1,2)}_1\\\oplus{\bf(1,2,2)}_{-1}\end{matrix}
        &\begin{matrix}[1]{\bf(1,1,1)}_2\\\oplus{\bf(1,1,1)}_{-2}\end{matrix}\\
\end{tikzcd}
\]
\normalsize
\caption{\it Zero-mode cohomology for the $D=4$, $\N=2$ dualised hypermultiplet.}
\end{table}

The supertranslation transformations of the cogenerators in the multiplet are encoded in
\begin{align}
    \sm\phi^i{}_j&=0\;,\nn\\
    \sm\psi^{\alpha i}&=\phi^i{}_j\lambda^{\alpha j}\;,\nn\\
    \sm\phi^{\dot\alpha\alpha,i}{}_j&=v^{\dot\alpha\alpha}\phi^i{}_j
        +\lambda^{\alpha i}\overline{\psi}^{\dot\alpha}{}_j-\psi^{\alpha i}\overline{\lambda}^{\dot\alpha}{}_j\;,\nn\\
    \sm g^{\dot\alpha\alpha}&=\overline{\lambda}^{\dot\alpha}{}_i\psi^{\alpha i}
        +\overline{\psi}^{\dot\alpha}{}_i\lambda^{\alpha i}\;,\\
    \sm h&=\epsilon_{\alpha\beta}\epsilon_{ij}\lambda^{\alpha i}\psi^{\beta j}\;,\nn\\
    \sm \psi^{\dot\alpha\alpha,\beta i}&=v^{\dot\alpha\alpha}\psi^{\beta i}
        +\phi^{\dot\alpha\alpha,i}{}_j\lambda^{\beta j}
        +\frac{1}{2}g^{\dot\alpha\alpha}\lambda^{\beta i}
        -g^{\dot\alpha\beta}\lambda^{\alpha i}
        +\epsilon^{\alpha\beta}\epsilon^{ij}\overline{\lambda}^{\dot\alpha}{}_jh\;.\nn\\
        &\ldots\nn
\end{align}
The action of the supersymmetric $3$-cocycle is
\begin{align}
    \omega_0\phi^i{}_j&=i\lambda^{\alpha i}v_{\alpha\dot\alpha}\overline{\lambda}^{\dot\alpha}{}_j
        \;,\nn\\
    \omega_0\psi^{\alpha i}&=i\lambda^{\beta i}v_{\beta\dot\alpha}v^{\dot\alpha\alpha}
        \;,\\
    \omega_0g^{\dot\alpha\alpha}
        &=-\frac{i}{2}v^{\dot\alpha\beta}v_{\beta\dot\beta}v^{\dot\beta\alpha}\;,\nn
\end{align}
where 
$v_{\alpha\dot\alpha}
=-\epsilon_{\alpha\beta}\epsilon_{\dot\alpha\dot\beta}v^{\dot\beta\beta}$. 

The explicit form of the outer derivation $\omega_1$ on $F[\overline{\fm}]$ can be found for the lowest internal degrees. 
The rest of the construction, in particular the argument leading to a conjecture about the absence of a $4$-bracket, goes along the exact same lines as in {\S}\ref{sec:D6N20}.

\numpar[sec:D11][Eleven-dimensional supersymmetry]
This example is fundamentally different from the non-gravitational examples, due to the presence of cohomology also in the ghost sector, represented by super-Killing vectors on Minkowski superspace.

The algebra $\fn$ is still freely generated by the dual of the multiplet.
In order to show this, we need to establish that the whole dual multiplet is a (minimal) generating set of $\fn$. 
Brackets must respect the double grading (see Table \ref{fig:elevenalgebra}). The only possible type of relation would be of the type
$[\overline{k},\overline{k}]=\overline{\phi}$.
This is not allowed by tensor products of the $\so(11)$ modules appearing.
(In addition, if such relations had been present, they would give rise to cohomology of $\fn$ at cohomological degree 4, which is outside the multiplet.) All cohomology is in generation 1, and $\fn$ is freely generated according to Proposition \ref{prop:nfree}. 

$S$ is a Gorenstein ring of odd codimension ($9$). Since 
$\vee^2(00001)=(10000)\oplus(01000)\oplus(00002)$, $Y$ is the closure of an intermediate (non-minimal) orbit. The tangent c\^one at the singular locus, the minimal orbit, is $Gr(2,5)$, 
so the Gorenstein property essentially follows from the Gorenstein property in {\S}\ref{sec:SL5}.
Some properties of $S$ have been examined in refs.
\cite{Berkovits2005TheSpinors,Cederwall2009TowardsSupergravity,Movshev:2011cy}, and the supersymmetric pure spinor BV field theory was constructed in refs. 
\cite{Cederwall2009TowardsSupergravity,Cederwall2009D=11Supersymmetry}.
In  \cite{Jonsson}, an investigation of the Tate resolution was initiated in increasing internal degree (up to degree $8$), leading to conjectures, some of which will be verified presently.

\begin{table}[h]
\[\small
    \begin{tikzcd}[row sep = 4 pt, column sep = 4 pt]
        \bf1 \ar[dr]  \\
        \cdot & \bf11 \ar[dr]  \\
        \cdot & \cdot & \bf55\oplus\bf11 
        \ar[dr,start anchor={[xshift=-5ex]},
        end anchor={[xshift=-0ex,yshift=-0ex]},bend right=15] 
        \ar[dr,start anchor={[yshift=1ex]},end anchor={[xshift=5ex]},bend left=5]
        & \bf32 \ar[dr]\\
        \cdot&\cdot&\cdot&\bf165\oplus\bf65\oplus\bf1
        \ar[drrr,controls={+(1,-2) and +(-1,2)}]
        &\bf320\oplus\bf32 \ar[dr,shorten=-1mm,start anchor={[xshift=-1ex]},
        end anchor={[xshift=1ex]}]\\
        \cdot&\cdot&\cdot&\cdot&\cdot&\bf32\oplus\bf320 \ar[dr]&\bf1\oplus\bf65\oplus\bf165
        \ar[dr,end anchor={[xshift=5ex]},
        start anchor={[xshift=-0ex,yshift=-0ex]},bend left=15]
        \ar[dr,end anchor={[yshift=-1ex]},start anchor={[xshift=-5ex]},bend right=5]
        \\
        \cdot&\cdot&\cdot&\cdot&\cdot&\cdot&\bf32&\bf11\oplus\bf55 \ar[dr]\\
        \cdot&\cdot&\cdot&\cdot&\cdot&\cdot&\cdot&\cdot&\bf11 \ar[dr]\\
        \cdot&\cdot&\cdot&\cdot&\cdot&\cdot&\cdot&\cdot&\cdot&\bf1
\end{tikzcd}
\]
\caption{\it Zero-mode cohomology for $D=11$ supergravity.}
\label{11zeromode}
\end{table}

\begin{table}[h]
    \centering
\begin{picture}(250,150)(150,-30)
\put(80,0){$0$}
\put(80,30){$1$}
\put(80,60){$2$}
\put(118,-30){$1$}
\put(156,-30){$2$}
\put(194,-30){$3$}
\put(232,-30){$4$}
\put(270,-30){$5$}
\put(308,-30){$6$}
\put(346,-30){$7$}
\put(384,-30){$8$}
\put(422,-30){$9$}
\put(455,-30){$\cdots$}
\put(80,90){$\vdots$}
\put(100,-15){\line(0,1){120}}
\put(100,-15){\line(1,0){381}}
\multiput(194,0)(38,0){7}{$\cdot$}
\multiput(118,30)(38,0){3}{$\cdot$}
\multiput(346,30)(38,0){3}{$\cdot$}
\multiput(118,60)(38,0){7}{$\cdot$}
\multiput(118,90)(38,0){9}{$\cdot$}
\put(115,0){$D_\alpha$}
\put(155,0){$D_a$}
\put(227,30){$\bar k^a$}
\put(265,30){$\bar k^\alpha$}
\put(303,30){$\bar k^{ab}$}
\put(379,60){$\bar\phi_{abcd}$}
\put(417,60){$\bar\phi_{ab}^\alpha$}
\put(455,60){$\cdots$}
\put(211,15){\line(0,1){90}}
\put(211,15){\line(1,0){270}}
\end{picture}
    \caption{\it Generators in the double grading of $\widetilde{\ft}$ for $D=11$ supergravity. Internal degree is on the horizontal axis and homological degree on the vertical. Only the generating set of $F[\overline{\fm}]$ is displayed. Note the symmetry around degree $(3,\frac12)$.}
    \label{fig:elevenalgebra}
\end{table}

Factoring out degree $1$ in the partition function for $S$ gives 
Table \ref{11zeromode}. Factoring out also degree 2 yields the linear multiplet:
\begin{align}
Z_S(t)&=(1-t)^{-(00001)}\otimes(1-t^2)^{(10000)}\nn\\
&\otimes
\Bigl[
(00000)\oplus(10000)t^4\ominus(00001)t^5\oplus(01000)t^6\\
&\qquad\ominus\bigoplus_{i=0}^\infty
\left((i0010)t^{8+2i}\ominus(i1001)t^{9+2i}\oplus(i2000)t^{10+2i}\right)
\Bigr]\;.\nn
\end{align}

On the algebra side, let $\overline{\fk}$ be the (dual) Killing supermultiplet at 
internal degrees $4,5,6$, and $\overline{\fp}$ the (dual) physical multiplet at degrees $\geq8$. We have $\overline{\fm}=\overline{\fk}\oplus\overline{\fp}$ as a vector space.  
The supertranslation algebra is encoded in the coalgebra differential
\begin{align}
s_0\lambda^\alpha&=0\;,\nn\\
s_0v^a&=(\lambda\gamma^a\lambda)\;.
\end{align}
We let $s_{0}$ be the coalgebra differential of $\ft\inplus(\overline{\fk}\oplus\overline{\fp})$. Thus,
\begin{align}
s_0k_a&=0\;,\nn\\
s_0k_\alpha&=(\gamma^a\lambda)_\alpha k_a\;,\\
s_0k_{ab}&=(\lambda\gamma_{ab}k)-2v_{[a}k_{b]}\;.\nn
\end{align}
and\footnote{The concrete expression for the formation of a $\gamma$-traceless $2$-form-spinor $\tilde X$ from the tensor product of a $2$-form and a spinor is
$\tilde X_{ab}=X_{ab}+\frac29\gamma_{[a}\gamma^cX_{b]c}-\frac1{90}\gamma_{ab}\gamma^{cd}X_{cd}$.}
\begin{align}
s_0\phi_{abcd}&=0\;,\nn\\
s_0\phi_{ab\alpha}&=\frac12(\gamma^{cd}\lambda)_\alpha\phi_{abcd}-(\gamma\hbox{-trace})\;,\nn\\
s_0\phi_{ab,cd}&=\frac12\bigl((\lambda\gamma_{ab}\phi_{cd})+(\lambda\gamma_{cd}\phi_{ab})\bigr)
-[abcd]\;,\\
s_0\phi_{a,bcde}&=v_a\phi_{bcde}+3(\lambda\gamma_{a[bc}\phi_{de]})-([abcde]+\hbox{trace})\;,\nn\\
&\cdots\nn
\end{align}
Here, $k_a$, $k_\alpha$ and $k_{ab}$ are Killing ghosts for translations, fermionic translations and rotations, while
$\phi_{abcd}$, $\phi_{ab\alpha}$ and $\phi_{ab,cd}$ are the (linearised) $4$-form, gravitino field strength and Weyl tensor of the physical supergravity multiplet.

When the Tate resolution is examined by internal degree, a minor surprise arises at degree $9$. There it turns out that the $2$-bracket needs to be modified. The multiplet $\overline{\fm}=\overline{\fk}\oplus\overline{\fp}$ is still freely generating $\fn$, but the action of the supertranslation generators on the multiplet is modified. Here, we do the construction by arity of brackets, so we start with this deformation. It turns out that there is a $2$-cocycle deformation of $\ft\inplus F[\overline{\fm}]$ as defined above. As a cocycle of the supertranslation algebra, it belongs to $H^1(\ft,\overline{\fp}\otimes\wedge^2\fk)$, and allows for a deformation of the $2$-bracket $[\ft,\overline{\phi}]$: $\ft\otimes\overline{\fp}\rightarrow \wedge^2 \overline{\fk}$ landing in generation $2$ of $F[\overline{\fk}]$.
Let us use $\nu$ for the corresponding term in the coalgebra differential, and denote the generation-$2$ basis elements as 
$k$ with double indices. 
Up to a scaling, the complete expression for $\nu$ is
\begin{align}
\nu k_{a,b}&=0\;,\nn\\
\nu k_{a,\alpha}&=\frac1{10}(\gamma^{bcd}\lambda)_\alpha\phi_{abcd}-\frac1{60}(\gamma_a{}^{bcde}\lambda)_\alpha\phi_{bcde}\;,\nn\\
\nu k_{\alpha,\beta}&=\frac1{16}\gamma^{abcd}_{\alpha\beta}(\lambda\gamma_{ab}\phi_{cd})
+\frac1{24}\gamma^{abc}_{\alpha\beta}v^d\phi_{abcd}\;,\nn\\
\nu k_{a,bc}&=\frac94(\lambda\gamma_a\phi_{bc})-\frac32(\lambda\gamma_{[b}\phi_{c]a})
        +\frac14v^d\phi_{abcd}\;,
\label{eq:nu2bracket}\\
\nu k_{ab,\alpha}&=-\frac{10}{39}
(\gamma^{cde}\lambda)_\alpha\phi_{[a,b]cde}
    -\frac5{156}(\gamma_{[a}{}^{cdef}\lambda)_\alpha\phi_{b],cdef}
    +\frac{21}{26}(\gamma^{cd}\lambda)_\alpha\phi_{ab,cd}\nn\\
&\qquad+\frac6{13}v^c(\gamma_c\phi_{ab})_\alpha-\frac6{13}v^c(\gamma_{[a}\phi_{b]c})_\alpha\;,\nn\\
\nu k_{ab,cd}&=6((\lambda\gamma_{[a}\phi_{b],cd})+(\lambda\gamma_{[c}\phi_{d],ab}))
+\frac{10}{13}v^e\phi_{e,abcd}
\;.\nn
\end{align}
In the last term, $\phi_{a,bc\alpha}$ is ``the derivative of the gravitino field strength'', the degree $11$ element 
with $s\phi_{a,bc\alpha}=v_a\phi_{bc\alpha}+\ldots$.
The action of $\nu$ on $F[\overline\fm]$ is defined by tensor product.

This means that $F[\overline{\fm}]$ now has become an indecomposible $\ft$-module. The differential we need to modify by higher brackets is
$s=s_0+\nu+s_1$ on $\ft\inplus F[\overline{\fm}]$. Here, $s_0+\nu$ represents the deformed action of $\ft$ on the module $F[\overline{\fm}]$. The normalisation of $\nu$ is irrelevant as long as it is non-zero, we choose the one of 
 \eqref{eq:nu2bracket}.
There is still a bicomplex with differentials $s_0+\nu$ and $s_1$, which can be used to find higher brackets, always with at most one entry in $\overline{\fm}$, by descent equations as in {\S}\ref{sec:doublecomplexomega}.

We continue by investigating the $3$-bracket.
The connecting $3$-cocycle $\omega_0$ is dual to a bracket
$[\ft,\ft,\ft]$: $\wedge^3\ft\rightarrow\overline{\fk}$.
It is straightforwardly constructed:
\begin{align}
\omega_0 k_a&=(\lambda\gamma_{ab}\lambda)v^b\;,\nn\\
\omega_0 k_\alpha&=\frac12(\gamma_{ab}\lambda)_\alpha v^av^b\;,\label{OmegaCocycleEq}\\
\omega_0 k_{ab}&=0\;,\nn
\end{align}
Then, $\omega_1$, dual to a $3$-bracket $[\ft,\ft,\overline{\fm}]$, is found through the descent equations. Due to the presence of $\nu$, it will consist of two parts, $\omega'$: $[\ft,\ft,\overline{\fk}]\in\overline{\fp}$ and 
$\omega''$: $[\ft,\ft,\overline{\fk}]\in\overline{\fk}^{(2)}$, where $\overline{\fk}^{(2)}$ means second generation, \ie, $[\overline{\fk},\overline{\fk}]$ in $F[\overline{\fm}]$.
They are given by the concrete expressions
\begin{align}
\omega'\phi_{abcd}&=2(\lambda\gamma_{[ab}\lambda)k_{cd]}
+\frac1{12}(\lambda\gamma_{abcd}{}^{ef}\lambda)k_{ef}
-\frac43v_{[a}(\lambda\gamma_{bcd]}k)
-\frac16v^e(\lambda\gamma_{abcde}k)\;,\nn\\
\omega'\phi_{ab\alpha}&=\frac12(\gamma^c\lambda)_\alpha v_c k_{ab}
-(\gamma^c\lambda)_\alpha v_{[a}k_{b]c}
+v_av_bk_\alpha -(\gamma\hbox{-trace})\;,\label{OmegaPrimeEq}\\
\omega'\phi_{ab,cd}&=-\frac12(v_av_bk_{cd}+v_cv_dk_{ab})-([abcd]+\hbox{trace})\;,\nn
\end{align}
and 
\begin{align}
\omega'' k_{a,b}&=2(\lambda_{(a}{}^c\lambda)k_{b)c}-\frac16\eta_{ab}(\lambda\gamma^{cd}\lambda)k_{cd}\nn\\
&\qquad-2v_{(a}(\lambda\gamma_{b)}k)+\frac13\eta_{ab}v^c(\lambda\gamma_ck)\;,\nn\\
\omega'' k_{a,\alpha}&=\frac4{15}\lambda_\alpha v^bk_{ab}+\frac1{10}(\gamma^{bc}\lambda)_\alpha v_ak_{bc}
+\frac7{15}(\gamma^{bc}\lambda)_\alpha v_bk_{ac}\\
&\qquad-\frac1{15}(\gamma_a{}^b\lambda)_\alpha v^ck_{bc}
+\frac1{60}(\gamma_a{}^{bcd}\lambda)_\alpha v_bk_{cd}\nn\\
&\qquad+\frac4{15} v_av_b(\gamma^b k)_\alpha-\frac1{30} v_bv_c(\gamma_a{}^{bc}k)_\alpha\;, \nn\\
&\cdots\nn
\end{align}
The relative coefficient between $\omega'$ and $\omega''$ is fixed by the normalisation of $\nu$.

This is the structure of the $3$-bracket, as a linear deformation. Since $D=11$ supergravity contains a physical $4$-form field strengths, one may expect this to be reflected in the presence of some $4$-bracket.
This indeed happens. We need to check the $5$-identities with two $3$-brackets. With the explicit forms above, the result is that there is a single non-vanishing $4$-bracket $[\ft,\ft,\ft,\ft]$, namely
$[D_a,D_b,D_c,D_d]=H_{abcd}$.
There will of course also be some $4$-brackets $[\ft,\ft,\ft,\overline{\fm}]$.
Note that we do not reproduce the standard closed superspace $4$-form of $D=11$ supergravity. Our 4-bracket does not arise as a linear deformation (a cocycle), but from a failure of the $3$-bracket to be a finite deformation.
In the superspace language of \S\ref{sec:superspace} and the end of 
{\S}\ref{sec:D6N20}, the Bianchi identity of our $4$-form will be sourced by the Killing multiplet as $dH=[B,G]$. 
We conjecture that higher brackets are absent through arguments analogous to the ones in {\S}\ref{sec:D6N20}, but have not been able to prove this.

We do not know how to interpret the non-linear structure. It can not be understood in terms of physical fields, since the dimensions are wrong. Remember that the pure spinor superfield $\Psi$ for $D=11$ supergravity carries ghost number $3$ and inverse length dimension $-3$. In terms of the algebra (see Table \ref{fig:elevenalgebra}), the physical multiplet would fit at negative internal degrees, while maintaining the symmetry around internal degree $3$, leading to a structure similar to a tensor hierarchy algebra. We comment on such possible prolongations in the Discussion, \S\ref{sec:Discussion}.

\newpage 
\subsection{\texorpdfstring{$E_6$ supersymmetry}{E6 supersymmetry}\label{sec:E6}}
\hfill
\nopagebreak        
\begin{figure}[h]
    \centering
\begin{picture}(210,70)(45,0)
\multiput(50,10)(40,0){6}{\circle{10}}
\put(170,50){\circle{10}}
\multiput(55,10)(40,0){5}{\line(1,0){30}}
\put(170,15){\line(0,1){30}}
\put(50,10){\line(1,1){3.5}}
\put(50,10){\line(1,-1){3.5}}
\put(50,10){\line(-1,1){3.5}}
\put(50,10){\line(-1,-1){3.5}}
\end{picture}
    \caption{\it Dynkin diagram for $\mathscr{B}(E_6,\Lambda_1)$}
    \label{fig:be6}
\end{figure}
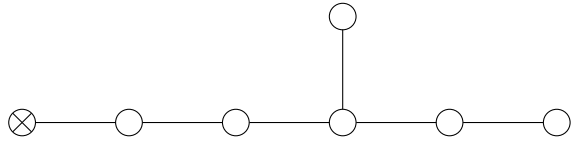

This example distinguishes itself from the lower-dimensional examples in the $E$ series ({\S}\ref{sec:SYM} and {\S}\ref{sec:SL5}) in that the BKM superalgebra $\mathscr{B}(E_6,\Lambda_1)$
is not freely generated from internal degree $3$.
This can be checked explicitly, the deviation starting at internal degree~9.

\begin{table}[H]
\[
    \begin{tikzcd}[row sep = 4 pt, column sep = 4 pt]
        \bf1 \ar[dr] &  \\
        \cdot & \overline{\bf27} \ar[drrr, bend left = 50] & \bf78 \ar[drrr, bend left = 45] \arrow[dr]{A} \\
        \cdot & \cdot & \cdot & \bf351  \ar[drrr, bend right = 50] &  \bf650 \ar[drrr, bend left = 60] \ar[dr] & \overline{\bf351}' \ar[dr]  \\
        \cdot & \cdot & \cdot & \cdot & \cdot & \bf351' \ar[drrr, bend right = 40] & \bf650 \ar[drrr, bend right = 50] & \overline{\bf351} \ar[dr]  \\
    \cdot & \cdot & \cdot & \cdot & \cdot & \cdot & \cdot & \cdot & \bf78 & \bf27 \ar[dr] \\
    \cdot & \cdot & \cdot & \cdot & \cdot & \cdot & \cdot & \cdot & \cdot & \cdot & \bf1
\end{tikzcd}
\]
\caption{\it Zero-mode cohomology for the $E_6$ model.\label{E6zeromodes}}
\end{table}

The corresponding minimal orbit is defined by the equations resulting from $\overline{27} \subset \Sym^2(27)$. $S$ is a Gorenstein ring of even codimension ($10$). The pure spinor variety is a c\^one over the Cayley plane $\mathbb{O} P^2$, which is a minimal $E_6$ orbit.
The zero-mode cohomology is depicted in Table \ref{E6zeromodes}.
Some Dynkin labels are
\[
    \overline{\bf27} =  
    \EWeight{00}0{01}0\;,\quad
    {\bf78} = 
    \EWeight{00}0{00}1\;,\quad
    {\bf351} = 
    \EWeight{01}0{00}0\;,  
\]
\[
    {\bf351}' = 
    \EWeight{20}0{00}0\;,\quad 
    {\bf650} = 
    \EWeight{10}0{01}0\;,\quad 
    {\bf2925}=\EWeight{00}1{00}0\;.
\]
The indicated differential operators should all be present.

The multiplet in this example is unconventional in that it lies in different cohomological degrees (ghost numbers). 
Recall that the multiplet consists of $H^\bullet(\fn)$, where $\fn={\mathscr B}(E_6,\Lambda_1)_{\geq3}$ is the subalgebra of 
${\mathscr B}(E_6,\Lambda_1)$ at internal degrees $\geq3$.
We have not solved it completely. 
It contains infinite-dimensional $\ft$-modules in $H^1$ and $H^2$, and probably nothing else.

\begin{table}[h]
    \centering
    \tiny
\begin{picture}(200,120)(150,-40)
\put(40,10){cohom.}
\put(40,0){degree $=$}
\put(80,0){$1$}
\put(80,45){$2$}
\put(50,-40){weight degree $=$}
\put(120,-40){$3$}
\put(160,-40){$4$}
\put(200,-40){$5$}
\put(240,-40){$6$}
\put(280,-40){$7$}
\put(320,-40){$8$}
\put(360,-40){$9$}
\put(400,-40){$10$}
\put(430,-40){$\cdots$}
\put(100,-25){\line(0,1){100}}
\put(460,-25){\line(0,1){100}}
\put(100,-25){\line(1,0){360}}
\put(100,17.5){\line(1,0){360}}
\put(100,75){\line(1,0){360}}
\put(110,0){$\EWeight{00}0{00}1$}
\put(150,0){$\EWeight{00}0{10}0$}
\put(190,0){$\EWeight{00}0{01}1$}
\put(190,-15){$\EWeight{00}0{01}0$}
\put(230,0){$\EWeight{00}0{11}0$}
\put(230,-15){$\EWeight{00}0{00}1$}
\put(270,0){$\EWeight{00}0{02}1$}
\put(270,-15){$\EWeight{00}0{10}0$}
\put(310,0){$\EWeight{00}0{12}0$}
\put(350,0){$\EWeight{00}0{03}1$}
\put(390,0){$\EWeight{00}0{13}0$}
\put(430,0){$\cdots$}
\put(350,45){$\EWeight{00}1{00}0$}
\put(390,60){$\EWeight{00}0{02}0$}
\put(390,45){$\EWeight{01}0{01}0$}
\put(390,30){$\EWeight{00}0{10}1$}
\put(430,45){$\cdots$}
\end{picture}
    \caption{\it Cohomology of $\fn$ for the $E_6$ model.}
    \label{table:E6ncoho}
\end{table}

We notice that the appearance of $H^2$ from weight degree $9$ is a direct consequence of a relation in the brackets between generators freely generated by elements dual to $H^1$. 
The conditions in Prop. \ref{prop:nfree} are not satisfied.
Interpreted physically, such cohomology implies the pathological situation that there are ``physical'' states with ghost number $-1$, if ghost number $0$ is assigned to $H^1$. 

Note that the multiplet is agnostic to which resolution we choose to form the algebra $\widetilde{\ft}$. Any resolution gives the same cohomology 
$H^\bullet(\fn)$. Presently we have chosen $\widetilde{\ft}={\mathscr B}(E_6,\Lambda_1)$. This gives the Tate resolution, which is the minimal free multiplicative resolution. Other, non-minimal, resolutions yield the same result, but will necessarily have a differential containing $1$-brackets.
Among them should be one based on $F[\overline{\fm}]$.  

The part of the multiplet at cohomological degree $1$ can be formulated quite explicitly in terms of
solutions to differential equations as follows.
Consider a ``gauge potential'' $A_M$ in $\overline{\bf27}$, and a ``spinor'' field $\Psi$ in 
$\bf{78}$. The gauge symmetry is $\delta_\Lambda A_M=\partial_M\Lambda$, so 
$F_{MN}=\partial_MA_N-\partial_NA_M$ in $\overline{\bf351}$ is invariant, and fulfills a Bianchi identity $\partial_{[M}F_{NP]}=0$ in $\bf2925$. In addition it obeys an equation of motion in $\bf650$. This implies that $\partial F\in(0|00011)\oplus(1|00000)$. Repeated use of the Bianchi identity and the equations of motion imply that 
\begin{align}
\partial^pF\in
\begin{cases}
\EWeight{00}0{1p}0\;,&p=0\,,\;p\geq2\;,\\[1em]
\EWeight{00}0{11}0\oplus\EWeight{00}0{00}1\;,&p=1\;..
\end{cases}
\end{align}
The ``spinor'' has equations of motion $\partial\Psi|_{\bf351}=0$ and $\partial^2\Psi|_{\overline{\bf351}'}=0$. Repeated use of these equations yields
\begin{align}
\partial^p\Psi\in
\begin{cases}
\EWeight{00}0{0p}1\;,&p=0\,,\;p\geq3\;,\\[1em]
\EWeight{00}0{01}1\oplus\EWeight{00}0{01}0\;,&p=1\;,\\[1em] 
\EWeight{00}0{02}1\oplus\EWeight{00}0{10}0\;,&p=2\;.
\end{cases}
\end{align}
This agrees with Table \ref{table:E6ncoho},
and we have checked that the pattern continues, with no
other modules in $H^1(\fn)$, at least up to weight degree $25$. An proof using partition functions should be straight-forward.
Note that the appearance of multiple derivatives (momenta) as $\EWeight{00}0{0p}0$ is consistent with a condition $p\times p=0$ where $\times$ denotes the Freudenthal product 
$\overline{\bf27}\times\overline{\bf27}\rightarrow {\bf27}$, forcing momenta to lie on a c\^one over the Cayley plane, thus providing a field theoretic realisation of the particle model of  \cite{Cederwall:1988hj}. The additional modules that we have found in $H^\bullet(\fn)$
are all at cohomological degree $2$. The leading ones are given in Table \ref{table:E6ncoho2}, meaning, for example, that since
$\EWeight{00}1{00}0$ appears at weight degree $9$,
the modules $\EWeight{00}1{0p}0$
are also present at weight degree $9+2p$ (at least for $p=0,1,2,3$) although not written out in the table.
\renewcommand\arraystretch{2}
\begin{table}[h]
\small
\begin{tabular}{|c|l|}
\hline
weight degree
&Leading $E_6$-module in $H^2(\fn)$\\
\hline
$\phantom{1}9$&$\EWeight{00}1{00}0$\\
$10$&$\EWeight{00}0{10}1\oplus\EWeight{00}0{02}0\oplus\EWeight{01}0{01}0$\\
$11$&$\EWeight{00}0{20}0\oplus\EWeight{00}0{01}1\oplus\EWeight{01}0{00}1\oplus\EWeight{10}0{02}0$\\
$12$&$\EWeight{00}0{11}0\oplus\EWeight{00}1{00}0\oplus\EWeight{01}0{10}0\oplus\EWeight{10}0{01}1$\\
$13$&$\EWeight{00}0{10}1\oplus\EWeight{10}0{11}0\oplus\EWeight{01}0{01}0\oplus\EWeight{10}1{00}0$\\
$14$&$\EWeight{00}0{20}0\oplus\EWeight{01}0{00}1\oplus\EWeight{10}0{02}0\oplus\EWeight{11}0{01}0\oplus\EWeight{10}0{10}1$\\
$15$&$\EWeight{01}0{10}0\oplus\EWeight{10}0{20}0\oplus\EWeight{10}0{01}1\oplus\EWeight{11}0{00}1\oplus\EWeight{20}0{02}0$\\
\hline
\end{tabular}
\vspace{.1cm}
\caption{\it Further cohomology of $\fn$ for the $E_6$ model.\label{table:E6ncoho2}}
\end{table}

The supertranslation algebra has the differential 
\begin{align}
    s_0v^M=c^{MNP}\lambda_N\lambda_P\;,
\end{align}
where $c$ is the symmetric invariant tensor of $E_6$. We can choose a normalisation where $c^{MPQ}c_{NPQ}=\delta^M_N$.
Part of the supertranslation transformations are encoded in
\begin{align}
s_0 f^{MN}=c^{PQ[M}\lambda_P\psi_Q{}^{N]}    
\end{align}
$\psi$ and $f$ being the coalgebra elements dual to $\Psi$ and $F$.
There is a cocycle $\eta\in H^2(\ft)$, represented by 
$\eta=P^{({\bf78})}{}_M{}^N{}_{,P}{}^Q\lambda_Qv^P$, $P^{({\bf78})}$ being the projector on $\bf78$ in $\bf27\otimes\overline{\bf27}$. Its closedness is trivial, since the symmetric product $\vee^3 \bf27$ does not contain $\bf78$. It is straightforward to extend $\eta$ to $\omega_0\in H^2(\ft,\overline{\fm}_1)$, where $\fm_1$ is the restriction of the multiplet to $H^1(\fn)$. 
Then, this supersymmetric cocycle takes the form
\begin{align}
\omega_0\psi_M{}^N&=P^{({\bf78})}{}_M{}^N{}_{,P}{}^Q\lambda_Qv^P\;,\nn\\
\omega_0f^{MN}&=av^Mv^N
\end{align}
for some coefficient $a$, which is uniquely determined by $\{s_0,\omega_0\}f^{MN}=0$, since 
$\lambda^2v$ is in $\vee^2\overline{\bf27}\otimes\bf27$, which contains a single $\bf351$.
We have not continued the explicit construction to involve also the higher cohomology.

\subsection{\texorpdfstring{$G_2\times SL(2)$ supersymmetry}{G2 x SL(2) supersymmetry}\label{sec:G2A1}}
\hfill
\nopagebreak

This is an example which turns out to go beyond our original assumption of a $2$-graded supertranslation algebra. However, unlike the example of \S\ref{sec:E6}, it is likely to have a physical interpretation related to supergravity, and the Koszul dual algebra is likely to be freely generated from internal level $5$.

Consider a construction as above with structure group $G_2\times SL(2)$, and $\ft_1$ and $\ft_2$ spanning the modules 
$({\bf7},{\bf2})=(10)(1)$ and $({\bf7},{\bf1})=(10)(0)$ respectively. 
The ``pure spinor'' $\lambda_a^i$ is constrained by
$\sigma^{abc}\epsilon_{ij}\lambda_b^i\lambda_c^j=0$. This does not define a minimal orbit, since 
$\vee^2(\bf{7},\bf{2})=(\bf7,\bf1)\oplus(\bf{14},\bf1)\oplus(\bf1,\bf3)\oplus(\bf{27},\bf3)$, and only the first one, representing superspace torsion, is demanded to vanish.

The partition function of $S$ is
\begin{align}
Z_S(t)&=(1-t)^{-(10)(1)}
    \bigl[(00)(0)-(10)(0)t^2+(00)(1)t^3\nn\\
&+\left((01)(0)+(10)(0)\right)t^4-\left((10)(1)+(00)(1)\right)t^5
+(00)(2)t^6\\
&-\left((20)(0)+(00)(0)\right)t^6+\left((01)(1)+(10)(1)\right)t^7\nn\\
&-\left((10)(2)+(00)(0)\right)t^8+(00)(3)t^9
\bigr]\;.\nn
\end{align}
When we construct $\widetilde{\ft}$ through a Tate resolution, an interesting pattern arises, which has no counterpart in the previous examples. At degree $3$, killing the cohomology 
$\lambda^a_i v_a$ (which is closed since 
$\sigma^{abc}\lambda_a^i\lambda_b^j\lambda_c^k=0$), there is a fermionic generator $D_i$. It is not part of a multiplet. This completes the generators at homological degree $0$. 
So far, we have the differential
\begin{align}
d_3=\sigma^{abc}\epsilon_{ij}\lambda_a^i\lambda_b^j\frac{\*}{\*v^c}
+\lambda_a^iv^a\frac{\*}{\*\xi^i}\;.
\end{align}
The next cohomology (of $d_3$) to be killed is at degree $5$, and comes through a $3$-bracket, so the generators carry homological degree $1$.
Its concrete form is
\begin{align}
\omega^i=\epsilon_{jk}\lambda_a^i\lambda^{aj}\xi^k
+\frac14\sigma^{abc}\lambda_a^iv^bv^c
\end{align}
(where we have used a normalisation where $\sigma^{124}=1$ and cyclic).
It forms the beginning of a finite-dimensional supermultiplet similar to the Killing supermultiplet of $D=11$ supergravity, and spans degrees $5$, $6$, $7$ and $8$. An infinite-dimensional supermultiplet at homological degree $2$ begins at internal degree $10$.

Whether one asks for it or not, one is lead to a supertranslation algebra spanning three levels, of the form
\begin{align}
[D^a_i,D^b_j]&=\sigma^{abc}\epsilon_{ij}D_c\;,\\
[D^a_i,D_b]&=\delta^a_b D_i\;,
\end{align}
and all other brackets $0$,
of which the multiplets will form modules
(but we may expect a deformation of the action of supertranslations as in {\S}\ref{sec:D11}).
The finite-dimensional module is indeed the coadjoint of this supertranslation algebra extended by $\mathfrak{g}_2\oplus\mathfrak{a}_1$ at degree $0$, which points towards supergravity. The physical fields in the infinite-dimensional module are probably off shell.

We will leave the details of this case for future examination. In particular, R.~Eager suggested to us that one might look for a connection to a twisting of type IIB supergravity; it would be interesting to explore this possibility or other potential  physical roles of this algebra further.

\vskip6pt
\section{Discussion\label{sec:Discussion}}
\vskip8pt
We have investigated Koszul duality in the context of (generalised) supersymmetry, defined as any superalgebra consistently residing in degrees $1$ and $2$. 
In particular, the relation between the coordinate rings of constrained (generalised) spinor spaces (``pure spinor spaces'') and the supermultiplet is established.
We rely on Koszul duality defined by the Tate resolution of the quadratic constraints. This definition agrees with the traditional notion of Koszul duality (``quadratic Koszul duality'') when the Koszul dual is (the universal enveloping algebra of) a Lie superalgebra, but leads to $L_\infty$ algebras without $1$-brackets in other cases.

The results are elucidating, in the sense that the r\^ole of the multiplet as part of the Koszul dual algebra becomes clear.
 On the other hand, we are left with some seemingly difficult problems, see \S\ref{sec:Properties}, of relating good mathematical properties of the ring to good physical properties of the multiplet, and thus of giving clear mathematical criteria for when a coordinate ring defines a physically acceptable model.

This conundrum is illustrated by the question:
How does one know, given a Dynkin dia\-gram and an integral dominant weight, if the BKM superalgebra is freely generated from some level?
Or even, in non-BKM settings, whether the Koszul dual of the coordinate ring of the closure of some non-minimal orbit has this property? 
Apart from the $E_6$ examples, there are numerous other BKM superalgebras that are not freely generated from internal degree $3$: 
${\mathscr B}(D_n,\Lambda_n)$ ($n\geq6$), 
${\mathscr B}(E_7,\Lambda_1)$,
${\mathscr B}(F_4,\Lambda_4)$, \ldots\
It should be noted however, that the subalgebras at degree $\geq3$ in these cases still contain freely generated algebras since, for example, $\mathscr{B}(E_5,\Lambda_1)_{\geq3}$ is a subalgebra of $\mathscr{B}(E_6,\Lambda_1)_{\geq3}$
\cite{Kleinschmidt:2013em}, and thus they do not 
provide 
counterexamples to the conjecture $C_{10}$  of  \cite{Avramov:1982} concerning the existence of non-abelian free subalgebras.
We have not been able to find an $E_6$-equivariant freely generated subalgebra in ${\mathscr B}(E_6,\Lambda_1)$, and strongly suspect that there is none.

One major motivation for the present investigation was the appearance of algebras encoding the interactions of $D=10$ SYM theory already in the algebraic structure of the non-interaction theory. As we have seen, this is generic for models describing gauge theory (with $1$-form connections). We have not come much further in suggesting similar interpretations when the models describe higher form gauge fields, although the deformation of the $2$-bracket in the algebra of $D=11$ supergravity looks interesting.

Many multiplets are not found as cohomology in a scalar pure spinor superfield, but rather in a field transforming in some $\fg$-module. The field is then subject to additional ``shift symmetry'' \cite{Cederwall:2011vy,Cederwall:2008vd,Cederwall:2008xu,Cederwall:2017cez,Cederwall:2020dui}. Equivalently, they belong to sections of sheaves over the pure spinor space \cite{Eager:2018dsx,perspectivesonpurespin} other than the structure sheaf of the scalar functions. We have not dealt with such multiplets here, but expect that Koszul duality carries over in the sense that the dual object becomes a module of the Koszul dual algebra $\widetilde{\ft}$ presently considered, and that the appropriately defined such duality can be interpreted as a character formula for such 
$\widetilde{\ft}$-modules, in the same way as the duality treated here can be seen as providing a denominator formula for the dual superalgebra.

The use of the complex $\widetilde{A}^\bullet$, with the differential $\widetilde{d}$, is attractive in that it, unlike the standard pure spinor superfield complex $A^\bullet(S)$, contains the full de Rham operator on superspace. One may then hope for a version of pure spinor superfield theory for supergravity which is more geometrical than the standard one.
The formalism could also open for new ways of constructing negative ghost number operators (such as the famous ``$b$-ghost'' \cite{Berkovits2005PureString}), with or without non-minimal variables \cite{Cederwall:2022qfn}.

Which r\^oles do extensions (prolongations) to negative internal 
degree play? For BKM superalgebras, there are typically different such extensions, contragredient or of tensor hierarchy algebra type \cite{Palmkvist:2013vya,Carbone:2018xqq,Cederwall:2019qnw}.
When a theory is conformal, it makes sense to extend to negative levels so that the superconformal generators appear at internal degrees $-2,-1,0,1,2$; in other cases it is difficult to tell which extension is more natural.
One interesting observation is that the multiplet $\fm$ may appear at negative levels, and with the correct dimension (proportional to internal degree), even if $\overline{\fm}$ appears at shifted internal degree. 
This applies \eg\ to $D=11$ supergravity, where there seems to be a symmetry in $\widetilde{\ft}$ under conjugation and reflection in internal degree $3$, and many other algebras display similar symmetries. 
This shift and symmetry is reflected in the fact that the pure spinor superfield has dimension $-3$.
It is fascinating that this type of behaviour persists beyond the BKM setting.
This would suggest the existence of a tensor hierarchy-like extension of the $L_\infty$ algebra to negative levels which is empty at degree $-1$, contains the physical $4$-form at degree $-2$, etc. Filtered deformations of such an algebra  along the lines of refs. \cite{Figueroa-OFarrill:2015rfh,Figueroa-OFarrill:2015tan} may provide a way towards pure spinor superfield formulation of supergravity in non-flat backgrounds.

There is a striking similarity of the algebraic structures underlying supersymmetry, studied in the present paper, and those appearing in extended geometry \cite{Cederwall:2017fjm,Cederwall:2018aab,Cederwall:2019bai,Cederwall:2021xqi,Cederwall:2023xbj}, tensor hierarchy algebras
\cite{Palmkvist:2013vya,Carbone:2018xqq,Cederwall:2019qnw,Cederwall:2021ymp,Cederwall:2022oyb,Cederwall:2023}.
Already in their original applications to gauged supergravity \cite{Greitz:2013pua,Howe:2015hpa},
the tensor hierarchy algebras seem to know about supersymmetry in the sense that the representation constraint on the embedding tensor predicted by the algebra
agrees with the one coming from supersymmetry \cite{deWit:2008ta}. The algebraic constraint is in turn a consequence of the Serre relation 
complementary to a minimal orbit in the related BKM superalgebra, and the symmetric section constraint in the associated extended geometry transform in the same representation as this Serre relation.
Is there a deeper connection?

\vskip18pt
\noindent
\underline{\it Acknowledgments:} Part of this work was done during the workshop on Higher Structures, Gravity and Fields at the Mainz Institute for Theoretical Physics
of the DFG Cluster of Excellence PRISMA${}^+$ (Project ID 39083149).
MC, JP and IAS would like to thank the institute for its hospitality. 
We are grateful to C. S\"amann and R. Eager for valuable input. IAS also owes special thanks to K.~Costello, C.~Elliott, F.~Hahner, J.~Huerta, S.~Noja, S.~Raghavendran, J.~Walcher, and B.~Williams for conversations and collaborations on related issues, and especially to F.H. and B.W. for their comments on the draft. SJ and MC would also like to thank Charles Young for valuable discussions. SJ is supported by the Leverhulme Trust, Research Project Grant number RPG-2021-092. 
The work of IAS was funded by the Deutsche Forschungsgemeinschaft (DFG, German Research Foundation) --- Projektnummer 517493862 (Homologische Algebra der Supersymmetrie: Lokalit\"at, Unitarit\"at, Dualit\"at), and by the Free State of Bavaria.


\bibliographystyle{utphysmod2}


\providecommand{\href}[2]{#2}\begingroup\raggedright\endgroup

\end{document}